\documentclass[11pt]{llncs}
\usepackage[margin=1in]{geometry}
\usepackage[T1]{fontenc}
\usepackage[english]{babel}
\usepackage{amsmath, amssymb, amsfonts}
\usepackage{bbm}
\usepackage{thmtools}
\usepackage{graphicx}
\usepackage{mathtools}
\usepackage{placeins}
\usepackage{subcaption}
\usepackage{booktabs}
\usepackage{enumitem}
\usepackage{wrapstuff}
\usepackage{pifont}
\usepackage{caption}
\usepackage[colorlinks=false,hidelinks]{hyperref}
\usepackage{cleveref}
\usepackage[authoryear]{natbib}

\usepackage{tikz}
\usepackage{comment}

\usetikzlibrary{arrows.meta, positioning, patterns}

\usepackage{fancyhdr}
\usepackage{lastpage}
\newcommand{\N}{\mathbb{N}}
\newcommand{\R}{\mathbb{R}}
\newcommand{\Q}{\mathbb{Q}}

\newcommand{\PO}{\mathrm{PO}}
\newcommand{\wPO}{\mathrm{wPO}}
\newcommand{\fPO}{\mathrm{fPO}}
\newcommand{\wfPO}{\mathrm{wfPO}}
\newcommand{\dom}{\mathcal{D}}

\newcommand{\Probs}{\mathcal{P}}
\newcommand{\p}{\mathbf{p}}
\newcommand{\Committees}{\mathcal{W}}

\newcommand{\fracU}{\mathcal{U}^{\text{frac}}}

\newcommand{\Weights}{\Delta}
\newcommand{\weight}{\omega}

\newcommand{\e}{\mathbf{e}}

\newenvironment{proofsketch}
  {\par\noindent\textit{Proof sketch.}\ \ignorespaces}
  {\nobreak\hfill$\diamond$\par\medskip}

\renewenvironment{proof}
  {\par\noindent\textit{Proof.}\ \ignorespaces}
  {\nobreak\hfill$\square$\par\medskip}

\spnewtheorem{observation}{Observation}{\bfseries}{\itshape}
\crefname{observation}{Observation}{Observations}
\Crefname{observation}{Observation}{Observations}

\begin{document}

\title{Fractional Pareto-Optimality in Multiwinner Voting}
\author{Patrick Becker \inst{1} \and
Niclas Boehmer \inst{2} \and
Fabian Frank \inst{1} \and
Lara Glessen \inst{2} \and
Yeeseok Oh \inst{3}
}
\institute{Technical University of Munich  \email{$\{$patrick.becker, fabian\_w.frank$\}$@tum.de} \and
            Hasso Plattner Institute (HPI), University of Potsdam \email{$\{$niclas.boehmer, lara.glessen$\}$@hpi.de} \and
            The University of Tokyo \email{oh@g.ecc.u-tokyo.ac.jp}
            }

\maketitle              

\begin{abstract}
    Efficiency in multiwinner voting is most naturally captured by Pareto-optimality ($\PO$), yet this notion is computationally and structurally difficult to handle. We therefore study fractional Pareto-optimality ($\fPO$), under which a committee may not be dominated even by a fractional committee, i.e., any convex combination of committees. $\fPO$ turns out to be a natural refinement of $\PO$ as it retains exactly those Pareto-optimal committees whose efficiency is robust under uniform cloning of candidates. Furthermore, $\fPO$ committees are guaranteed to exist and have strong structural properties. 
    We present a characterization of $\fPO$ in terms of weighted utilitarian welfare maximization, which yields a polynomial-time algorithm for verifying $\fPO$ and shows that the set of $\fPO$ committees satisfies committee monotonicity and is connected under single-candidate swaps.
    Analyzing \emph{welfarist rules} through the lens of $\fPO$, we further uncover an incompatibility between $\fPO$ and equality-oriented objectives.
    Most notably, we show that proportional approval voting (PAV) violates $\fPO$ in the approval setting.
    We close by pinpointing preference domains, including various one-dimensional ones, on which $\PO$ and $\fPO$ collapse into one notion.
\end{abstract}

\section{Introduction}

Multiwinner voting, the problem of selecting a committee, i.e., a fixed-size subset of candidates, based on voters' preferences, is currently one of the most intensively studied problems in social choice theory \citep{DBLP:journals/scw/ElkindFSS17,lackner2023multi,BrPe23a}. The dominant focus in this literature has been on ensuring and quantifying the fairness of computed committees. By contrast, the efficiency of committees has received comparatively little attention. The standard way to formalize efficiency in this setting is Pareto-optimality: a committee is Pareto-optimal if no other committee can make some voter strictly better off without making at least one other voter worse off \citep{Lackner_2020,laine2025pareto}. While Pareto-optimality might at first glance appear to be a fairly basic axiom, it is normatively very compelling: If a committee violates Pareto-optimality, it is hard to justify why the dominated committee, rather than the dominating one, should be implemented; no voter would object to such a change, while at least some would actively advocate for it. Yet, Pareto-optimality has proven to be a notoriously difficult axiom to work with, and the space of Pareto-optimal committees is poorly understood: Verifying Pareto-optimality of a given committee is coNP-hard under any of the standard preference models, and efficiently verifiable certificates of Pareto-optimality do not exist in general (under the assumption that coNP $\not\subseteq$ NP) \citep{Azi20}. At the same time, all known polynomial-time computable voting rules that satisfy even basic fairness axioms fail Pareto-optimality \citep{lackner2023multi,laine2025pareto}.\footnote{In fact, in concurrent work, \citet{dong2026jointly} showed that computing a Pareto-optimal committee satisfying the proportionality axiom of justified representation is NP-hard in the approval-based setting; hence, unless P\,=\,NP, no polynomial-time voting rule can guarantee both.}  This leaves us in an unsatisfactory position: violations of Pareto-optimality should undermine the legitimacy of outcomes. Yet, in light of the wealth of popular rules violating Pareto-optimality, they seem implicitly tolerated, presumably because neither detecting a violation nor certifying its absence can be done efficiently.

In this paper, we argue that this difficulty is not entirely inherent to efficiency as such, but rather to the classical, \emph{integral} notion of Pareto-optimality.
We put forward a systematic study of a stronger and better-behaved efficiency requirement in multiwinner voting, namely fractional Pareto-optimality ($\fPO$). 
Informally, a committee is fractionally Pareto-optimal if no fractional committee, i.e., no convex combination of feasible committees, Pareto-dominates it.
We study $\fPO$ under additive utilities, where each voter assigns a value to each candidate, and a voter's utility for a committee is the sum of her values for the selected candidates; this generalizes the standard approval-based preference model. Although $\PO$ and $\fPO$ may appear closely related, we show that the gap between them can be substantial. First, in utility terms, there are instances in which a Pareto-optimal committee is dominated by a fractional committee that gives every voter almost $k$ times her utility, where $k$ is the committee size. Second, in committee space, there are instances containing a Pareto-optimal committee whose closest $\fPO$ committee is at swap distance $k-1$. 

Central to our normative analysis of $\fPO$ is a classical result from convex optimization: fractionally Pareto-optimal committees are exactly those that maximize weighted utilitarian welfare under some strictly positive voter weights \citep{geoffrion1968proper,isermann1974proper}. 
This connection yields an efficient algorithm for checking whether a given committee satisfies $\fPO$, together with a succinct certificate: the weight vector under which the committee is welfare-maximizing. Beyond computational tractability, this characterization also makes $\fPO$ a transparent efficiency notion, with the certifying weight vector offering insight into the implicit distribution of influence across voters \citep{DBLP:journals/corr/abs-2604-24307}. The characterization further allows us to show that for any given instance, the set of $\fPO$ committees satisfies committee monotonicity and is connected under single-candidate swaps.
By contrast, we show that standard Pareto-optimality lacks both of these structural properties, even in the restricted setting where voters provide approval ballots.
The latter answers an open problem suggested by \citet{joshua-thesis}.
\Cref{table:Results} gives an overview of the comparison of the two efficiency notions. Note that all results hold in both settings: since approval-based instances are a special case of additive ones, the negative and hardness results for $\PO$ are obtained on approval instances, while the positive and tractability results for $\fPO$ are proven for additive utilities.

A natural objection to $\fPO$ is that it measures efficiency against fractional committees, which are rarely implementable, and thus, that $\fPO$ holds outcomes to an unrealistically high standard. We advocate for a different view: $\fPO$ should not be read as a literal benchmark comparing against (unimplementable) fractional committees, but as a well-behaved, transparent, and \emph{robust} refinement of $\PO$. In particular, a committee that satisfies $\PO$, but not $\fPO$, owes its efficiency to the discreteness of the current candidate set and is not robust to uniform subdivisions of that set. 
We make this formal by showing in \Cref{thm:cloningInvariance} that $\fPO$ is exactly the cloning-robust version of $\PO$: a committee is $\fPO$ if and only if, in every instance obtained by uniformly replicating each candidate a fixed number of times, the respective cloned committee remains Pareto-optimal. 
As an example, consider a public funding body selecting a fixed-size portfolio from a coarse shortlist of project proposals. A portfolio may be $\PO$ for the submitted shortlist, but become Pareto-dominated once the shortlist is refined, for example, by splitting projects into equal parts that can be funded separately. 
This is especially relevant in settings where proposers submit projects strategically, as recently surfaced and studied in the context of participatory budgeting \citep{faliszewski2025projectsubmissiongamesparticipatory}. In such contexts, the submitted candidate set cannot be viewed as fixed; $\fPO$ rules out precisely those portfolios whose efficiency is an artifact of a coarse, potentially strategically shaped proposal set.

\begin{table}[!t]
\centering
\caption{Axiomatic comparison of $\PO$ and $\fPO$ under additive and approval preferences. The result marked with $\dagger$ was shown by \citet{Azi20}.}
\label{table:Results}
\setlength{\tabcolsep}{12pt} 
\begin{tabular}{lcc}
\toprule

\textbf{Property} & \textbf{$\PO$} & \textbf{$\fPO$} \\
\midrule
    Verification of a given committee $W$
            & coNP-hard$^\dagger$
            & P
    \\
    Connectedness
            &  \ding{55} 
            & \ding{51}
    \\
    Committee monotonicity
            & \ding{55} 
            & \ding{51} 
    \\
    Cloning-robustness 
            & \ding{55} 
            & \ding{51} 
    \\
\bottomrule
\end{tabular}
\end{table}

Further, we demonstrate that $\fPO$ also offers a new perspective on the analysis of voting rules.
Previous work has shown that popular sequential multiwinner rules violate $\PO$ (and thereby $\fPO$), whereas rules that maximize a global objective typically satisfy $\PO$.
The class of welfarist rules, which select committees that maximize a scoring function of voters’ utilities, encompasses most of the known Pareto-optimal rules.
For $\fPO$, we obtain a nuanced picture of this class: \emph{strictly monotone, convex} scoring functions guarantee $\fPO$, while natural equality-oriented objectives may instead select fractionally Pareto-dominated committees. 
For Thiele rules, a particularly well-studied subclass of welfarist rules, we obtain a sharper picture that also holds under approval utilities: a Thiele rule violates $\fPO$ whenever the marginal gain a voter derives from an additional approved committee member decreases at some point. This, in particular, implies that proportional approval voting (PAV)
violates $\fPO$.
Overall, this signals that $\fPO$ is in tension with inequality-averse aggregation, which underpins many proportional voting rules.

Finally, building on the work of \citet{joshua-thesis}, we identify conditions under which the notions $\fPO$ and $\PO$ coincide for the approval setting. 
We show that the \emph{linearly consistent domain} by \citet{pier2022corestable}, which contains the voter-interval, candidate-interval, and voter-candidate-interval domains, satisfies these conditions, implying that $\fPO$ and $\PO$ are equivalent on instances from any of these domains. 
We give an overview of our results in \Cref{table:Results} and \Cref{fig:characterization-overview}. Omitted proofs can be found in the appendix.

\smallskip
\noindent\textbf{Related Work.}
Pareto-optimality in multiwinner voting has previously been studied along three
lines: the complexity of verifying it~\citep{Azi20}, which rules satisfy it~\citep{PeSk20a,lackner2023multi,laine2025pareto},
and its compatibility with other axioms \citep{Caragiannisetal2024,Lackner_2020,skowron2021proportionality}. Other formulations of efficiency we are aware of are either substantially weaker~\citep{Caragiannisetal2024,lackner2018consistent} or quantitative~\citep{Lackner_2020,skowron2021proportionality}.
In multiwinner voting, fractional Pareto-optimality has so far only appeared implicitly, as ex-ante efficiency in best-of-both-worlds approaches~\citep{DBLP:journals/corr/abs-2303-03642,suzuki2024maximum}, whereas it is well established in fair division of indivisible goods~\citep{DBLP:conf/aaai/KawaseM26,DBLP:journals/ior/SandomirskiyS22,DBLP:conf/ijcai/GargMQ23,DBLP:conf/sigecom/BarmanKV18,DBLP:journals/jair/GargM24,DBLP:journals/tcs/GargM23}. 
In work concurrent to ours, \citet{jin2026limitations} showed that $\wfPO$ is incompatible with justified representation in the approval setting, which complements our negative results in \Cref{sec:rules}.
Several of our results for additive utilities build on classical results from convex and multiobjective optimization: 
$\fPO$ committees are precisely the integral Pareto-efficient solutions of a multiobjective linear program over the fractional committee polytope~\citep{geoffrion1968proper,isermann1974proper,ehrgott2005multicriteria}, i.e., of the linear relaxation of a multiobjective knapsack problem with unit weights and a cardinality constraint~\citep{ehrgott1997connectedness,gorski2011connectedness}. We refer to \Cref{sec:RW} for a detailed discussion, including the connections to multiobjective combinatorial optimization.

\section{Preliminaries}

\noindent\textbf{General notation.} 
For $n \in \N$, we write $[n] \coloneqq \{1,2,\dots,n\}$, and we let $[0,1]$ denote the unit interval. 
Let $\e_i \in \R^n$ denote the $i$-th standard basis vector and $\mathbf 1_n$ the all ones vector of dimension $n$.
We define the \emph{standard unit simplex} as
$\Weights \coloneqq \big\{ \weight \in \R^n \;\colon\; \|\weight\|_1 = 1$ and $\weight_i \ge 0 \text{ for all } i \in [n] \big\}$,
and the \emph{positive unit simplex} $\Weights^+$ as the subset of $\Weights$ consisting of vectors with strictly positive entries, that is, $\weight_i > 0$ for all $i \in [n]$. 

A function $g \;\colon\; \R^n \to \R$ is \emph{monotone} if for all $x\neq y \in \R^n$ with $x_i \leq y_i$ for all $i\in [n]$, it holds that $g(x) \leq g(y)$.
It is \emph{strictly monotone} if $g(x) < g(y)$ holds instead. 

\smallskip
\noindent\textbf{Multiwinner elections with additive utilities.} Let the set of \emph{candidates} be $C=\{c_1,\dots,c_m\}$ and the set of \emph{voters} be $N=\{1,\dots,n\}$. We refer to size-$k$ subsets of $C$ as \emph{committees}, and denote the set of all size-$k$ committees by $\Committees_k$.
We refer to $k$ as the size of the committee with $k \leq m$ and write $\Committees$ when $k$ is clear from context.
For a committee $W\in \Committees$, let $\chi^W\in\{0,1\}^m$ be its incidence vector, where $\chi^W_c=1$ if $c\in W$, and $\chi^W_c=0$ otherwise. 

Each voter $i\in N$ reports a non-negative valuation function $v_i \colon C\to \Q_{\ge 0}$ mapping each candidate to a non-negative rational number; for a candidate $c\in C$, we write $v_{ic} \coloneqq v_i(c)$.
Voters' utilities are additive over committees: for each voter $i\in N$ and committee $W\in \Committees$, $u_i(W) \coloneqq \sum_{c\in W} v_{ic}$. We write $u(W)=(u_1(W),\dots,u_n(W))\in\Q_{\ge 0}^n$ for the \emph{utility vector} induced by $W$.
Finally, we write $V=(v_{ic})_{i\in N,c\in C}$ for the \emph{valuation profile}, and define an instance $I = (V, k)$ of a multiwinner election by the valuation profile and the committee size $k$. We refer to this setting as the \emph{additive-utility setting}.

\smallskip
\noindent\textbf{Approval-based multiwinner elections.}
We also consider the popular special case in which voters' valuation functions are binary, known as the \emph{approval-based setting}. 
Here, each voter $i$ reports an approval set $A_i\subseteq C$, which induces binary valuations by setting $v_{ic}=1$ if $c\in A_i$, and $v_{ic}=0$ otherwise.
To distinguish this case, we denote approval-based instances as $I=(A,k)$, where $A=(A_i)_{i\in N}$. We denote by $N_c = \{i \in N \colon c \in A_i\}$ the support of candidate~$c$.

\smallskip
\noindent\textbf{Efficiency Axioms.}
The standard efficiency notion for committees in the multiwinner voting literature is Pareto-optimality \citep{lackner2023multi,PeSk20a,laine2025pareto,Azi20}, which requires that no other committee can increase one voter's utility without decreasing another's.
Formally, a committee $W \in \Committees$ is \emph{Pareto-optimal} if there is no committee $T \in \Committees$ such that $u_i(T) \geq u_i(W)$ for all $i \in N$ and $u_j(T) > u_j(W)$ for some $j \in N$.
Given an instance $I$, we write $\PO(I)$ for the set of all Pareto-optimal committees.
We further say that a committee $W$ satisfies \emph{weak Pareto-optimality} if no committee $T$ gives strictly higher utility $u_i(T) > u_i(W)$ to every voter $i \in N$, and we denote the set of all weakly Pareto-optimal committees by $\wPO(I)$.

We study two notions of efficiency based on a fractional relaxation of the committee set $\Committees$.
Let $\Probs_k \coloneqq \{\p \in [0,1]^m \colon \sum_{c \in C} \p(c) = k\}$ denote the set of all fractional committees:
A \emph{fractional committee} $\p \in \Probs_k$ assigns each candidate $c$ a value $\p(c)$ between $0$ and $1$, summing up to $k$; again we write $\Probs$ if $k$ is clear from context.
By contrast, we occasionally refer to committees from $\Committees$ as \emph{integral committees} and identify each such committee with its incidence vector $\chi$. Under this identification, the integral committees are precisely the vertices of the polytope $\Probs$.
The utility that a voter $i$ derives from a fractional committee $\p$ is $u_i(\p) = \sum_{c \in C} v_{ic}\p(c)$, i.e., the weighted sum of the candidate values, with weights given by $\p$.
We denote the set of \emph{utility vectors} achievable by a fractional committee by $\fracU \coloneqq \{u(\p) \colon\; \p \in \Probs\}$.\footnote{This set is a polytope and every vertex of $\fracU$ is of the form $u(W) \in \Q_{\geq 0}^n$ for some integral committee $W\in\Committees$, since every vertex of a linear image of a polytope is the image of a vertex~\citep{ziegler1995lectures}.}

\begin{definition}[Fractional Pareto-optimality]
    A committee $W\in \Committees$ is \emph{fractional Pareto-optimal} $(\fPO)$ if there does not exist a fractional committee $\p \in \Probs$ such that $u_i(\p) \geq u_i(W)$ for all $i \in N$, and there exists at least one voter $j \in N$ with $u_j(\p) > u_j(W)$.
\end{definition}

\begin{definition}[Weak fractional Pareto-optimality]
    A committee $W$ is \emph{weak fractional Pareto-optimal} $(\wfPO)$ if there does not exist a fractional committee $\p \in \Probs$ such that $u_i(\p) > u_i(W)$ for all $i \in N$.
\end{definition}

Clearly, a committee that satisfies fractional Pareto-optimality also satisfies its weak version, while the converse is not necessarily true.
Given an instance $I$, we write $\fPO(I)$ for the set of all fractional Pareto-optimal committees and $\wfPO(I)$ for the set of all weakly fractional Pareto-optimal committees.

\smallskip
\noindent\textbf{Voting rules.}
A (multiwinner) voting rule is a function $f$ that maps every instance $I=(V,k)$ to a non-empty set $f(I)\subseteq\mathcal{W}_k$ of committees tied for the win.
We say that a voting rule $f$ satisfies a property $X$ if, for every instance $I$ and every committee $W \in f(I)$, $W$ satisfies $X$.

\section{General Relationship between $\PO$ and $\fPO$}
We start by establishing the relationships among the introduced efficiency axioms.

\begin{figure}[t]
\begin{minipage}[b]{0.45\textwidth}

    \centering
    \scalebox{0.8}{
        \begin{tikzpicture}[yscale=0.6,xscale=0.6,voter/.style={anchor=south}]
      \draw[fill=blue!40] (0, 4) rectangle (1, 5);
      \draw[fill=blue!40] (1, 4) rectangle (2, 5);
      \draw[fill=blue!40] (2, 4) rectangle (3, 5);
      \draw[fill=gray!10] (3, 4) rectangle (4, 5);
      \draw[fill=gray!10] (4, 4) rectangle (5, 5);
      \draw[fill=gray!10] (5, 4) rectangle (6, 5);
      \draw[fill=gray!10] (6, 4) rectangle (7, 5);
    
      \draw[fill=gray!10] (0, 3) rectangle (1, 4);
      \draw[fill=gray!10] (1, 3) rectangle (2, 4);
      \draw[fill=gray!10] (2, 3) rectangle (3, 4);
      \draw[fill=blue!40] (3, 3) rectangle (4, 4);
      \draw[fill=blue!40] (4, 3) rectangle (5, 4);
      \draw[fill=blue!40] (5, 3) rectangle (6, 4);
      \draw[fill=blue!40] (6, 3) rectangle (7, 4);
    
      \draw[fill=blue!40] (0, 2) rectangle (1, 3);
      \draw[fill=blue!40] (1, 2) rectangle (2, 3);
      \draw[fill=gray!10] (2, 2) rectangle (3, 3);
      \draw[fill=gray!10] (3, 2) rectangle (4, 3);
      \draw[fill=gray!10] (4, 2) rectangle (5, 3);
      \draw[fill=blue!40] (5, 2) rectangle (6, 3);
      \draw[fill=blue!40] (6, 2) rectangle (7, 3);
    
      \draw[fill=blue!40] (0, 1) rectangle (1, 2);
      \draw[fill=gray!10] (1, 1) rectangle (2, 2);
      \draw[fill=blue!40] (2, 1) rectangle (3, 2);
      \draw[fill=gray!10] (3, 1) rectangle (4, 2);
      \draw[fill=blue!40] (4, 1) rectangle (5, 2);
      \draw[fill=gray!10] (5, 1) rectangle (6, 2);
      \draw[fill=gray!10] (6, 1) rectangle (7, 2);
    
      \draw[fill=gray!10] (0, 0) rectangle (1, 1);
      \draw[fill=blue!40] (1, 0) rectangle (2, 1);
      \draw[fill=blue!40] (2, 0) rectangle (3, 1);
      \draw[fill=blue!40] (3, 0) rectangle (4, 1);
      \draw[fill=gray!10] (4, 0) rectangle (5, 1);
      \draw[fill=gray!10] (5, 0) rectangle (6, 1);
      \draw[fill=blue!40] (6, 0) rectangle (7, 1);
    
      \node[voter] at (0.5, -1) {$1$};
      \node[voter] at (1.5, -1) {$2$};
      \node[voter] at (2.5, -1) {$3$};
      \node[voter] at (3.5, -1) {$4$};
      \node[voter] at (4.5, -1) {$5$};
      \node[voter] at (5.5, -1) {$6$};
      \node[voter] at (6.5, -1) {$7$};
    
      \node[anchor=east] at (-0.3, 4.5) {$c_{1}$};
      \node[anchor=east] at (-0.3, 3.5) {$c_{2}$};
      \node[anchor=east] at (-0.3, 2.5) {$c_{3}$};
      \node[anchor=east] at (-0.3, 1.5) {$c_{4}$};
      \node[anchor=east] at (-0.3, 0.5) {$c_{5}$};
    \end{tikzpicture}
    }
    \subcaption{Approval-based instance $I=(A,k)$ with $k=2$; blue cells indicate approvals. The committee $W=\{c_1,c_2\}$ satisfies $\PO$, but the uniform fractional committee with weight $1/2$ over $\{c_2,c_3,c_4,c_5\}$ witnesses an $\fPO$ violation.}
    \label{figure:counter_example_PO_fPO_small}
\end{minipage}\hfill
\begin{minipage}[b]{0.52\textwidth}

    \centering
    \scalebox{0.8}{
    \begin{tikzpicture}[yscale=0.6,xscale=0.6,voter/.style={anchor=south}]
  \draw[fill=blue!40] (0, 6) rectangle (1, 7);
  \draw[fill=blue!40] (1, 6) rectangle (2, 7);
  \draw[fill=blue!40] (2, 6) rectangle (3, 7);
  \draw[fill=blue!40] (3, 6) rectangle (4, 7);
  \draw[fill=blue!40] (4, 6) rectangle (5, 7);
  \draw[fill=gray!10] (5, 6) rectangle (6, 7);
  \draw[fill=gray!10] (6, 6) rectangle (7, 7);
  \draw[fill=gray!10] (7, 6) rectangle (8, 7);
  \draw[fill=gray!10] (8, 6) rectangle (9, 7);
  \draw[fill=gray!10] (9, 6) rectangle (10, 7);

  \draw[fill=gray!10] (0, 5) rectangle (1, 6);
  \draw[fill=gray!10] (1, 5) rectangle (2, 6);
  \draw[fill=gray!10] (2, 5) rectangle (3, 6);
  \draw[fill=gray!10] (3, 5) rectangle (4, 6);
  \draw[fill=gray!10] (4, 5) rectangle (5, 6);
  \draw[fill=blue!40] (5, 5) rectangle (6, 6);
  \draw[fill=blue!40] (6, 5) rectangle (7, 6);
  \draw[fill=blue!40] (7, 5) rectangle (8, 6);
  \draw[fill=blue!40] (8, 5) rectangle (9, 6);
  \draw[fill=blue!40] (9, 5) rectangle (10, 6);

  \draw[fill=blue!40] (0, 4) rectangle (1, 5);
  \draw[fill=blue!40] (1, 4) rectangle (2, 5);
  \draw[fill=blue!40] (2, 4) rectangle (3, 5);
  \draw[fill=gray!10] (3, 4) rectangle (4, 5);
  \draw[fill=gray!10] (4, 4) rectangle (5, 5);
  \draw[fill=gray!10] (5, 4) rectangle (6, 5);
  \draw[fill=gray!10] (6, 4) rectangle (7, 5);
  \draw[fill=blue!40] (7, 4) rectangle (8, 5);
  \draw[fill=blue!40] (8, 4) rectangle (9, 5);
  \draw[fill=blue!40] (9, 4) rectangle (10, 5);

  \draw[fill=blue!40] (0, 3) rectangle (1, 4);
  \draw[fill=gray!10] (1, 3) rectangle (2, 4);
  \draw[fill=blue!40] (2, 3) rectangle (3, 4);
  \draw[fill=gray!10] (3, 3) rectangle (4, 4);
  \draw[fill=blue!40] (4, 3) rectangle (5, 4);
  \draw[fill=blue!40] (5, 3) rectangle (6, 4);
  \draw[fill=blue!40] (6, 3) rectangle (7, 4);
  \draw[fill=gray!10] (7, 3) rectangle (8, 4);
  \draw[fill=gray!10] (8, 3) rectangle (9, 4);
  \draw[fill=blue!40] (9, 3) rectangle (10, 4);

  \draw[fill=blue!40] (0, 2) rectangle (1, 3);
  \draw[fill=blue!40] (1, 2) rectangle (2, 3);
  \draw[fill=gray!10] (2, 2) rectangle (3, 3);
  \draw[fill=blue!40] (3, 2) rectangle (4, 3);
  \draw[fill=gray!10] (4, 2) rectangle (5, 3);
  \draw[fill=blue!40] (5, 2) rectangle (6, 3);
  \draw[fill=blue!40] (6, 2) rectangle (7, 3);
  \draw[fill=blue!40] (7, 2) rectangle (8, 3);
  \draw[fill=gray!10] (8, 2) rectangle (9, 3);
  \draw[fill=gray!10] (9, 2) rectangle (10, 3);

  \draw[fill=gray!10] (0, 1) rectangle (1, 2);
  \draw[fill=blue!40] (1, 1) rectangle (2, 2);
  \draw[fill=gray!10] (2, 1) rectangle (3, 2);
  \draw[fill=blue!40] (3, 1) rectangle (4, 2);
  \draw[fill=blue!40] (4, 1) rectangle (5, 2);
  \draw[fill=gray!10] (5, 1) rectangle (6, 2);
  \draw[fill=blue!40] (6, 1) rectangle (7, 2);
  \draw[fill=gray!10] (7, 1) rectangle (8, 2);
  \draw[fill=blue!40] (8, 1) rectangle (9, 2);
  \draw[fill=blue!40] (9, 1) rectangle (10, 2);

  \draw[fill=gray!10] (0, 0) rectangle (1, 1);
  \draw[fill=gray!10] (1, 0) rectangle (2, 1);
  \draw[fill=blue!40] (2, 0) rectangle (3, 1);
  \draw[fill=blue!40] (3, 0) rectangle (4, 1);
  \draw[fill=blue!40] (4, 0) rectangle (5, 1);
  \draw[fill=blue!40] (5, 0) rectangle (6, 1);
  \draw[fill=gray!10] (6, 0) rectangle (7, 1);
  \draw[fill=blue!40] (7, 0) rectangle (8, 1);
  \draw[fill=blue!40] (8, 0) rectangle (9, 1);
  \draw[fill=gray!10] (9, 0) rectangle (10, 1);

  \node[voter] at (0.5, -1) {$1$};
  \node[voter] at (1.5, -1) {$2$};
  \node[voter] at (2.5, -1) {$3$};
  \node[voter] at (3.5, -1) {$4$};
  \node[voter] at (4.5, -1) {$5$};
  \node[voter] at (5.5, -1) {$6$};
  \node[voter] at (6.5, -1) {$7$};
  \node[voter] at (7.5, -1) {$8$};
  \node[voter] at (8.5, -1) {$9$};
  \node[voter] at (9.5, -1) {$10$};

  \node[anchor=east] at (-0.3, 6.5) {$c_{1}$};
  \node[anchor=east] at (-0.3, 5.5) {$c_{2}$};
  \node[anchor=east] at (-0.3, 4.5) {$c_{3}$};
  \node[anchor=east] at (-0.3, 3.5) {$c_{4}$};
  \node[anchor=east] at (-0.3, 2.5) {$c_{5}$};
  \node[anchor=east] at (-0.3, 1.5) {$c_{6}$};
  \node[anchor=east] at (-0.3, 0.5) {$c_{7}$};
\end{tikzpicture}
    }
    \subcaption{Approval-based instance $I=(A,k)$ with $k=2$; blue cells indicate approvals. The committee $W=\{c_1,c_2\}$ satisfies $\PO$, but the uniform fractional committee with weight $2/5$ over $\{c_3,c_4,c_5,c_6,c_7\}$ gives every voter strictly more utility.}
    \label{figure:counter_example_weak_fPO}
\end{minipage}
  \caption{Counterexamples illustrating the relationships between Pareto-optimality and (weak) fractional Pareto-optimality from \Cref{proposition:axiomatic_relationships}; blue cells indicate approvals.}
  \label{fig:counterexamples_combined}
\end{figure}

\begin{restatable}{proposition}{fPOimpliesPO}\label{proposition:axiomatic_relationships}
    The following relationships hold: \emph{(i)} $\fPO$ is a strict refinement of $\PO$, \emph{(ii)} $\wfPO$ and $\PO$ are incomparable, and \emph{(iii)} $\wfPO$ is a strict refinement of $\wPO$.
\end{restatable}

We provide some intuition for the relationship between $\PO$ and $\mathrm{(w)fPO}$ here.
Consider \Cref{figure:counter_example_PO_fPO_small}, which illustrates an instance where the committee $W = \{c_1, c_2\}$, for which $u_i(W) = 1$ for all $i \in N$, satisfies Pareto-optimality but not fractional Pareto-optimality.
The fractional committee $\p = (0, \frac{1}{2}, \dots, \frac{1}{2})$ dominates $W$: no voter is worse off under $\p$, while the last voter is strictly better off. This follows from the fact that every voter approves at least two candidates among $\{c_2,\dots,c_5\}$, so assigning weight $1/2$ to each of these candidates gives every voter utility at least $1$.
\Cref{figure:counter_example_weak_fPO} strengthens this: $W = \{c_1, c_2\} \in \PO(I)$ is \emph{strictly} fractional Pareto-dominated.
Again $W= \{c_1, c_2\}$ gives every voter utility $1$, while every other pair of candidates leaves some voter at utility $0$. Yet, the uniform fractional committee over $\{c_3, \ldots, c_7\}$ strictly increases every voter's utility to $\frac{6}{5}$. 

As a next step in our comparison of $\PO$ and $\fPO$, we give a distinctive characterization of $\fPO$ in terms of \emph{cloning-robustness} \citep{DBLP:journals/jair/ElkindFS11,tideman1987independence,DBLP:conf/aaai/TalmonF19,DBLP:conf/sigecom/BerkerCROCE25}.

\begin{definition}[$r$-cloned instance]
    Given an instance $I=(V,k)$, its $r$-cloned instance $rI$ is obtained by replacing every candidate $c\in C$ by $r$ identical copies $c^1,\ldots,c^r$, each with valuation $v_{ic^\ell}=v_{ic}$ for all $i\in N$ and $\ell\in[r]$, and setting the committee size to $rk$.
\end{definition}

For a committee $W\subseteq C$, let $rW=\{c^\ell:c\in W,\ \ell\in[r]\}$ denote its canonical cloned committee.
We call a property $\mathcal{T}$ \emph{cloning-robust} if, for any instance $I$ and every committee $W \in \mathcal{T}(I)$, its corresponding cloned committee satisfies the same property in the $r$-cloned instance, i.e., $rW \in \mathcal{T}(rI)$ for every $r \in \N_{>0}$.
Under this definition, $\fPO$ is precisely the largest cloning-robust refinement of Pareto-optimality.

\begin{restatable}{theorem}{cloningInvariance} \label{thm:cloningInvariance}
    Let $I=(V,k)$ and $W\in \Committees$. 
    Then, 
    \[
        W \in \fPO(I) \quad \iff \quad rW \in \PO(rI) \text{ for all } r \in \N_{>0}.
    \]
    The same holds for $\wfPO$ and $\wPO$.
\end{restatable}

Since not every Pareto-optimal committee satisfies $\fPO$, we can directly conclude that $\PO$ violates cloning-robustness. In contrast, $\fPO$ and $\wfPO$ are cloning-robust.

\begin{restatable}{corollary}{cloningInvariancePartTwo} \label{thm:cloningInvariance_corollary}
    $\fPO$ and $\wfPO$ are cloning-robust. $\PO$ and $\wPO$ violate cloning-robustness.
\end{restatable}

Together, \Cref{thm:cloningInvariance} and \Cref{thm:cloningInvariance_corollary} give a normative justification for studying $\fPO$ even when only integral committees are ever implemented.
A committee that is $\PO$ but not $\fPO$ owes its efficiency to the coarseness of the current candidate set: uniform cloning is a simple way to refine that set, and already this can render the committee inefficient, e.g., in \Cref{figure:counter_example_PO_fPO_small}, it suffices to clone every candidate once to make the cloned committee $2 \times \{c_1, c_2\}$ Pareto-dominated by $\{c_2, \dots, c_5\}$ with $k = 2\cdot 2$.
In this sense, $\fPO$ certifies an efficiency that is robust to how finely the candidate set happens to be drawn up.

\medskip
\noindent\textbf{Remark.}
Although our general model allows for arbitrary additive utilities $v_{ic}\in\mathbb{Q}_{\ge 0}$, separations between $\PO$ and $\fPO$ are most meaningful in the approval-based setting.
For illustration, consider the following simple instance with additive utilities: two voters, three candidates, $k=1$, and valuations $v_1 = (1,3,0)$ and $v_2 = (1,0,3)$.
Even though $W = \{c_1\}$ satisfies $\PO$, it is strictly fractionally dominated by the uniform convex combination over $c_2$ and $c_3$.
Approval utilities, by contrast, impose a strong combinatorial restriction.
Any negative result obtained in this restricted model is therefore considerably stronger, since it shows that the gap between $\PO$ and $\fPO$ is not merely an artifact of adversarially picked cardinal utilities.
\medskip

With this in mind, a natural next step is to quantify how far a Pareto-optimal committee that fails $\fPO$ can be from the set of fractionally Pareto-optimal committees. We consider two complementary notions of distance: one based on voters' utilities and one based on the composition of the committees themselves. In both cases, we show that the gap can be as large as possible.

\begin{restatable}{theorem}{ArbitrarilyBad} \label{thm:arbirtrarilyBad}
    For every $k \geq 1$ and every $\varepsilon>0$, there exists an approval-based instance $I=(A,k)$, a committee $W\in\PO(I)$, and a fractional committee $\p\in\Probs$ such that
    \[
        (k-\varepsilon)\cdot u_i(W) < \sum_{c\in A_i}\p(c) \qquad\text{for all } i\in N.
    \]
\end{restatable}

\Cref{thm:arbirtrarilyBad} shows that, in terms of the utility voters derive from the committee, the multiplicative gap between $\PO$ and $\fPO$ can be arbitrarily close to $k$.
Such integral Pareto-optimal committees are therefore efficient only in a rather weak sense: a fractional reallocation can make every voter substantially better off.
An alternative interpretation of the theorem is that, in some instances, slightly more than one fractional candidate (seat) can provide every voter with a higher utility than a Pareto-optimal committee achieves using $k$ integral candidates. This follows from scaling the fractional committee $\p$ from the above theorem to $\p/(k-\varepsilon)$ with $\varepsilon > 0$ small.

Turning from distances in utility to distances in committee space, for two committees $W,W'\in\Committees_k$, we define their \emph{swap distance} by $d(W,W') \coloneqq \lvert W\setminus W' \rvert$. Equivalently, $d(W, W')$ is the minimum number of single-candidate swaps required to transform $W$ into $W'$, where a single-candidate swap replaces one selected candidate by an unselected one. 

\begin{restatable}{theorem}{WorstDistanceTOfPO} \label{thm:worst_distance_to_fPO}
    For every $k>1$, there exists an instance $I=(A,k)$ and committee $W\in\PO(I)$ such that $\min_{W' \in \fPO(I)}d(W, W') = k-1$.
\end{restatable}

That is, there exist Pareto-optimal committees that share at most one candidate with any fractionally Pareto-optimal committee. In the approval setting, this bound is tight: every candidate contained in a Pareto-optimal committee also belongs to some $\fPO$ committee (see the proof in \Cref{appendix:axiomatic_relationships}). In contrast, in the additive-utility setting, there can be Pareto-optimal committees disjoint from all fractionally Pareto-optimal committees (see instance from preceding remark).

\section{Weighted Welfare Maximization and $\fPO$} \label{section:weight_characterization}

In this section, we characterize the set of fractionally Pareto-optimal committees to derive algorithmic tractability results and to analyze structural properties of committees satisfying $\fPO$ in \Cref{sec:conn}.
More specifically, every committee satisfying (weak) fractional Pareto-optimality arises as a solution to a weighted utility maximization problem.
Formally, for a given instance $I$ and weight vector $\weight\in \Weights$, we define:
\begin{equation}\label{equation:weighted_utilitarian}
    \PO^\weight(I) = \Big\{W \in \Committees \colon \chi^W \in \arg\max_{\p \in \Probs}\sum_{i \in N}\weight_i u_i(\p) \Big\}
\end{equation} 
to be the set of integral committees maximizing the $\weight$-weighted utilitarian welfare $\sum_{i\in N}\weight_i u_i(\p)$ over all fractional committees $\p\in\Probs$.
We say that a committee $W\in \Committees$ is \emph{supported} by a weight $\weight \in \Weights$ if $W \in \PO^\weight(I)$. 
Since $\Probs$ is a non-empty compact polytope and the objective $\p \mapsto \sum_{i\in N}\weight_i u_i(\p)$ is linear, the maximum is attained. Moreover, any linear objective over a polytope is maximized at a vertex, and the vertices of $\Probs$ are exactly the integral committees, so $\PO^\weight(I)\neq\emptyset$.

The following characterization is a direct consequence of classical results in welfare economics (Pareto-efficient allocations are linked to maximizers of weighted welfare~\citep{FeSe05a}) and in multiobjective optimization (efficient solutions of multiobjective linear programs are characterized via weighted-sum scalarizations~\citep{geoffrion1968proper,BoVa04a}):

\begin{restatable}{theorem}{CharacterizationFPO} \label{theorem:fractional_po_equal_weights}
    For a given instance $I$ and committee $W \in \Committees$, 
    \begin{equation*}
    \begin{alignedat}{2}
        \text{(i)}\quad & W \in \fPO(I)  &&\iff \exists\, \weight \in \Delta^+ : W \in \PO^\weight(I),\\
        \text{(ii)}\quad & W \in \wfPO(I) &&\iff \exists\, \weight \in \Delta : W \in \PO^\weight(I).
    \end{alignedat}
    \end{equation*}
\end{restatable}

\begin{proof}
    Our setting corresponds to a multiobjective linear program: we maximize the linear functions $u_1,\dots,u_n$ over the polytope $\Probs$.
    By definition, $W \in \fPO(I)$ [$W \in \wfPO(I)$] if and only if $\chi^W$ is a (weakly) efficient solution of the multiobjective linear program maximizing $u_1,\dots,u_n$ over $\Probs$, and $W \in \PO^\weight(I)$ if and only if $\chi^W$ maximizes $\sum_{i\in N}\weight_i u_i(\p)$ over $\Probs$. We now show the two directions of the statement:

    ``$\Leftarrow$'': For $(i)$, let $W\in\PO^\weight(I)$ for some $\weight\in\Weights^+$. For contradiction assume that there is some $\p\in\Probs$ with $u_i(\p)\ge u_i(W)$ for all $i\in N$ with one inequality strict, then $\sum_{i\in N}\weight_i u_i(\p)>\sum_{i\in N}\weight_i u_i(W)$, a contradiction; hence $W\in\fPO(I)$. For $(ii)$, let $W\in\PO^\weight(I)$ for some $\weight\in\Weights$. If for some $\p\in\Probs$ we had $u_i(\p)>u_i(W)$ for all $i\in N$, then again $\sum_{i\in N}\weight_i u_i(\p)>\sum_{i\in N}\weight_i u_i(W)$ follows since $\weight\ge 0$ and $\weight\neq 0$; hence, $W\in\wfPO(I)$.
    
    ``$\Rightarrow$'': For~(i), \citet{isermann1974proper} showed that every efficient solution of a multiobjective linear program is properly efficient. Since the objectives are linear and $\Probs$ is convex, \citet{geoffrion1968proper}  yields for every $W\in \fPO$, some $\weight\in\R^n_{>0}$ such that $\chi^W$ maximizes $\sum_{i\in N}\weight_i u_i(\p)$ over $\Probs$; normalizing $\weight$ gives $\weight\in\Weights^+$. Part~(ii) follows from the analogous characterization of weakly efficient solutions of convex problems via non-negative weights~\citep{ehrgott2005multicriteria}.
\end{proof}

\begin{figure}[!t]
    \centering
    \begin{tikzpicture}[
        transform shape,
        node distance=1cm and 1.4cm,
        box/.style={
            draw,
            rounded corners,
            align=center,
            inner sep=2.5pt,
            minimum width=2.5cm,
            minimum height=0.7cm
        },
    ]
    \node[box] (topk) {Top-\(k\) candidates\\
    with respect to \(\succsim_{\weight}\)};
    
    \node[box, left=2cm of topk] (supported) {Supported committee\\
    $W \in \PO^\weight(I)$};
    
    \node[box, below left=0.8cm and 0.2cm of supported] (fpo) {$W \in \fPO(I)$};
    \node[box, below right=0.8cm and 0.2cm of supported] (wfpo) {$W \in \wfPO(I)$};
    \node[box, below=0.8cm of fpo] (po) {$W \in \PO(I)$};
    \node[box, below=0.8cm of wfpo] (wpo) {$W \in \wPO(I)$};
    
    \draw[Latex-Latex] (topk) -- node[above] {\Cref{lemma:weights_to_scores}} (supported);
    \draw[Latex-Latex] (supported) -- node[above left] {$\exists \weight \in \Weights^+$} (fpo);
    \draw[Latex-Latex] (supported) -- node[above right] {$\exists \weight \in \Weights$} (wfpo);
    \node[below= 0.1cm of supported] (name) {\Cref{theorem:fractional_po_equal_weights}};
    
    \draw[-Latex] (fpo) -- (wfpo);
    \draw[-Latex] (po) -- (wpo);
    \draw[-Latex] (fpo) -- node[right] {\Cref{proposition:axiomatic_relationships}} (po);
    \draw[-Latex] (wfpo) -- node[right] {\Cref{proposition:axiomatic_relationships}} (wpo);
    \end{tikzpicture}
    \caption{Relationships between score-induced weak orders, weighted utility maximization, and the efficiency notions studied in this paper. 
    }
    \label{fig:characterization-overview}
\end{figure}

Next, we show that computing the committees supported by a given weight vector reduces to ranking the candidates by a candidate-specific score and selecting the top-$k$, with ties broken arbitrarily: Given a weight vector $\weight \in \Weights$, we define the score $s_c(\weight)$ of a candidate $c \in C$ as the weighted sum of the voters' valuations for $c$, that is, $s_c(\weight) \coloneqq \sum_{i \in N} \weight_i v_{ic}$.
Accordingly, every weight vector $\weight\in\Weights$ induces a weak order $\succsim_\weight$ over the candidates, where $c \succsim_\weight c'$ if and only if $s_c(\weight) \geq s_{c'}(\weight)$.

\begin{restatable}{lemma}{weightsToScores}\label{lemma:weights_to_scores}
    Let $I$ be an instance and $\weight \in \Weights$.
    A committee $W \in \Committees$ belongs to $\PO^\weight(I)$ if and only if $W$ consists of the $k$ highest-ranked candidates according to $\succsim_\weight$ under some tie-breaking.
\end{restatable}

If there are no ties, the characterization amounts to selecting the $k$ candidates with the highest scores, since valuations are additive.
If several candidates are tied for the $k$-th highest score, we include every candidate whose score lies strictly above the tied value and fill the remaining seats with an arbitrary subset of the tied candidates.
Intuitively, each such choice corresponds to a supported committee and a point of the associated face of the utility polytope.

Combining \Cref{theorem:fractional_po_equal_weights} and \Cref{lemma:weights_to_scores}, we conclude that a committee $W$ satisfies $\fPO$ if and only if it meets the following separation criterion:
\begin{equation}\label{equation:separation_condition}
    \exists\, \weight \in \Weights^+,\ \forall\, c \in W,\ c' \in C\setminus W: \quad s_c(\weight) \geq s_{c'}(\weight).
\end{equation}
The corresponding result for $\wfPO$ is obtained by asking for the existence of a non-negative instead of positive weight vector $\weight$.

\smallskip
\noindent\textbf{Computability.}
A notable advantage of fractional Pareto-optimality is that it is computationally tractable.
While \citet{Azi20} showed that deciding whether a committee is Pareto-optimal is coNP-complete, using standard linear-programming techniques~\citep[Theorem~3.3]{yu1975set}, we observe that $\fPO$ can be checked in polynomial time.
Moreover, when $W$ satisfies $\fPO$, we can compute a succinct certificate of this fact in polynomial time, namely a supporting weight vector $\weight$. Conversely, by \Cref{lemma:weights_to_scores}, checking whether $W$ is supported by a given weight vector $\weight$ only requires comparing candidate scores.

\begin{restatable}{proposition}{ComputabilityFPO}\label{theorem:computing_fPO_domination}
    Given an instance $I$, verifying \textnormal{[weak]} fractional Pareto-optimality for a committee $W \in \Committees$ can be done in polynomial time. 
    Further, if $W \in \fPO(I)$ \textnormal{$[W \in \wfPO(I)]$}, a weight vector $\weight \in \Weights^+$ \textnormal{$[\weight \in \Weights]$} such that $W \in \PO^\weight(I)$ can be computed in polynomial time. 
\end{restatable}

\section{Structural Properties of $\PO$ and $\fPO$}\label{sec:conn}
Building on our characterization, we establish two structural results for the set of fractionally Pareto-optimal committees: their monotonicity and connectedness. Taken together, these two properties suggest that the set of fractionally Pareto-optimal committees is well-behaved.

\medskip
\noindent\textbf{Committee monotonicity} expresses a natural form of consistency across committee sizes.
Intuitively, a committee that satisfies a certain property can be enlarged or reduced by some candidate while upholding this property.

\begin{definition}[Committee monotonicity]
    A family of committee sets $(S_k(V))_{k=1}^{m}$ satisfies \emph{committee monotonicity} if, for every profile $V$ and every $k < m$,
    \begin{itemize}
        \item every committee in $S_k(V)$ is contained in some committee in $S_{k+1}(V)$, and
        \item every committee in $S_{k+1}(V)$ contains some committee in $S_k(V)$.
    \end{itemize}
\end{definition}

Committee monotonicity has previously been studied as an axiom for multiwinner voting rules \citep{DBLP:journals/scw/ElkindFSS17}. It turns out that the sets of $\fPO$ and $\wfPO$ committees satisfy this property; in contrast, we show that $\PO$ and $\wPO$ fail it. 
Let $\fPO_k(V)$ \textnormal{[$\wfPO_k(V)$]} denote the set of $\fPO$ \textnormal{[$\wfPO$]} committees for profile $V$ and committee size $k$. The definitions for $\PO$ and $\wPO$ are analogous.

\begin{restatable}{proposition}{CommitteeMonotonicity} \label{th:CommitteeMonotonicity}
     For every additive profile $V$, $(\fPO_k(V))_{k=1}^{m}$ and $(\wfPO_k(V))_{k=1}^{m}$ satisfy committee monotonicity.
\end{restatable}

\begin{proofsketch}
    By \Cref{lemma:weights_to_scores}, a committee $W$ is (weakly) fractionally Pareto-optimal if and only if it consists of the top-$k$ candidates under the score order induced by some weight vector $\weight$. 
    Adding one of the next highest-scoring candidates under $\weight$ to $W$ yields a top-$(k+1)$ committee, and removing a lowest-scoring candidate from a supported committee of size $k+1$ under $\weight$ yields a top-$k$ committee.
    Thus, $(\fPO_k(V))_{k=1}^m$ and $(\wfPO_k(V))_{k=1}^m$ satisfy committee monotonicity for every profile $V$.
\end{proofsketch}

In contrast, we show that (weak) Pareto-optimality violates committee monotonicity.
Our proof relies on counterexamples generated via a SAT encoding, both of which we present in Appendix~\ref{appendix:structural_properties}. 
\begin{restatable}{proposition}{violationMonotonicity}
    There exist approval profiles $A, A'$ such that $(\PO_k(A))_{k=1}^{m}$ and $(\wPO_k(A'))_{k=1}^{m}$ violate committee monotonicity.
\end{restatable}

Our SAT encoding also certifies the minimality of our obtained counterexamples with respect to the number of candidates, which yields the following tight bounds:

\begin{restatable}{proposition}{committeeMonotoneM}
    For every approval profile $A$ with $m \le 8$ \textnormal{[$m \le 9$]}, $(\PO_k(A))_{k=1}^{m}$ \textnormal{[$(\wPO_k(A))_{k=1}^{m}$]} satisfies committee monotonicity.
\end{restatable}

\medskip
\noindent\textbf{Connectedness} expresses that committees of a fixed size satisfying a certain property form a connected set under single-candidate swaps:

\begin{definition}[Connectedness]
    A set of committees $\mathcal{T} \subseteq \Committees$ is \emph{connected} if for every two committees $W, W' \in \mathcal{T}$, there exists a sequence of committees $W_0, W_1, \dots, W_t \in \mathcal{T}$ such that $W_0 = W$, $W_t = W'$, and for every $\ell \in \{0,\dots,t-1\}$, the committee $W_{\ell+1}$ can be obtained from $W_\ell$ by swapping exactly one candidate, i.e., $\lvert W_\ell \setminus W_{\ell+1}\rvert = \lvert W_{\ell+1} \setminus W_\ell\rvert = 1$.
\end{definition}

This notion is closely related to reconfiguration problems, where one asks whether a given feasible solution can be transformed into another through a sequence of elementary local moves while preserving a desired property throughout \citep{chen2024multiwinnerreconfiguration,dong2026reconfiguring}.
In the fractional setting, the set of (weakly) Pareto-optimal fractional committees is a connected union of faces of the feasible polytope~\citep[Theorem~2.2]{luc1989theory}. 
The vertices of these faces correspond to the $\fPO$ ($\wfPO$) committees. Every edge between such vertices in the polytope corresponds to a single-candidate swap~\citep{GGMS87a}. Thus, vertices of the same face are connected via single-candidate swaps. Furthermore, vertices of intersecting faces are connected via a path over a shared vertex of both faces. Since the union of the faces is connected, this yields the following result. 

\begin{restatable}
{proposition}{Connectedness} \label{theorem:connectedness}
    For every instance $I$, $\fPO(I)$ and $\wfPO(I)$ are connected.
\end{restatable}

To upper bound the number of required swaps to get from one $\fPO$ ($\wfPO$) committee $W$ to another $W'$, we provide an alternative proof for \Cref{theorem:connectedness} involving the construction of a direct path between committees via candidate scores and \Cref{lemma:weights_to_scores}.
Note that $s_c(\weight)=\sum_{i\in N}\weight_i v_{ic}$ is linear.
Since both $W$ and $W'$ are fractionally Pareto-optimal, there exist supporting weights $\weight$ and $\weight'$ by \Cref{theorem:fractional_po_equal_weights}, respectively. Interpolating linearly via $\gamma(\lambda)\coloneqq (1-\lambda)\weight+\lambda\weight'$ makes each candidate's score a linear function of $\lambda$.
As $\lambda$ increases from $0$ to $1$, the supported committee changes only when the identity of the top-$k$ candidates changes.
Viewing each candidate's score as a line in $\lambda$, every such change is a breakpoint at which a line that was below the $k$-th highest rises to become one of the $k$ highest, which corresponds to a single-candidate swap.\footnote{We address the case of simultaneous crossings in the proof of \Cref{proposition:diameter}.}

Thus, the swaps along the interpolation correspond to the breakpoints traced by the $k$-th highest score line as $\lambda$ varies from 0 to 1. The maximum number of such breakpoints is the complexity of the $k$-level in an arrangement of $m$ lines --- a classical problem from discrete geometry --- and its best-known upper bound \citep{Dey98} yields the following result:

\begin{restatable}{proposition}{diameter}\label{proposition:diameter}
    Any two committees in $\fPO$\textnormal{[$\wfPO$]} can be transformed into one another by a sequence of at most $\mathcal{O}(m \cdot k^{1/3})$ single-candidate swaps within $\fPO$\textnormal{[$\wfPO$]}.
\end{restatable}

Complementing this upper bound, we establish two lower bounds on the number of required swaps, one for the additive-utility setting and one for the approval-based setting.

\begin{restatable}{proposition}{diameterLower}
\label{proposition:lower_diameter}
    \begin{enumerate}[leftmargin=*,label=\textit{(\roman*)}]
        \item For every $k\geq 1$, there exists an approval instance $I=(A,k)$ and committees $W,W'\in\fPO(I)$ such that every $\fPO$-preserving swap sequence from $W$ to $W'$ has length at least $2k-1$.
        \item For every $m>k\geq2$, there exists an instance $I=(V,k)$ and committees $W,W'\in\fPO(I)$ such that every $\fPO$-preserving swap sequence from $W$ to $W'$ has length at least $m-k$.
    \end{enumerate}
\end{restatable}

The first construction forces any $\fPO$-preserving path between two disjoint committees to detour through auxiliary candidates, yielding the $2k-1$ bound. The second realizes the entire $\fPO$ swap graph as a sliding-window path of length $m-k$.

In contrast, the set of (weak) Pareto-optimal committees can be disconnected. For additive utilities and $\PO$, this is implied by the work of \citet{gorski2011connectedness}. Regarding the approval setting, it was conjectured by \citet{joshua-thesis} to be connected. However, we show that the negative result extends to the approval setting, again via a counterexample obtained using a SAT formulation, thereby resolving the open problem.

\begin{restatable}{proposition}{disconnectedness} \label{proposition:disconnectedness}
    There exist approval instances $I, I'$ for which $\PO(I)$ and $\wPO(I')$ violate connectedness.
\end{restatable}

\section{The Tension between Equality-Oriented Voting Rules and $\fPO$}\label{sec:rules}

While the set of fractionally Pareto-optimal committees turns out to admit strong structural properties, we now show that $\fPO$ can conflict with equality-oriented objectives, causing many desirable voting rules to violate the axiom.
Many prominent (approval-based) sequential rules with appealing fairness guarantees, such as the Method of Equal Shares (MES), sequential PAV, and sequential Phragmén, already fail Pareto-optimality. The same is true for the optimization variant of Phragmén's rule \citep{lackner2023multi}. 
We focus on classes of voting rules that contain most of the popular Pareto-optimal rules.

\subsection{Welfarist Rules}
Many established voting rules (e.g., AV, PAV, NSW) fall into the class of so-called \emph{welfarist} rules \citep{PeSk20a}, which evaluate each committee solely by its induced utility vector.
For this, let $g = (g_n)_{n \geq 1}$ be a family of monotone aggregation functions with $g_n \colon \R^n_{\ge 0} \to \R$.\footnote{Whenever we say that $g$ satisfies a property $X$, we mean that every function $g_n$ in the family satisfies $X$.}

\begin{definition}[Welfarist rules] 
    A voting rule is \emph{welfarist} if it can be written as 
    $$f_g(I) = \arg\max_{W \in \Committees}g_n(u(W)),$$ 
    for every instance $I$ with $n$ voters. 
\end{definition}

Naturally, the properties of $f_g$ depend on those of $g$.
In particular, convex aggregation functions reward uneven utility gains and therefore tend to favor committees that concentrate utility, whereas concave aggregation functions reflect diminishing marginal returns and are naturally associated with more equality-oriented utility outcomes.
The central result of this section is that $\fPO$ is incompatible with such equality-oriented welfarist rules. In \Cref{sec:Thiele}, we refine this analysis for a subclass, known as Thiele rules.

\citet{PeSk20a} have shown that if $g$ is \emph{strictly monotone}, then $f_g$ satisfies $\PO$.
If $g$ is additionally convex, one can use Jensen's inequality to show that $f_g$ satisfies $\fPO$.

\begin{restatable}{proposition}{convexWelfarist} \label{th:convexWelfarist}
    Let $g$ be strictly monotone and convex. Then, $f_g$ satisfies $\fPO$.
\end{restatable}

\begin{proof}   
    Let $W \in f_g(I)$ for some instance $I$ with $n$ voters and let $g$ be convex. 
    Hence, $W \in \arg\max_{W' \in \Committees}g_n(u(W'))$. Assume for contradiction that $W \notin\fPO(I)$.
    Then, there exists a dominating fractional committee $\p \in \Probs$ with $u_i(\p)\geq u_i(W)$ for all $i \in N$, with at least one strict inequality. By strict monotonicity of $g$, $g_n(u(\p)) > g_n(u(W))$.
    We can write $\p$ as a convex combination of integral committees \citep{schrijver2003combinatorial}, which gives
    $\p=\sum_{j}\lambda_j \chi^{W_j}$ with $\lambda_j \in [0, 1]$ and $\sum_j\lambda_j = 1$.
    By linearity of $u(\cdot)$, we get $u(\p)=\sum_j \lambda_j u(W_j)$. This yields
    \begin{equation*}
        \sum_j\lambda_j g_n(u(W_j)) \ge g_n\Big(\sum_j\lambda_j u(W_j)\Big)\ = g_n(u(\p)) \ > \ g_n(u(W)).
    \end{equation*}
    The first inequality follows from Jensen's inequality, since $g$ is convex and the last by the strict monotonicity of $g$.
    Hence, there exists $j$ with $g_n(u(W_j))> g_n(u(W))$.
    Hence, $W$ cannot be a maximizer of $g$ among integral committees; a contradiction.
\end{proof}

By contrast, no strictly concave aggregation function can even guarantee $\wfPO$ in the additive-utility setting. For every strictly concave aggregation function $g$, there exists an instance on which $f_g$ selects a committee that is not weakly fractionally Pareto-optimal.

\begin{restatable}{theorem}{concaveWelfarist} \label{theorem:concave_welfarist}
    Let $g$ be strictly concave. For every $n \geq 2$, there exists an instance $I$ with $n$ voters on which $f_g$ violates $\wfPO$.
\end{restatable}

For additive utilities, this yields a sharp contrast between welfarist rules that satisfy $\fPO$ and those that violate it.
In the approval setting, there exists a strictly concave and strictly monotone family $g$ for which the corresponding welfarist rule $f_g$ satisfies $\fPO$. The construction uses a slight perturbation of the utilitarian welfare objective in order to make it strictly concave. We provide it in \Cref{appendix:welfarist_rules}.

Hence, in the approval setting, concavity alone is not enough to create a conflict with $\fPO$.
The conflict reappears, however, as soon as $g$ exhibits a preference for covering all voters:
a welfarist rule $f_g$ with $g_n\colon \R_{\geq0}^n \to \R$ is called \emph{zero-avoiding} if for every $k$ and $n$,
$$g_n(\mathbf{1}_n) > \max_{x\in\{0,\dots,k\}^n:\, \min_i x_i=0} g_n(x).$$ 
In words, a zero-avoiding rule prefers a committee in which every voter receives utility $1$ above every committee that leaves some voter with utility $0$; it therefore prioritizes universal coverage over additional utility for already-represented voters.
To see that zero-avoidance is incompatible even with $\wfPO$ in the approval-based setting, consider the example in \Cref{figure:counter_example_weak_fPO}. 
This class includes several prominent egalitarian objectives, such as Nash social welfare and egalitarian welfare.

\subsection{Thiele Rules}\label{sec:Thiele}
One of the most prominent classes of multiwinner voting rules is the class of Thiele rules \citep{thiele1895om, lackner2023multi, janson2016phragmen, lackner2018consistent}.
Under such rules, each voter's utility is transformed into a score by a common function $h$, and committees are evaluated by summing these scores across voters.
Thereby, monotone Thiele rules form a well-structured subclass of welfarist rules. 
Since Thiele rules are typically studied in the approval-based setting, we restrict our attention to approval profiles and $h$-Thiele rules for $h\colon\N\to\R$ for the remainder of this section. 

\begin{definition}[Thiele rules]
    Let $h\colon \N \to \R$ be a (scoring) function.
    For a given instance $I$, the corresponding \emph{$h$-Thiele rule} selects the committees $W\in \Committees$ maximizing $\sum_{i\in N} h(u_i(W))$.
\end{definition}

\citet{lackner2023multi} have shown that all strictly monotone Thiele rules satisfy Pareto-optimality. We find that the tension between $\fPO$ and equality-driven aggregation functions, already observed for welfarist rules, admits an even sharper formulation for Thiele rules.
We define the \emph{marginal contribution} sequence of a function $h\colon\N\to\R$ to be $\delta_h(t) \coloneqq h(t)-h(t-1)$ for $t \in \N_{>0}$.

We now show that if the marginal contribution sequence $\delta_h$ decreases at any point, then $h$-Thiele violates even $\wfPO$. 

\begin{restatable}{theorem}{ThieleFPOviolation}\label{theorem:Thiele}
    Let $h$ be a function for which there exists some $t \in \N_{> 0}$ with $\delta_h(t)>\delta_h(t+1)$. Then, $h$-Thiele  violates $\wfPO$.
\end{restatable}

\begin{proofsketch}
    Let $h$ be a function with $\delta_h(t)>\delta_h(t+1)$ for some $t \in \N_{> 0}$.
    We construct an instance in which $h$-Thiele  violates $\wfPO$,  building directly on the counterexample in \Cref{figure:counter_example_weak_fPO}.
    We add $t-1$ universally approved candidates, denoted $C_{\mathrm{opt}}$, and set $k=t+1$. We then use the approval pattern from \Cref{figure:counter_example_weak_fPO} to add two candidates $c_1$ and $c_2$ that have disjoint supports and together cover all voters. Thus, the committee $W=C_{\mathrm{opt}}\cup\{c_1,c_2\}$ gives every voter utility exactly $t$.
    We add further candidates as in \Cref{figure:counter_example_weak_fPO}.
    The key idea is then that replacing $\{c_1,c_2\}$ by another pair from the counterexample creates a trade-off: some voters lose one unit of utility, moving from $t$ to $t-1$, while others may gain one unit, moving from $t$ to $t+1$. 
    We can adapt the approval pattern in \Cref{figure:counter_example_weak_fPO} so that this change will always decrease the $h$-Thiele score of the committee, exploiting that  $\delta_h(t)>\delta_h(t+1)$:
    Under the $h$-Thiele score, each voter who decreases her utility by one decreases the $h$-Thiele score by $\delta_h(t)$, whereas every voter who increases her utility by one increases the $h$-Thiele score by $\delta_h(t+1)$. Consequently, $W$ uniquely maximizes the $h$-Thiele score.

    At the same time, our adaptation preserves the fPO-domination of $W$ from \Cref{figure:counter_example_weak_fPO}: the $t-1$ universally approved candidates, together with the fractional distribution of the two remaining seats among the alternative candidates, give every voter strictly more than $t$ utility. Thus, $W\notin\wfPO(I)$. 
\end{proofsketch}

In particular, \Cref{theorem:Thiele} applies to Proportional Approval Voting (PAV), where $h(t)=\sum_{\ell=1}^t \frac{1}{\ell}$.
Here, the marginal contributions are $\delta_h(t)=\frac{1}{t}$, which is strictly decreasing.
Hence, by \Cref{theorem:Thiele}, PAV violates even $\wfPO$.
More generally, \Cref{theorem:Thiele} shows that Thiele rules implementing the standard diminishing-returns intuition, which underlies many fairness-oriented approaches to utility aggregation, violate even $\wfPO$.  
Using \Cref{th:convexWelfarist} and \Cref{theorem:Thiele}, we can show a characterization of when Thiele rules satisfy~$\fPO$:

\begin{restatable}{proposition}{corrolaryCharacterization}\label{co:Thi}
    Let $h$ be any function, then $h$-Thiele satisfies $\fPO$ if and only if $h$ satisfies $\delta_h(t)\leq\delta_h(t+1)$ for all $t \in \N_{>0}$ and $h$ is strictly monotone.
\end{restatable} 
 
Thus, among strictly monotone Thiele rules, fractional Pareto-optimality is equivalent to a convex scoring function that rewards concentrating rather than spreading voters' utility and, in particular, is incompatible with the diminishing-marginal-returns principle underlying many fairness and minority-protection ideals.

\section{Restricted Domains} \label{sec:restricted}
While we have seen that $\PO$ and $\fPO$ may differ in general, it is natural to ask under which conditions both coincide.
Focusing on the approval setting, we show that on instances satisfying a single dominance condition on candidates and a mild closure requirement, $\PO$ and $\fPO$ are equivalent.
These restrictions are general enough to cover many well-studied subdomains, foremost among them, linearly consistent (LC) preferences.
This class in turn includes the previously studied voter-candidate-interval (VCI), voter-interval (VI), and candidate-interval (CI) subdomains \citep{DBLP:conf/aaai/GodziszewskiB0F21,ElkindLackner2015}.\footnote{In the VCI domain, each voter $v$ and candidate $c$ is assigned a position $x_v,x_c\in\R$ and a radius $r_v,r_c\ge 0$, and $v$ approves $c$ if $|x_c-x_v|\le r_c+r_v$, i.e., their intervals on the line intersect. CI and VI are the special cases in which all candidates or all voters have a radius of zero, respectively.} 

Our starting point is the work of \citet{joshua-thesis}, who established that, in the VI and CI domains, the set of Pareto-optimal committees satisfies committee monotonicity and connectedness.
\citet{joshua-thesis} showed this by introducing the single-dominance-only (SDO) property that captures instances in which Pareto-optimality reduces to checking whether an unselected candidate Pareto-dominates a selected candidate and showing that VI and CI satisfy SDO.

\begin{definition}[Single dominance only]
    An instance $I=(A, k)$ satisfies \emph{single dominance only (SDO)} if, for every committee $W\in \Committees$,
    \begin{equation*}
\begin{split}
    W \in \PO(I) \iff{} \nexists\, c\in W,\ d\in C\setminus W \text{ such that }
    \{d\} \text{ Pareto-dominates } \{c\}.
\end{split}
\end{equation*}
\end{definition}

We refer to a set of approval-based instances as an \emph{approval domain} $\dom$ and say that $\dom$ is \emph{closed under cloning} if, for every instance $I\in\dom$ and every $r\in\N_{>0}$, the $r$-cloned instance $rI$ also belongs to $\dom$.
With this, we can state the equivalence of $\fPO$ and $\PO$ under these conditions.

\begin{restatable}{theorem}{thmclspofpo}\label{lemma:closure}
    For every approval domain $\dom$ that is closed under cloning and for which all instances from the domain satisfy SDO, we have $\fPO(I)=\PO(I)$ for all $I\in \dom$.
\end{restatable}
\begin{proof}
    The inclusion $\fPO(I)\subseteq\PO(I)$ holds on every instance (\Cref{proposition:axiomatic_relationships}).
    For the reverse direction, assume for the sake of contradiction that an instance $I\in \dom$, together with a committee $W \in \PO(I)$, exists such that $W \notin \fPO(I)$.
    By \Cref{thm:cloningInvariance}, there exists an $r \in \N_{>0}$ such that $rW \notin \PO(rI)$.
    Since $\dom$ is closed under cloning, $rI \in \dom$, and thus, the SDO property implies that $rW$ contains at least one candidate $c^{\ell}$ for some $\ell\in [r]$ that is $\PO$-dominated by a candidate $d^{\ell'} \notin rW$ for some $\ell'\in [r]$.
    As $rW$ consists of clones of candidates in $W$, we get $c \in W$ and with the same argument $d \in C \setminus W$.
    Thus, by the SDO property, $W$ is not Pareto-optimal, contradicting our assumption.
\end{proof}

Consequently, the positive results for $\fPO$ carry over to $\PO$ on every domain to which \Cref{lemma:closure} applies.
Furthermore, \citet{joshua-thesis} has established that for every instance $I=(A, k)$ satisfying SDO and any two committees $W_1, W_2 \in \PO(I)$, one can be transformed into the other via $k - \lvert W_1\cap W_2 \rvert$ single-candidate swaps, with every intermediate committee remaining Pareto-optimal.
This strengthens our connectedness result from \Cref{sec:conn} for these restricted instances.

A specific class of instances, namely \emph{linearly consistent} (LC) instances, introduced by \citet{pier2022corestable}, satisfies the conditions enforced in \Cref{lemma:closure}.
This domain is of particular interest because it subsumes the VCI domain and, by extension, the CI and VI domains \citep{AvramidisLassotaSchmidtKraepelinVetta2026}.
We provide definitions in \Cref{appendix:restricted_domain}.

\begin{restatable}{lemma}{thmLC} \label{thm:lc-sdo} 
    The LC domain is closed under cloning, and every instance in its domain satisfies SDO.
\end{restatable}

The key idea behind the proof is that, under the LC order, any Pareto-improving multi-candidate exchange contains a single selected candidate whose supporter set is strictly contained in that of an unselected candidate; cloning preserves LC by replacing each candidate with a consecutive block of copies.
Combining \Cref{thm:lc-sdo} and \Cref{lemma:closure} yields the following:

\begin{corollary} \label{cor:lc-fpo-po}
    On every LC, VCI, VI, and CI  instance $I$, $\fPO(I)=\PO(I)$. 
\end{corollary}

We show in \Cref{appendix:restricted_domain} that $\fPO$ and $\PO$ do not coincide on the 2D-Euclidean domain, which is another generalization of VCI.

\section{Discussion}\label{sec:disc}

Our results show that fractional Pareto-optimality is not merely a natural strengthening of Pareto-optimality, but a structurally different and substantially better-behaved efficiency notion.
The weighted-utilitarian characterization equips $\fPO$ committees with succinct certificates, an efficient verification procedure, and strong structural properties, such as connectedness under single-candidate swaps and committee monotonicity.
Moreover, by our cloning characterization, $\fPO$ singles out exactly those committees whose efficiency survives uniform refinements of the feasible set.

Our analysis of welfarist and Thiele rules reveals a tension between $\fPO$ and equality-driven voting rules.
Concurrent work by \citet{jin2026limitations} underlines this conflict by showing that there are instances for which no committee can both satisfy $\wfPO$ and justified representation (JR), a weak proportional representation axiom. It would be interesting to determine whether there exists a natural intermediate efficiency notion between $\PO$ and $\fPO$ that is compatible with JR and structurally more similar to $\fPO$.

Another natural direction to extend this work is participatory budgeting, which generalizes multiwinner voting by allowing candidates to have different costs. The weighted-utilitarian characterization of $\fPO$ extends naturally to this setting, yet $\fPO$ outcomes may fail to exist. Again, it would be interesting whether a natural, always satisfiable relaxation for $\fPO$ exists in this setting. 

Further, it would be interesting to empirically investigate how often established voting rules, like MES or PAV, violate $\fPO$ and how often there is no committee that can satisfy JR and $\fPO$ simultaneously.


\newpage
\appendix
\crefalias{section}{appendix}
\onecolumn

\section*{Disclosure of AI Usage}
We used large language models [Claude Fable 5.1, Claude Opus 4.7, Claude Opus 5.5, ChatGPT~5.5] to prepare this paper in the following ways: First, we used them to enhance the wording and exposition of the authors' drafts. Second, we used them as an additional check of the proofs written by the authors. 
Third, we used them to broaden our literature search (Section~1 and Appendix~A), particularly to identify prior work on specific aspects of the paper. In particular, this search pointed us to the $k$-level problem from discrete geometry and the bound of~\citep{Dey98}, which we use to prove the upper bound in \Cref{proposition:diameter}. All cited works were retrieved and checked against the original sources by the authors.
Fourth, they were used to assist in coding; in particular, various models of ChatGPT [ChatGPT~5.2, ChatGPT~5.6, ChatGPT~6] implemented our SAT formulations at various stages. Finally, the disk embedding in \Cref{fig:circles} was found with the help of ChatGPT~5.5 and Claude Opus~4.7; the authors verified it by checking all voter--candidate disk intersections against the approval profile of \Cref{figure:counter_example_PO_fPO_small}. 

\section{Related Work}\label{sec:RW}

\paragraph{Efficiency in multiwinner voting.}

Efficiency is one of the most undisputed properties in social choice and economic theory \citep{Arro51a, Debr59b, mas1995microeconomic,FeSe05a}.
In the study of multiwinner voting, it has nevertheless received little attention so far.
Pareto-optimality is the standard formulation of efficiency in the literature; the alternative formalizations we are aware of are either (substantially) weaker---such as unanimity~\citep{Caragiannisetal2024} or local non-wastefulness conditions~\citep{lackner2018consistent}, which require that no candidate without any approvals is included while approved candidates remain unselected---or of quantitative nature, capturing for instance the approximation factor with respect to utilitarian welfare~\citep{Lackner_2020,skowron2021proportionality}. 

Pareto-optimality itself has primarily been studied through three questions.
First, the complexity of verifying the axiom: \citet{Azi20} established coNP-completeness for several ways of extending preferences over candidates to preferences over committees, including approval preferences. Recently, \citet{dong2026jointly} showed that computing a committee satisfying both JR and Pareto-optimality is \textnormal{NP}-hard.
Second, analyzing which rules always return Pareto-optimal committees: For approval preferences, \citet{PeSk20a} showed that all strictly monotone welfarist rules are Pareto-optimal, whereas many prominent sequential rules fail this property (see \citet{lackner2023multi} for an overview).
Furthermore, \citet{laine2025pareto} distinguished different voting rules according to the kind of optimality requirement that they can satisfy in the ordinal multiwinner voting setting.
Third, its compatibility with other axioms: \citet{Caragiannisetal2024}, for example, showed that unanimity, as a minimal efficiency requirement, is incompatible with a variant of group-strategyproofness. Another focus lies on the trade-offs between efficiency and (proportional) representation \citep{Lackner_2020,skowron2021proportionality}.

\paragraph{Fractional committees in multiwinner voting.}
Fractional Pareto-optimality evaluates an integral committee relative to the set of feasible fractional committees.
Here, candidates may be selected to fractional extents, subject to the same committee-size constraint. 
Fractional committees have appeared previously in the best-of-both-worlds study of multiwinner voting, where one computes lotteries over integral committees that have desirable ex-post and ex-ante guarantees~\citep{DBLP:journals/corr/abs-2303-03642, suzuki2024maximum}. For instance, \citet{suzuki2024maximum} designed a rule that combines ex-ante Pareto-optimality with best-of-both-worlds fairness guarantees by using a weighted utilitarian welfare approach. 
Notably, a committee is fractionally Pareto-optimal if and only if it is ex-ante efficient in the best-of-both-worlds sense, offering an alternative interpretation and motivation of fractional Pareto-optimality.
Along different lines, previous work has established that moving from integral to fractional committees enables stronger desirable properties~\citep{DBLP:journals/teco/ChengJMW20, KrPe25b}, such as guaranteed core-stability under additive utilities~\citep{KrPe25b}.

\paragraph{Fractional Pareto-optimality in fair division.}
Fractional Pareto-optimality has already been considered in the study of fair division with indivisible goods~\citep{DBLP:conf/aaai/KawaseM26,DBLP:journals/ior/SandomirskiyS22,DBLP:conf/ijcai/GargMQ23,DBLP:conf/sigecom/BarmanKV18,DBLP:journals/jair/GargM24,DBLP:journals/tcs/GargM23}.
For instance, \citet{DBLP:conf/sigecom/BarmanKV18} study how to compute an allocation that is envy-free up to one good (EF1) and fractionally Pareto-optimal under additive utilities.
Their algorithm maintains and adjusts a weighted welfare objective, similar in spirit to the weight certificates that drive our characterization.
Further, analogous to our picture in multiwinner voting, it has been observed that---in contrast to Pareto-optimality, whose verification is coNP-complete~\citep{DBLP:conf/aldt/KeijzerBKZ09}---the fractional Pareto-optimality of an allocation can be verified in polynomial time via linear programming~\citep{DBLP:conf/sigecom/BarmanKV18,DBLP:journals/ior/SandomirskiyS22}.

\paragraph{Multiobjective optimization.}
Committee selection with additive utilities can be seen as a multiobjective combinatorial optimization problem: every voter's utility is a linear objective over the incidence vectors of size-$k$ committees, which makes it a multiobjective knapsack problem in which every item has unit weight and exactly $k$ items are selected. Its relaxation over $\Probs$ is a multiobjective linear program, and a committee is $\fPO$ [$\wfPO$] if and only if its incidence vector is an [weakly] efficient solution of this program. Consequently, several of our general results for additive utilities follow from these classical results: the characterization of (weakly) efficient solutions via weighted sums~\citep{geoffrion1968proper,isermann1974proper,ehrgott2005multicriteria} yields \Cref{theorem:fractional_po_equal_weights}, the linear-programming efficiency test~\citep{yu1975set} yields \Cref{theorem:computing_fPO_domination}, the connectedness of (weakly) efficient solution sets~\citep{luc1989theory} yields \Cref{theorem:connectedness}, and the fact that a convex function attains its maximum over a polytope at a vertex~\citep{rockafellar1970convex} yields \Cref{th:convexWelfarist}. Connectedness of efficient solutions in combinatorial problems has been studied as a prerequisite for neighborhood search~\citep{ehrgott1997connectedness,gorski2011connectedness}; in particular, \citet{gorski2011connectedness} showed that the efficient set need not be connected under single-item swaps even for the bi-objective 0-1 knapsack problem with equal weights and a cardinality constraint, which is exactly the additive case of \Cref{proposition:disconnectedness} for $\PO$. 
By contrast, our results for approval preferences, such as \Cref{theorem:Thiele}, the approval-based counterexamples for \Cref{proposition:disconnectedness}, and the equivalence of $\fPO$ and $\PO$ on the domains in \Cref{sec:restricted}, rely on combinatorial arguments.

\section{Omitted Proofs from Main Body}
This appendix contains all omitted proofs from the main body. They are in the same order as they appear in the paper.

\subsection{General Relationship between $\PO$ and $\fPO$} \label{appendix:axiomatic_relationships}
\fPOimpliesPO*
\begin{proof}
    \begin{enumerate}[leftmargin=*,label=\textit{(\roman*)}]
    \item
    Let $W \in \fPO(I)$.
    Then, there does not exist a fractional committee $\p \in \Probs$ with $u_i(\p) \geq u_i(W)$ for all $i \in N$ and at least one of these inequalities strict.
    Since every size-$k$ committee is contained in $\Probs$, we get that there also does not exist another committee $W'$ with the above conditions.
    Hence, $W \in \PO(I)$. \\
    The other direction, however, does not always hold.
    \Cref{figure:counter_example_PO_fPO_small} shows an instance $I$ where a committee $W$ is $\PO$, but not $\fPO$.
    $W = \{c_1, c_2\}$ is the only committee with $u(W) \geq \mathbf{1}_n$ and thus it is Pareto-optimal.
    If we, however, consider the fractional committee $\p = (0, 0.5, 0.5, 0.5, 0.5)$, we get $u(\p) = (1, 1, 1, 1, 1, 1, 1.5)$.
    Therefore, $W$ is fractional Pareto-dominated by $\p$ and not in $\fPO(I)$. 
    \item
    To prove the claim, we provide two instances.
    One that contains a committee that satisfies $\PO$ but not $\wfPO$, and one with a committee that satisfies $\wfPO$ but not PO.
    The latter can be shown with a very simple example: $C = \{a, b\}$, $k = 1$, and $A_1 = \{a, b\}$, $A_2 = \{b\}$.
    Consider the committee $W = \{a\}$.
    It is easy to see that voter 1 can never strictly increase her utility for any fractional committee and therefore $W \in \wfPO(I)$.
    The committee $W' = \{b\}$, however, strictly increases the utility of the second voter without decreasing the utility of the first one.
    Hence, $W'$ Pareto-dominates $W$ and $W \notin \PO(I)$.\\
    \Cref{figure:counter_example_weak_fPO} provides an instance which shows that Pareto-optimality does not imply weak fractional Pareto-optimality.
    Consider the committee $W=\{c_1,c_2\}$. 
    By construction, every voter approves exactly one candidate in $W$, and hence $u_i(W)=1$ for all $i \in N$. 
    Moreover, no committee of size $2$ can provide utility of at least $1$ to every voter and strictly more than $1$ to some voter. 
    It follows that $W \in \PO(I)$.
    Now consider the fractional committee
    \begin{equation*}
        \p=(0,0,0.4,0.4,0.4,0.4,0.4).
    \end{equation*}
    Since $\sum_{c \in C} \p(c) = 2$, we have $\p \in \Probs$. 
    Furthermore, each voter approves exactly three of the candidates \(c_3,\dots,c_7\). 
    Therefore, for every voter \(i \in N\), $u_i(\p)=3\cdot 0.4=1.2 > 1 = u_i(W)$.
    Hence, $\p$ strictly fractionally dominates $W$, and thus $W \notin \wfPO(I)$. \label{itemitemitem}
    \item 
    Let $W \in \wfPO(I)$.
    Then, there does not exist a fractional committee $\p \in \Probs$ with $u_i(\p) > u_i(W)$ for all $i \in N$.
    Since every size-$k$ committee is contained in $\Probs$, we get that there also does not exist another committee $W'$ with the above conditions.
    Hence, $W \in \wPO(I)$. \\
    To show that the inclusion is strict, we can use \ref{itemitemitem}. There we showed that there exists an instance $I$ and a committee $W$ such that $W \in \PO(I)$ and thus also in $\wPO(I)$ but not in $\wfPO(I)$. 
    \end{enumerate}
\end{proof}

\cloningInvariance*
\begin{proof}
    $\Rightarrow:$
    First assume that $W\in \fPO(I)$.
    Suppose, for contradiction, that there exists some $r\in \N_{>0}$ such that $rW\notin \PO(rI)$.
    Then there exists an integral committee $T$ of size $rk$ in the cloned instance $rI$ such that $u_i^{rI}(T)\ge u_i^{rI}(rW)$ for all $i \in N$, with at least one strict inequality.
    
    For each original candidate $c\in C$, let
    \[
        \p_c \coloneqq \frac{|\{s\in [r]: c^s\in T\}|}{r}.
    \]
    Then $0\le \p_c\le 1$ for all $c\in C$, and since $T$ has size $rk$, we have $\sum_{c\in C} \p_c = k$.
    Thus, $\p \in \Probs$ is a fractional committee in the original instance $I$.
    
    Moreover, because every copy of $c$ has the same valuation as $c$, we get
    \[
        u_i(\p) = \sum_{c\in C} \p_c v_{ic} = \frac{1}{r}u_i^{rI}(T)
    \]
    for every voter $i\in N$.
    Similarly, $u_i(W)=\frac{1}{r}u_i^{rI}(rW)$.
    Hence, $u_i(\p) \ge u_i(W)$ for all $i \in N$, with at least one strict inequality.
    This contradicts $W\in \fPO(I)$.
    Therefore, $rW\in \PO(rI)$ for every $r\in\N_{>0}$.

    $\Leftarrow:$
    Conversely, assume that $W\notin\fPO(I)$. 
    Then there exists a fractional committee $\p$ such that $u_i(\p)\ge u_i(W)$ for all $i \in N$, with at least one strict inequality. 
    Note that one can choose such a witness $\p$ with rational coordinates.
    This follows from the fact that the linear program used in \Cref{theorem:computing_fPO_domination} has only rational parameters, and hence, a rational optimal solution.
    
    Choose $r\in\mathbb{N}_{>0}$ such that $r  \p_c$ is an integer for every $c\in C$. 
    In the cloned instance $rI$, construct a committee $T$ by selecting exactly $r\p_c$ copies of each original candidate $c$. 
    This is feasible because
    \[
        \lvert T \rvert =\sum_{c\in C} r \p_c = rk.
    \]
    For every voter $i\in N$,
    \[
        u_i^{rI}(T)=r u_i(\p)\ge r u_i(W)=u_i^{rI}(rW),
    \]
    with at least one strict inequality. 
    Thus $T$ Pareto-dominates $rW$ in $rI$, so $rW\notin\PO(rI)$.
    
    Therefore, $W\in\fPO(I)$ if and only if $rW\in\PO(rI)$ for every $r\in\N_{>0}$.
    Note that the same steps can be applied when replacing $\fPO$ with $\wfPO$ and $\PO$ with $\wPO$. 
\end{proof}

\cloningInvariancePartTwo*
\begin{proof}
    We first prove that $\fPO$ is cloning-robust. For this, we will use \Cref{thm:cloningInvariance}.
    Assume for the sake of contradiction that the statement does not hold. Then, there exists an instance $I$, a committee $W \in \fPO(I)$, and an
    $r \in \N_{>0}$ such that $rW \notin \fPO(rI)$, and thus, by \Cref{thm:cloningInvariance}, there exists an $s$ such that $rsW \notin \PO(rsI)$. This contradicts the fact that again, due to \Cref{thm:cloningInvariance}, it holds that $tW \in \PO(tI)$  for $t = rs \in \N_{>0}$.  
    
    It remains to show that $\PO$ is not cloning-robust.
    Observe that by \Cref{thm:cloningInvariance} $\fPO$ is the cloning-robust part of $\PO$. Since $\fPO \subsetneq \PO$ by \Cref{proposition:axiomatic_relationships} $\PO$ is not cloning-robust.
    The same argument applies to the two weak efficiency notions.
\end{proof}

\ArbitrarilyBad*
\begin{proof}
    Choose an integer $r$ such that $k/r<\varepsilon$. 
    Let $C_1=\{c_1,\dots,c_k\}$, $C_2=\{c'_1,\dots,c'_{rk}\}$, and set $C=C_1\cup C_2$. 
    Further, let $\mathcal{S}=\{S\subseteq [rk]: |S|=k\}$.

    For every $S\in\mathcal{S}$ and every $b\in[k]$, introduce a voter $x_{S,b}$. Hence, $N=\{x_{S,b}: S\in\mathcal{S},\ b\in[k]\}$.

    The approval set of voter $x_{S,b}$ is defined by
    \[
        A_{x_{S,b}} = \{c_b\} \cup \{c'_j : j\notin S\}.
    \]
    Thus, every voter approves exactly one candidate in $C_1$ and exactly $rk-k$ candidates in $C_2$.
    Let $W \coloneqq C_1$.
    Then every voter approves exactly one candidate in $W$.

    We first show that $W$ is Pareto-optimal. 
    To this end, let $T\neq W$ be any committee of size $k$. 
    There exists some $b\in[k]$ such that $c_b\notin T$. 
    Let $T_2 \coloneqq T\cap C_2$.

    Since $|T_2|\le k$, we can choose a set $S\in\mathcal{S}$ such that $\{j: c'_j\in T_2\}\subseteq S$.
    Consider the voter $x_{S,b}$. 
    This voter does not approve any candidate from $T\cap C_1$, because her only approved candidate in $C_1$ is $c_b$, and $c_b\notin T$. 
    She also does not approve any candidate from $T_2$, because she approves precisely the candidates $c'_j$ with $j\notin S$. 
    Hence, $A_{x_{S,b}}\cap T = \emptyset$.

    Since every voter receives utility $1$ from $W$, no committee $T\neq W$ can Pareto-dominate $W$. 
    Thus, $W\in\PO(I)$.

    Now define a fractional committee $\p$ by assigning weight $ \p(c'_j)=\frac1r$ to every candidate in $C_2$ and weight $0$ to all candidates in $C_1$. 
    This is feasible because $\sum_{j=1}^{rk}\frac1r=k$.

    For every voter $i$, the voter approves exactly $rk-k$ candidates in $C_2$, and therefore
    \[
        \sum_{c\in A_i}\p(c) = \frac{rk-k}{r} = k-\frac{k}{r} > k-\varepsilon.
    \]
    Since $|A_i\cap W|=1$, this gives $\sum_{c\in A_i}\p(c) > k-\varepsilon = (k-\varepsilon)\cdot |A_i\cap W|$.
\end{proof}

The following observation shows that $k$ is the largest possible such factor: unless no voter approves any candidate, no Pareto-optimal committee admits a fractional improvement that scales every voter's utility by at least $k$.

\begin{observation}
    For every approval-based instance $I = (A,k)$ with $k \geq 2$ in which at least one candidate is approved and every committee $W\in\PO(I)$, there is no fractional committee $\p\in\Probs$ such that
    \[
        \sum_{c\in A_i}\p(c) \geq k \cdot u_i(W) \qquad\text{for all } i\in N .
    \]
\end{observation}

\begin{proof}
    Observe that at least one voter approves a candidate in $W$. This follows from the fact that, by assumption, there exists a candidate that is approved at least once, and replacing any candidate in $W$ with her would otherwise be a Pareto improvement. Assume for the sake of contradiction that there exists a fractional committee $\p$ such that $u_j(\p) \geq k u_j(W)$ for every $j \in N$ with $u_j(W) \geq 1$.
    Let $i$ be any voter with $A_i \cap W \neq \emptyset$. Then, since $k \geq u_i(\p) \geq k u_i(W) \geq k$, every candidate in the support of $\p$ is approved by $i$.
    Therefore, it holds that for every candidate $c$ in the support of $\p$ that $u_j(\{c\}) \geq u_j(W)$ for all $j \in N$. Since there are at least $k$ candidates in the support, taking a committee with $k\geq 2$ of them Pareto-dominates $W$. Contradiction.
\end{proof}

We next show that the number of single-candidate swaps between a $\PO$ committee and an $\fPO$ committee can be $k-1$. The proof makes use of \Cref{theorem:fractional_po_equal_weights} which is established in \Cref{section:weight_characterization}.

\WorstDistanceTOfPO*

\begin{proof}
    Let $X=\{x_1,\ldots,x_k\}$ and $Y=\{y_1,\ldots,y_{3k}\}$ be the candidate set, i.e., $C = X \cup Y$.
    For every $b\in[k]$ and every set $S\subseteq Y$ with $\lvert S \rvert = k$, introduce a voter $v_{b,S}$ with approval set $A_{b,S}=\{x_b\}\cup(Y\setminus S)$.
    In words, voter $v_{b, S}$ approves exactly one candidate from $X$, namely $x_b$, and all candidates in $Y$ except those in $S$.

    Consider the committee $W\coloneqq X$.
    We first show that $W\in\PO(I)$. To this end, suppose that a committee $T\in\Committees_k$ Pareto-dominates $W$.
    If $T\neq W$, then $T$ omits some candidate $x_b\in X$.
    Since $\lvert T \rvert = k$, we have $\lvert T\cap Y \rvert \leq k$, and hence, there exists a set $S\subseteq Y$ of size $k$ such that $T\cap Y\subseteq S$.
    Voter $v_{b,S}$ obtains utility $1$ from $W$, as she approves $x_b$. By contrast, she obtains utility $0$ from $T$: the candidate $x_b$ is not selected, she approves no other candidate in $X$, and every candidate in $T\cap Y$ belongs to $S$.
    This contradicts that $T$ Pareto-dominates $W$, and hence, $W \in \PO(I)$.

    We next show that every $\fPO$ committee contains at most one candidate from $W$.
    Let $W^*\in\fPO(I)$, and let $\weight\in\Weights^+$ be a weight vector supporting $W^*$, whose existence follows from \Cref{theorem:fractional_po_equal_weights}.

    For every $b\in[k]$, define
    \[
        B_b \coloneqq \sum_{\substack{S\subseteq Y\\|S|=k}} \weight_{b,S}
    \]
    to be the weight of all voters approving candidate $x_b$. Since $x_b$ is approved exactly by the voters $v_{b,S}$, $s_{x_b}(\weight)=B_b$ and $\sum_{b=1}^k B_b=1$.

    For every auxiliary candidate $y_j\in Y$, define
    \[
        L_j \coloneqq \sum_{\substack{b,S\\y_j\in S}} \weight_{b,S}.
    \]
    Thus, $L_j$ is the total weight of voters who disapprove $y_j$, and therefore, $s_{y_j}(\weight)=1-L_j$.
    Moreover, since every set $S$ has size $k$,
    \[
        \sum_{j=1}^{3k}L_j
        = \sum_{j=1}^{3k} \sum_{\substack{b,S\\y_j\in S}} \weight_{b,S} 
        = \sum_{b,S}\lvert S \rvert\weight_{b,S}
        = k.
    \]

    Suppose, toward a contradiction, that $W^*$ contains at least $q \geq 2$ candidates from $X$.
    Since the total score of all candidates in $X$ is $1$, one of the $q$ selected candidates from $X$, say $x_b$, satisfies $s_{x_b}(\weight)\leq\frac{1}{q}$.
    As $W^*$ is supported by $\weight$, every unselected candidate has score at most that of every selected candidate. Hence, for every unselected $y_j\in Y$,
    \[
        1-L_j
        = s_{y_j}(\weight)
        \leq s_{x_b}(\weight)
        \leq \frac{1}{q},
    \]
    and therefore, $L_j\geq (1-\frac1q)$

    Since $W^*$ contains $q$ candidates from $X$, it contains $k-q$ candidates from $Y$. Thus, among the $3k$ candidates in $Y$, exactly $3k-(k-q)=2k+q$ are unselected. Consequently,
    \[
        k = \sum_{j=1}^{3k}L_j \geq (2k+q) \left(1-\frac1q\right)
        \geq (2k+2) \left(1-\frac12\right)
        = (k+1);
    \]
    a contradiction.
    Hence, $\lvert W^*\cap W \rvert\leq1$ for every $W^* \in \fPO(I)$. Since $W$ and $W^*$ both have size $k$, $d(W,W^*)=k-\lvert W\cap W^*\rvert\geq k-1$.

    It remains to show that the bound is attained. We first observe that, in the approval setting, every candidate contained in a Pareto-optimal committee also belongs to some $\fPO$ committee.

    Indeed, fix $c\in W\in\PO(I)$. If $N_c = N$ we are done immediately. Now assume that $\emptyset \neq N_c\neq N$. Assign weight $(1-\varepsilon)/\lvert N_c\rvert$ to every voter in $N_c$ for $\varepsilon>0$ sufficiently small and  every voter outside $N_c$ obtains weight $\varepsilon/(n-\lvert N_c \rvert)$. Any candidate $d$ with $N_c\nsubseteq N_d$ then has strictly smaller score than $c$, provided $\varepsilon$ is sufficiently small. Hence, only candidates satisfying $N_c\subsetneq N_d$ can have strictly larger score than $c$. There are at most $k-1$ such candidates: otherwise, at least one of them would lie outside $W$, and replacing $c$ by this candidate would Pareto-dominate $W$. Thus, $c$ can be included in a top-$k$ committee under these strictly positive weights, and by \Cref{theorem:fractional_po_equal_weights}, some $\fPO$ committee contains $c$.
    
    If $N_c=\emptyset$, Pareto-optimality of $W$ implies that every candidate outside $W$ is also unapproved; hence, $c$ can again be included in a top-$k$ committee under any strictly positive weight vector.
    
    Applying this observation to $x_1\in W$, there exists $\bar W\in\fPO(I)$ with $x_1\in \bar W$. 
    Together with the previously established lower bound, we conclude that $\min_{\bar W \in \fPO(I)}d(W,\bar W)=k-1$.
\end{proof}

\subsection{Weighted Welfare Maximization and $\fPO$} \label{appendix:characterization}

\weightsToScores*
\begin{proof}
    We rewrite the weighted sum of utilities with respect to the weight vector $\weight$ as follows:
    \begin{equation*}
        \sum_{i \in N}\weight_i u_i(\p) 
        = \sum_{i \in N}\weight_i \sum_{c \in C}v_{ic}\p(c) 
        = \sum_{c \in C}(\sum_{i \in N}\weight_i v_{ic})\p(c) 
        = \sum_{c \in C}s_c(\weight)\p(c).
    \end{equation*}
    Hence, the objective is linear in $\p$, with coefficient $s_c(\weight)$ for each candidate $c \in C$. 
    Now let $W \in \Committees$. 
    Since $\chi^W \in \Probs$, we have $W \in \PO^\weight(I)$ if and only if $\chi^W$ maximizes $\sum_{c \in C}s_c(\weight)\p(c)$ over all $\p \in \Probs$. 
    Because $\sum_{c \in C} \p(c) = k$ and $0 \le \p(c) \le 1$ for all $c \in C$, this linear objective is maximized by setting $\p(c) = 1$ for $k$ candidates with the highest score and $\p(c) = 0$ for all remaining candidates. 
    Thus, $\chi^W$ is optimal if and only if $W$ consists of the $k$ candidates with the highest score $s_c(\weight)$, breaking ties arbitrarily at the $k$-th position.
\end{proof}

\subsubsection{Computability}\label{appendix:computability}

\ComputabilityFPO*
\begin{proof}
    Let $u = u(W)$. 
    We claim that the following linear program can be used to prove the statement: 
    \begin{align*}
        \max&\sum_{i \in N} \varepsilon_i \\
        & \sum_{c \in C} v_{ic}\p(c) \geq u_i + \varepsilon_i \;\;\forall i \in N, \\
        & \sum_{c \in C} \p(c) = k, \quad \varepsilon_i \geq 0, \quad 0 \leq \p(c) \leq 1.
    \end{align*}
    The first set of constraints ensures that the fractional committee $\p$ gives every voter at least the utility obtained from $W$, while each variable $\varepsilon_i$ captures the possible utility gains of voter $i$ beyond $u_i$. 
    Since the incidence vector of $W$ is feasible for this program, the feasible region is non-empty, and the optimal value is at least $0$. 
    Let $(\p^*, \varepsilon^*)$ be an optimal solution.
    We now show that $W \in \fPO(I)$ if and only if the optimum value of the program is equal to $0$. 
    Suppose first that $\sum_{i \in N} \varepsilon_i^* = 0$. 
    Then $\varepsilon_i^* = 0$ for all $i \in N$, and therefore no feasible fractional committee can strictly improve the utility of any voter while weakly improving the utility of all voters. 
    Hence, there is no fractional committee that fractionally Pareto-dominates $W$, and thus $W \in \fPO(I)$.
    Conversely, suppose that $\sum_{i \in N} \varepsilon_i^* > 0$. 
    Then, there exists some voter $j \in N$ with $\varepsilon_j^* > 0$. 
    By feasibility of $(\p^*, \varepsilon^*)$, we obtain
    \begin{align*}
         &u_i(\p^*) = \sum_{c \in C} v_{ic}\p^*(c) \geq u_i(W) \qquad \text{for all } i \in N, \text{ and }\\
         &u_j(\p^*) \geq u_j(W) + \varepsilon_j^* > u_j(W).
    \end{align*}
    Thus, $\p^*$ weakly improves upon $W$ for every voter and strictly improves upon $W$ for voter $j$. It follows that $\p^*$ fractionally Pareto-dominates $W$, so $W \notin \fPO(I)$.\\
    Therefore, solving the above linear program allows us to decide in polynomial time whether $W \in \fPO(I)$. Moreover, whenever $W \notin \fPO(I)$, the optimal solution $\p^*$ yields a witnessing fractional committee.
    Finally, since all coefficients of the LP are rational, it can be solved in polynomial time, and the witnessing fractional committee only contains rational probabilities.

    \bigskip 
    \noindent
    The proof for weak fractional Pareto-optimality works similarly.
    Let $u = u(W)$ again. Consider the following linear program:
    \begin{align*}
        \max \quad & \varepsilon \\
        \text{s.t.} \quad
        & \sum_{c \in C} v_{ic}\p(c) \ge u_i + \varepsilon \qquad \forall i \in N, \\
        & \sum_{c \in C} \p(c) = k, \\
        & 0 \le \p(c) \le 1 \qquad \forall c \in C, \\
        & \varepsilon \ge 0.
    \end{align*}
    The first set of constraints requires that the fractional committee $\p$ improves the utility of every voter by at least $\varepsilon$ over the utility obtained from $W$.
    Since the incidence vector of $W$ together with $\varepsilon = 0$ is feasible, the feasible region is non-empty. Let $(\p^*, \varepsilon^*)$ be an optimal solution.
    
    We claim that $W \in \wfPO(I)$ if and only if the optimum value satisfies
    $\varepsilon^* = 0$.
    Suppose first that $\varepsilon^* = 0$. 
    Then there is no feasible fractional committee that strictly increases each voter's utility. Hence, there is no fractional committee $\p \in \Probs$ such that $u_i(\p) > u_i(W)$ for all $i \in N$.
    Therefore, $W \in \wfPO(I)$.
    Conversely, suppose that $\varepsilon^* > 0$. By feasibility of $(\p^*, \varepsilon^*)$, we obtain
    \begin{equation*}
            u_i(\p^*) = \sum_{c \in C} v_{ic}\p^*(c) \ge u_i(W) + \varepsilon^* > u_i(W) \qquad \text{for all } i \in N.
    \end{equation*}
    Thus, $\p^*$ strictly improves upon $W$ for every voter, so $\p^*$ fractionally Pareto-dominates $W$ in the weak sense. Hence, $W \notin \wfPO(I)$.
    Since all coefficients of the LP are rational, it can be solved in polynomial time.

    \bigskip
    \noindent
    We finally show that, if $W\in\fPO(I)$, then one can compute in polynomial time a weight vector $\weight\in\Weights^+$ such that $W\in\PO^\weight(I)$. 
    The argument for $\wfPO$ is analogous, except that the resulting weight vector is only required to be non-negative.
    
    For $\fPO$, consider the following linear program:
    \begin{align*}
        \max \quad & \eta \\
        \text{s.t.}\quad 
        & s_c(\weight) \ge s_{c'}(\weight) \qquad \text{for all } c\in W,\ c'\in C\setminus W,\\
        & \sum_{i\in N}\weight_i = 1,\\
        & \weight_i \ge \eta \qquad \text{for all } i\in N,\\
        & \eta \ge 0.
    \end{align*}
    
    The first set of constraints ensures that every selected candidate has score at least as large as every unselected candidate. 
    Hence, any feasible solution $\weight$ makes $W$ a top-$k$ committee with respect to the weighted utilitarian objective. 
    The normalization $\sum_{i\in N}\weight_i=1$ ensures that the weights lie in $\Weights$, while the variable $\eta$ maximizes the minimum voter weight.
    
    Since $W\in\fPO(I)$, the weighted-utilitarian characterization of $\fPO$ yields a strictly positive weight vector supporting $W$. 
    After normalization, this vector is feasible with $\eta>0$. 
    Thus, the optimal solution satisfies $\eta^*>0$, and the corresponding weight vector $\weight^*$ belongs to $\Weights^+$. 
    By the score constraints, $W$ maximizes the weighted utilitarian objective induced by $\weight^*$, and therefore $W\in\PO^{\weight^*}(I)$.
    
    For $\wfPO$, one uses the same score constraints but only requires non-negative weights.
    Any feasible solution gives a weight vector $\weight\in\Weights$ such that $W\in\PO^\weight(I)$.
\end{proof}

\subsection{Structural Properties of $\PO$ and $\fPO$} \label{appendix:structural_properties}
\subsubsection{Committee Monotonicity}
\CommitteeMonotonicity*
\begin{proof}
    We prove the statement simultaneously for $\fPO$ and $\wfPO$.
    Fix an instance $I$ and let $k < m$.
    
    By the weighted characterizations of $\fPO$ and $\wfPO$ (see \Cref{theorem:fractional_po_equal_weights}), a committee belongs to the respective solution set if and only if it is supported by some weight vector $\weight$, where $\weight \in \Delta^+$ in the case of $\fPO$ and $\weight \in \Delta$ in the case of $\wfPO$.
    Given such a weight vector $\weight$, \Cref{lemma:weights_to_scores} induces a weak order $\succsim_\weight$ over the candidates according to their $\weight$-scores.
    
    We first show the forward direction.
    Let $W \in S_k(V)$, where $S_k(V)$ denotes either $\fPO_k(V)$ or $\wfPO_k(V)$.
    Then, there exists a weight vector $\weight$ such that $W$ consists of the top-$k$ candidates with respect to $\succsim_\weight$ (with arbitrary tie-breaking at the threshold).
    Let $W'$ be obtained by adding one further candidate of the highest remaining $\weight$-score.
    Then $W'$ consists of the top-$(k+1)$ candidates with respect to the same order $\succsim_\weight$, and hence $W' \in S_{k+1}(V)$.
    By construction, $W \subseteq W'$.
    
    For the reverse direction, let $W' \in S_{k+1}(V)$.
    Again, there exists a weight vector $\weight$ such that $W'$ consists of the top-$(k+1)$ candidates with respect to $\succsim_\weight$.
    Let $W$ be obtained from $W'$ by removing one candidate of the lowest $\weight$-score in $W'$.
    Then $W$ consists of the top-$k$ candidates with respect to the same order $\succsim_\weight$, and therefore $W \in S_k(V)$.
    Clearly, $W \subseteq W'$.
    
    Thus, for every instance $I$ and every $k < m$, every committee in $S_k(V)$ is contained in some committee in $S_{k+1}(V)$, and every committee in $S_{k+1}(V)$ contains some committee in $S_k(V)$.
    Hence, both $\fPO$ and $\wfPO$ satisfy committee monotonicity.
\end{proof}

\violationMonotonicity*

We first show a simple illustrative counterexample for the additive setting. Here, tuning the valuations allows for a simple two-voter instance.
\begin{figure}[!ht]
    \centering
\begin{tikzpicture}[
        scale=0.37,
        every node/.style={font=\scriptsize},
        single/.style={circle, fill=red!70, inner sep=1.8pt},
        bad/.style={circle, fill=green!60!black, inner sep=1.8pt},
        dom/.style={circle, fill=blue!70, inner sep=1.8pt},
        axis/.style={->, thick},
        domarrow/.style={-{Stealth[length=1.5mm,width=1.5mm]}, very thick, gray!50,
                    shorten >=1pt, shorten <=1pt}
    ]
    
    \draw[axis] (0,0) -- (13,0) node[below] {$u_1$};
    \draw[axis] (0,0) -- (0,13) node[left] {$u_2$};
    
    \foreach \x in {2,4,6,8,10,12}
        \draw (\x,0.15) -- (\x,-0.15) node[below=2pt] {\tiny \x};
    \foreach \y in {2,4,6,8,10,12}
        \draw (0.15,\y) -- (-0.15,\y) node[left=2pt] {\tiny \y};
    
    \node[single, label={[below= 0.5pt]{$\{c_1\}$}}] (c1) at (1.5,1.5) {};
    \node[single, label={[below= 0.5pt]{$\{c_2\}$}}] (c23) at (5,1) {};
    \node[single, label={[below= 0.5pt]{$\{c_4\}$}}] (c45) at (1,5) {};
    \node[single, label={[above left=-5pt]{$\{c_6\}$}}] (c6) at (7,0) {};
    \node[single, label={[below right=-1pt]{$\{c_7\}$}}] (c7) at (0,7) {};
    
    \node[bad, label={[below= 0.5pt]{$\{c_1,c_6\}$}}] (c16) at (8.5,1.5) {};
    \node[bad, label={[below left=-1pt]{$\{c_1,c_2\}$}}] (c12) at (6.5,2.5) {};
    \node[bad, label={[below= 0.5pt]{$\{c_1,c_4\}$}}] (c14) at (2.5,6.5) {};
    \node[bad, label={[below= 0.5pt]{$\{c_1,c_7\}$}}] (c17) at (1.5,8.5) {};
    
    \node[dom, label={[below right=-1pt]{$\{c_2,c_3\}$}}] (c23pair) at (10,2) {};
    \node[dom, label={[above=-1pt]{$\{c_4,c_5\}$}}] (c45pair) at (2,10) {};
    \node[dom, label={[above right=-1pt]{$\{c_6,c_7\}$}}] (c67) at (7,7) {};
    
    \draw[domarrow] (c23pair) -- (c16);
    \draw[domarrow] (c45pair) -- (c17);
    \draw[domarrow] (c67) -- (c12);
    \draw[domarrow] (c67) -- (c14);
    
    \fill (0,0) circle (1pt);
    
\end{tikzpicture}
    \caption{$\PO$ violates committee monotonicity in the additive-utility setting.} 
    \label{fig:po_not_monotone}
\end{figure}

Consider two voters and seven candidates with valuations \begin{align*}
    v_1 &=(3/2,5,5,1,1,7,0), \text{ and} \\
    v_2 &=(3/2,1,1,5,5,0,7).
\end{align*}
For $k=1$, the committee $W = \{c_1\}$ is Pareto-optimal, as it is not dominated by any other size-$1$ committee.
For $k=2$, however, every extension of $W$ even fails weak Pareto-optimality: for each candidate $c \in \{c_2, \dots, c_7\}$, the committee $W \cup \{c\}$ is strictly Pareto-dominated.
\Cref{fig:po_not_monotone} shows the utility vectors of these extensions in green; each is strictly dominated by a size-$2$ committee drawn in blue, as indicated by the arrows.
Hence, both $\PO$ and $\wPO$ violate committee monotonicity.

\begin{figure}[t]
\begin{minipage}[t]{0.48\columnwidth}
    \centering
    \scalebox{0.8}{
        \begin{tikzpicture}[yscale=0.6,xscale=0.6,voter/.style={anchor=south}]
          \draw[fill=gray!10] (0, 8) rectangle (1, 9);
          \draw[fill=gray!10] (1, 8) rectangle (2, 9);
          \draw[fill=gray!10] (2, 8) rectangle (3, 9);
          \draw[fill=gray!10] (3, 8) rectangle (4, 9);
          \draw[fill=gray!10] (4, 8) rectangle (5, 9);
          \draw[fill=gray!10] (5, 8) rectangle (6, 9);
          \draw[fill=gray!10] (6, 8) rectangle (7, 9);
          \draw[fill=blue!40] (7, 8) rectangle (8, 9);
          \draw[fill=blue!40] (8, 8) rectangle (9, 9);
          \draw[fill=blue!40] (9, 8) rectangle (10, 9);
          \draw[fill=blue!40] (10, 8) rectangle (11, 9);
          \draw[fill=blue!40] (11, 8) rectangle (12, 9);
        
          \draw[fill=blue!40] (0, 7) rectangle (1, 8);
          \draw[fill=blue!40] (1, 7) rectangle (2, 8);
          \draw[fill=blue!40] (2, 7) rectangle (3, 8);
          \draw[fill=blue!40] (3, 7) rectangle (4, 8);
          \draw[fill=blue!40] (4, 7) rectangle (5, 8);
          \draw[fill=blue!40] (5, 7) rectangle (6, 8);
          \draw[fill=blue!40] (6, 7) rectangle (7, 8);
          \draw[fill=gray!10] (7, 7) rectangle (8, 8);
          \draw[fill=gray!10] (8, 7) rectangle (9, 8);
          \draw[fill=gray!10] (9, 7) rectangle (10, 8);
          \draw[fill=gray!10] (10, 7) rectangle (11, 8);
          \draw[fill=gray!10] (11, 7) rectangle (12, 8);
        
          \draw[fill=blue!40] (0, 6) rectangle (1, 7);
          \draw[fill=gray!10] (1, 6) rectangle (2, 7);
          \draw[fill=gray!10] (2, 6) rectangle (3, 7);
          \draw[fill=gray!10] (3, 6) rectangle (4, 7);
          \draw[fill=blue!40] (4, 6) rectangle (5, 7);
          \draw[fill=blue!40] (5, 6) rectangle (6, 7);
          \draw[fill=gray!10] (6, 6) rectangle (7, 7);
          \draw[fill=blue!40] (7, 6) rectangle (8, 7);
          \draw[fill=gray!10] (8, 6) rectangle (9, 7);
          \draw[fill=blue!40] (9, 6) rectangle (10, 7);
          \draw[fill=blue!40] (10, 6) rectangle (11, 7);
          \draw[fill=blue!40] (11, 6) rectangle (12, 7);
        
          \draw[fill=gray!10] (0, 5) rectangle (1, 6);
          \draw[fill=blue!40] (1, 5) rectangle (2, 6);
          \draw[fill=blue!40] (2, 5) rectangle (3, 6);
          \draw[fill=gray!10] (3, 5) rectangle (4, 6);
          \draw[fill=blue!40] (4, 5) rectangle (5, 6);
          \draw[fill=blue!40] (5, 5) rectangle (6, 6);
          \draw[fill=blue!40] (6, 5) rectangle (7, 6);
          \draw[fill=gray!10] (7, 5) rectangle (8, 6);
          \draw[fill=blue!40] (8, 5) rectangle (9, 6);
          \draw[fill=blue!40] (9, 5) rectangle (10, 6);
          \draw[fill=gray!10] (10, 5) rectangle (11, 6);
          \draw[fill=gray!10] (11, 5) rectangle (12, 6);
        
          \draw[fill=gray!10] (0, 4) rectangle (1, 5);
          \draw[fill=blue!40] (1, 4) rectangle (2, 5);
          \draw[fill=gray!10] (2, 4) rectangle (3, 5);
          \draw[fill=blue!40] (3, 4) rectangle (4, 5);
          \draw[fill=gray!10] (4, 4) rectangle (5, 5);
          \draw[fill=gray!10] (5, 4) rectangle (6, 5);
          \draw[fill=blue!40] (6, 4) rectangle (7, 5);
          \draw[fill=gray!10] (7, 4) rectangle (8, 5);
          \draw[fill=blue!40] (8, 4) rectangle (9, 5);
          \draw[fill=blue!40] (9, 4) rectangle (10, 5);
          \draw[fill=blue!40] (10, 4) rectangle (11, 5);
          \draw[fill=blue!40] (11, 4) rectangle (12, 5);
        
          \draw[fill=gray!10] (0, 3) rectangle (1, 4);
          \draw[fill=gray!10] (1, 3) rectangle (2, 4);
          \draw[fill=blue!40] (2, 3) rectangle (3, 4);
          \draw[fill=gray!10] (3, 3) rectangle (4, 4);
          \draw[fill=blue!40] (4, 3) rectangle (5, 4);
          \draw[fill=blue!40] (5, 3) rectangle (6, 4);
          \draw[fill=gray!10] (6, 3) rectangle (7, 4);
          \draw[fill=blue!40] (7, 3) rectangle (8, 4);
          \draw[fill=blue!40] (8, 3) rectangle (9, 4);
          \draw[fill=blue!40] (9, 3) rectangle (10, 4);
          \draw[fill=blue!40] (10, 3) rectangle (11, 4);
          \draw[fill=gray!10] (11, 3) rectangle (12, 4);
        
          \draw[fill=blue!40] (0, 2) rectangle (1, 3);
          \draw[fill=gray!10] (1, 2) rectangle (2, 3);
          \draw[fill=blue!40] (2, 2) rectangle (3, 3);
          \draw[fill=blue!40] (3, 2) rectangle (4, 3);
          \draw[fill=blue!40] (4, 2) rectangle (5, 3);
          \draw[fill=gray!10] (5, 2) rectangle (6, 3);
          \draw[fill=blue!40] (6, 2) rectangle (7, 3);
          \draw[fill=blue!40] (7, 2) rectangle (8, 3);
          \draw[fill=gray!10] (8, 2) rectangle (9, 3);
          \draw[fill=gray!10] (9, 2) rectangle (10, 3);
          \draw[fill=gray!10] (10, 2) rectangle (11, 3);
          \draw[fill=blue!40] (11, 2) rectangle (12, 3);
        
          \draw[fill=gray!10] (0, 1) rectangle (1, 2);
          \draw[fill=gray!10] (1, 1) rectangle (2, 2);
          \draw[fill=blue!40] (2, 1) rectangle (3, 2);
          \draw[fill=blue!40] (3, 1) rectangle (4, 2);
          \draw[fill=gray!10] (4, 1) rectangle (5, 2);
          \draw[fill=gray!10] (5, 1) rectangle (6, 2);
          \draw[fill=blue!40] (6, 1) rectangle (7, 2);
          \draw[fill=blue!40] (7, 1) rectangle (8, 2);
          \draw[fill=blue!40] (8, 1) rectangle (9, 2);
          \draw[fill=gray!10] (9, 1) rectangle (10, 2);
          \draw[fill=blue!40] (10, 1) rectangle (11, 2);
          \draw[fill=blue!40] (11, 1) rectangle (12, 2);
        
          \draw[fill=blue!40] (0, 0) rectangle (1, 1);
          \draw[fill=blue!40] (1, 0) rectangle (2, 1);
          \draw[fill=gray!10] (2, 0) rectangle (3, 1);
          \draw[fill=blue!40] (3, 0) rectangle (4, 1);
          \draw[fill=gray!10] (4, 0) rectangle (5, 1);
          \draw[fill=blue!40] (5, 0) rectangle (6, 1);
          \draw[fill=gray!10] (6, 0) rectangle (7, 1);
          \draw[fill=gray!10] (7, 0) rectangle (8, 1);
          \draw[fill=gray!10] (8, 0) rectangle (9, 1);
          \draw[fill=blue!40] (9, 0) rectangle (10, 1);
          \draw[fill=blue!40] (10, 0) rectangle (11, 1);
          \draw[fill=blue!40] (11, 0) rectangle (12, 1);
        
          \node[voter] at (0.5, -1) {$1$};
          \node[voter] at (1.5, -1) {$2$};
          \node[voter] at (2.5, -1) {$3$};
          \node[voter] at (3.5, -1) {$4$};
          \node[voter] at (4.5, -1) {$5$};
          \node[voter] at (5.5, -1) {$6$};
          \node[voter] at (6.5, -1) {$7$};
          \node[voter] at (7.5, -1) {$8$};
          \node[voter] at (8.5, -1) {$9$};
          \node[voter] at (9.5, -1) {$10$};
          \node[voter] at (10.5, -1) {$11$};
          \node[voter] at (11.5, -1) {$12$};
        
          \node[anchor=east] at (-0.3, 8.5) {$c_{1}$};
          \node[anchor=east] at (-0.3, 7.5) {$c_{2}$};
          \node[anchor=east] at (-0.3, 6.5) {$c_{3}$};
          \node[anchor=east] at (-0.3, 5.5) {$c_{4}$};
          \node[anchor=east] at (-0.3, 4.5) {$c_{5}$};
          \node[anchor=east] at (-0.3, 3.5) {$c_{6}$};
          \node[anchor=east] at (-0.3, 2.5) {$c_{7}$};
          \node[anchor=east] at (-0.3, 1.5) {$c_{8}$};
          \node[anchor=east] at (-0.3, 0.5) {$c_{9}$};
        \end{tikzpicture}
    }
    \subcaption{Approval-based instance. The committee $W=\{c_1,c_2\}$ for $k=2$ satisfies $\PO$, but cannot be included in any $\PO$ committee for $k=3$.
    }
\label{figure:counter_example_PO_committee_monotonicity}
    \end{minipage}
    \hfill
    \begin{minipage}[t]{0.48\columnwidth}
    \centering
    \scalebox{0.8}{
\begin{tikzpicture}[yscale=0.6,xscale=0.6,voter/.style={anchor=south}]
  \draw[fill=gray!10] (0, 9) rectangle (1, 10);
  \draw[fill=gray!10] (1, 9) rectangle (2, 10);
  \draw[fill=gray!10] (2, 9) rectangle (3, 10);
  \draw[fill=gray!10] (3, 9) rectangle (4, 10);
  \draw[fill=gray!10] (4, 9) rectangle (5, 10);
  \draw[fill=gray!10] (5, 9) rectangle (6, 10);
  \draw[fill=gray!10] (6, 9) rectangle (7, 10);
  \draw[fill=gray!10] (7, 9) rectangle (8, 10);
  \draw[fill=gray!10] (8, 9) rectangle (9, 10);
  \draw[fill=gray!10] (9, 9) rectangle (10, 10);
  \draw[fill=gray!10] (10, 9) rectangle (11, 10);
  \draw[fill=gray!10] (11, 9) rectangle (12, 10);

  \draw[fill=gray!10] (0, 8) rectangle (1, 9);
  \draw[fill=gray!10] (1, 8) rectangle (2, 9);
  \draw[fill=gray!10] (2, 8) rectangle (3, 9);
  \draw[fill=gray!10] (3, 8) rectangle (4, 9);
  \draw[fill=gray!10] (4, 8) rectangle (5, 9);
  \draw[fill=gray!10] (5, 8) rectangle (6, 9);
  \draw[fill=gray!10] (6, 8) rectangle (7, 9);
  \draw[fill=gray!10] (7, 8) rectangle (8, 9);
  \draw[fill=gray!10] (8, 8) rectangle (9, 9);
  \draw[fill=gray!10] (9, 8) rectangle (10, 9);
  \draw[fill=gray!10] (10, 8) rectangle (11, 9);
  \draw[fill=gray!10] (11, 8) rectangle (12, 9);

  \draw[fill=blue!40] (0, 7) rectangle (1, 8);
  \draw[fill=blue!40] (1, 7) rectangle (2, 8);
  \draw[fill=gray!10] (2, 7) rectangle (3, 8);
  \draw[fill=gray!10] (3, 7) rectangle (4, 8);
  \draw[fill=blue!40] (4, 7) rectangle (5, 8);
  \draw[fill=blue!40] (5, 7) rectangle (6, 8);
  \draw[fill=blue!40] (6, 7) rectangle (7, 8);
  \draw[fill=blue!40] (7, 7) rectangle (8, 8);
  \draw[fill=blue!40] (8, 7) rectangle (9, 8);
  \draw[fill=gray!10] (9, 7) rectangle (10, 8);
  \draw[fill=gray!10] (10, 7) rectangle (11, 8);
  \draw[fill=gray!10] (11, 7) rectangle (12, 8);

  \draw[fill=blue!40] (0, 6) rectangle (1, 7);
  \draw[fill=gray!10] (1, 6) rectangle (2, 7);
  \draw[fill=blue!40] (2, 6) rectangle (3, 7);
  \draw[fill=gray!10] (3, 6) rectangle (4, 7);
  \draw[fill=blue!40] (4, 6) rectangle (5, 7);
  \draw[fill=blue!40] (5, 6) rectangle (6, 7);
  \draw[fill=blue!40] (6, 6) rectangle (7, 7);
  \draw[fill=blue!40] (7, 6) rectangle (8, 7);
  \draw[fill=gray!10] (8, 6) rectangle (9, 7);
  \draw[fill=blue!40] (9, 6) rectangle (10, 7);
  \draw[fill=gray!10] (10, 6) rectangle (11, 7);
  \draw[fill=gray!10] (11, 6) rectangle (12, 7);

  \draw[fill=gray!10] (0, 5) rectangle (1, 6);
  \draw[fill=gray!10] (1, 5) rectangle (2, 6);
  \draw[fill=blue!40] (2, 5) rectangle (3, 6);
  \draw[fill=blue!40] (3, 5) rectangle (4, 6);
  \draw[fill=blue!40] (4, 5) rectangle (5, 6);
  \draw[fill=blue!40] (5, 5) rectangle (6, 6);
  \draw[fill=gray!10] (6, 5) rectangle (7, 6);
  \draw[fill=gray!10] (7, 5) rectangle (8, 6);
  \draw[fill=blue!40] (8, 5) rectangle (9, 6);
  \draw[fill=gray!10] (9, 5) rectangle (10, 6);
  \draw[fill=blue!40] (10, 5) rectangle (11, 6);
  \draw[fill=blue!40] (11, 5) rectangle (12, 6);

  \draw[fill=blue!40] (0, 4) rectangle (1, 5);
  \draw[fill=blue!40] (1, 4) rectangle (2, 5);
  \draw[fill=gray!10] (2, 4) rectangle (3, 5);
  \draw[fill=gray!10] (3, 4) rectangle (4, 5);
  \draw[fill=gray!10] (4, 4) rectangle (5, 5);
  \draw[fill=gray!10] (5, 4) rectangle (6, 5);
  \draw[fill=blue!40] (6, 4) rectangle (7, 5);
  \draw[fill=gray!10] (7, 4) rectangle (8, 5);
  \draw[fill=blue!40] (8, 4) rectangle (9, 5);
  \draw[fill=blue!40] (9, 4) rectangle (10, 5);
  \draw[fill=blue!40] (10, 4) rectangle (11, 5);
  \draw[fill=blue!40] (11, 4) rectangle (12, 5);

  \draw[fill=gray!10] (0, 3) rectangle (1, 4);
  \draw[fill=gray!10] (1, 3) rectangle (2, 4);
  \draw[fill=blue!40] (2, 3) rectangle (3, 4);
  \draw[fill=blue!40] (3, 3) rectangle (4, 4);
  \draw[fill=blue!40] (4, 3) rectangle (5, 4);
  \draw[fill=gray!10] (5, 3) rectangle (6, 4);
  \draw[fill=blue!40] (6, 3) rectangle (7, 4);
  \draw[fill=blue!40] (7, 3) rectangle (8, 4);
  \draw[fill=gray!10] (8, 3) rectangle (9, 4);
  \draw[fill=blue!40] (9, 3) rectangle (10, 4);
  \draw[fill=blue!40] (10, 3) rectangle (11, 4);
  \draw[fill=gray!10] (11, 3) rectangle (12, 4);

  \draw[fill=gray!10] (0, 2) rectangle (1, 3);
  \draw[fill=blue!40] (1, 2) rectangle (2, 3);
  \draw[fill=gray!10] (2, 2) rectangle (3, 3);
  \draw[fill=blue!40] (3, 2) rectangle (4, 3);
  \draw[fill=gray!10] (4, 2) rectangle (5, 3);
  \draw[fill=gray!10] (5, 2) rectangle (6, 3);
  \draw[fill=blue!40] (6, 2) rectangle (7, 3);
  \draw[fill=blue!40] (7, 2) rectangle (8, 3);
  \draw[fill=gray!10] (8, 2) rectangle (9, 3);
  \draw[fill=blue!40] (9, 2) rectangle (10, 3);
  \draw[fill=blue!40] (10, 2) rectangle (11, 3);
  \draw[fill=blue!40] (11, 2) rectangle (12, 3);

  \draw[fill=gray!10] (0, 1) rectangle (1, 2);
  \draw[fill=blue!40] (1, 1) rectangle (2, 2);
  \draw[fill=gray!10] (2, 1) rectangle (3, 2);
  \draw[fill=blue!40] (3, 1) rectangle (4, 2);
  \draw[fill=blue!40] (4, 1) rectangle (5, 2);
  \draw[fill=blue!40] (5, 1) rectangle (6, 2);
  \draw[fill=gray!10] (6, 1) rectangle (7, 2);
  \draw[fill=blue!40] (7, 1) rectangle (8, 2);
  \draw[fill=blue!40] (8, 1) rectangle (9, 2);
  \draw[fill=gray!10] (9, 1) rectangle (10, 2);
  \draw[fill=gray!10] (10, 1) rectangle (11, 2);
  \draw[fill=blue!40] (11, 1) rectangle (12, 2);

  \draw[fill=blue!40] (0, 0) rectangle (1, 1);
  \draw[fill=gray!10] (1, 0) rectangle (2, 1);
  \draw[fill=blue!40] (2, 0) rectangle (3, 1);
  \draw[fill=gray!10] (3, 0) rectangle (4, 1);
  \draw[fill=gray!10] (4, 0) rectangle (5, 1);
  \draw[fill=blue!40] (5, 0) rectangle (6, 1);
  \draw[fill=gray!10] (6, 0) rectangle (7, 1);
  \draw[fill=gray!10] (7, 0) rectangle (8, 1);
  \draw[fill=blue!40] (8, 0) rectangle (9, 1);
  \draw[fill=blue!40] (9, 0) rectangle (10, 1);
  \draw[fill=blue!40] (10, 0) rectangle (11, 1);
  \draw[fill=blue!40] (11, 0) rectangle (12, 1);

  \node[voter] at (0.5, -1) {$1$};
  \node[voter] at (1.5, -1) {$2$};
  \node[voter] at (2.5, -1) {$3$};
  \node[voter] at (3.5, -1) {$4$};
  \node[voter] at (4.5, -1) {$5$};
  \node[voter] at (5.5, -1) {$6$};
  \node[voter] at (6.5, -1) {$7$};
  \node[voter] at (7.5, -1) {$8$};
  \node[voter] at (8.5, -1) {$9$};
  \node[voter] at (9.5, -1) {$10$};
  \node[voter] at (10.5, -1) {$11$};
  \node[voter] at (11.5, -1) {$12$};

  \node[anchor=east] at (-0.3, 9.5) {$c_{1}$};
  \node[anchor=east] at (-0.3, 8.5) {$c_{2}$};
  \node[anchor=east] at (-0.3, 7.5) {$c_{3}$};
  \node[anchor=east] at (-0.3, 6.5) {$c_{4}$};
  \node[anchor=east] at (-0.3, 5.5) {$c_{5}$};
  \node[anchor=east] at (-0.3, 4.5) {$c_{6}$};
  \node[anchor=east] at (-0.3, 3.5) {$c_{7}$};
  \node[anchor=east] at (-0.3, 2.5) {$c_{8}$};
  \node[anchor=east] at (-0.3, 1.5) {$c_{9}$};
  \node[anchor=east] at (-0.3, 0.5) {$c_{10}$};
\end{tikzpicture}
    }
    \subcaption{Approval-based instance. The committee $W=\{c_1,c_2\}$ for $k=2$ satisfies $\wPO$, but cannot be included in any $\wPO$ committee for $k=3$.}
    \label{figure:counter_example_wPO_committee_monotonicity}
  \end{minipage}
  \caption{Counterexamples for $\PO$ and $\wPO$ violating committee monotonicity in the approval-based setting.}
  \label{fig:counterexamples_committee_monotonicity}
\end{figure}

Showing that committee monotonicity fails for $\PO$ and $\wPO$ in the approval-based setting is not as straightforward.
We obtain the counterexamples depicted in \Cref{fig:counterexamples_committee_monotonicity} using a SAT-based search.
Fixing parameters $m$ and $k$, we can encode Pareto-optimality into a SAT instance since Pareto domination requires that \emph{for all} voters the number of approved candidates does not decrease, while \emph{there exists} a voter with strictly more approved candidates in the committee.

For each possible ballot \emph{type} $\emptyset \ne A \subseteq C = [m]$, we assign a binary variable $x_A \in \{0,1\}$ which encodes whether there exists a voter with approval set $A$. Note that there are $2^m - 1$ of them. We include the clause
\[
    \bigvee_{\emptyset \ne A\subseteq C} x_A
\]
to ensure that the election instance has at least one voter.
Further, without loss of generality, we encode $W = [k]$ to be our desired $\PO$ committee.
Since $W$ is $\PO$ if and only if for each possible deviation committee $W'$, either there exists a voter who strictly prefers $W$ over $W'$, or there exists no voter who strictly prefers $W'$ over $W$. This is encoded as follows:
\[
    \bigwedge_{W' \ne W} \left[ \left( \bigvee_{A: W \succ_A W'}x_A \right) \vee \left( \bigwedge_{A: W \prec_A W' }\neg x_A \right)   \right],
\]
where $W \succ_A W'$ denotes that a voter with approval set $A$ strictly prefers the committee $W$ over $W'$.
We can encode this into an equivalent CNF using Tseitin's encoding.
Introduce a binary variable $y_{W'} \in \{0,1\}$ for each possible deviation committee $W' \ne W$ which encodes whether there exists a voter that strictly prefers $W'$ over $W$. 
Then, we can transform the formula to the following CNF:
\[
    \bigwedge_{W' \ne W} \left[ \left(\bigvee_{A: W \succ_A W'} (x_A \vee \neg y_{W'}) \right) \wedge \left( \bigwedge_{A: W \prec_A W' } (\neg x_A \vee y_{W'}) \right)   \right].
\]

Finally, we encode that no committee of size $k+1$ including $W$ satisfies $\PO$.
For each $c \in C \setminus W$ and a possible deviation $W''$ of size $k+1$, we define a binary variable $z_{c,W''} \in \{0,1\}$ that encodes whether $W''$ Pareto dominates $W \cup \{c\}$. Without loss of generality, since $W$ is Pareto-optimal, this can only happen when $W''$ does not include $c$. 
Hence, for $D(c) = \{W'' \in \Committees_{k+1}: c \notin W''\}$, we include the following CNF:
$$\bigwedge_{c \in C\setminus W} \bigvee_{W'' \in D(c)} z_{c,W''}.$$
What remains is encoding the semantics of the binary variables $z_{c, W''}$. It is sufficient that we only encode $z_{c,W''} = 1$ implies $W''$ Pareto dominating $W \cup \{c\}$ (but not the other way around). Hence, we get:
$$\bigwedge_{c \in C \setminus W}\bigwedge_{W'' \in D(c)}\left[\left( \neg z_{c, W''} \vee \bigvee_{A: W\cup\{c\} \prec_A W''} x_A \right) \wedge \bigwedge_{A: W\cup\{c\} \succ_A W''} (\neg z_{c,W''} \vee \neg x_A) \right].$$

Combining these formulas concludes the SAT encoding. 
For given $(m,k)$ (with $k \le m-2$ as the $k=m-1$ case is trivially contained in the only committee $C$ of size $m$), there exists a $\PO$ committee $W$ of size $k$ not contained in any $\PO$ committee of size $k+1$ if and only if the SAT instance is satisfiable.
The smallest $m$ for which the SAT is satisfiable is $m=9$, yielding the above counterexample (see \Cref{figure:counter_example_PO_committee_monotonicity}). 

The idea for $\wPO$ is analogous.

For every possible deviation committee $W'\ne W$, there must exist a voter who does not strictly prefer $W'$ over $W$. The following CNF encodes this:
\[
    \bigwedge_{W'\neq W} \bigvee_{A:\ W'\not\succ_A W} x_A.
\]

To encode that no $\wPO$ committee of size $k+1$ includes $W$, we again, for each $c\in C\setminus W$ and each $W''\in D(c)=\{W''\in \Committees_{k+1}: c\notin W''\}$, introduce a binary variable $z_{c,W''}\in\{0,1\}$, which encodes that $W''$ blocks $W\cup \{c\}$ from being $\wPO$.
For every $c\in C\setminus W$, we require at least one such blocker, i.e., 
\[
    \bigwedge_{c\in C\setminus W} \left(\bigvee_{W''\in D(c)} z_{c,W''}\right).
\]

It remains to encode the logic that $z_{c,W''} = 1$ implies $W''$ (strongly) Pareto dominating $W \cup \{c\}$:
\[
    \bigwedge_{c\in C\setminus W} \bigwedge_{W''\in D(c)} \bigwedge_{A:\ W''\not\succ_A W\cup\{c\}} \left( \neg z_{c,W''}\vee \neg x_A \right).
\]

This concludes the SAT encoding for $\wPO$. 
For given $(m,k)$, there exists a $\wPO$ committee $W$ of size $k$ not contained in any $\wPO$ committee of size $k+1$ if and only if the SAT instance is satisfiable. The smallest $m$ for which the SAT is satisfiable is $m=10$, yielding the above counterexample (see \Cref{figure:counter_example_wPO_committee_monotonicity}). 

Now recall that the above SAT formulations are unsatisfiable for $\PO$ ($\wPO$) if $m\le 8$ ($m \le 9$, resp.). These imply that committee monotonicity holds for $\PO$ ($\wPO$) if $m \le 8$ ($m \le 9$, resp.), with the help of the following lemma.

Given an approval profile $A$, define its \emph{dual profile} $A^*$ by $A_i^* = C \setminus A_i$ for each voter $i$. For a committee $W \in \Committees_k$, write $W^* = C \setminus W \in \Committees_{m-k}$.

\begin{lemma}\label{lem:approval-dual}
    For an approval profile $A$ with $m$ candidates, $W \in \PO_k(A)$ if and only if $W^* \in \PO_{m-k}(A^*)$. Similarly, $W \in \wPO_k(A)$ if and only if $W^* \in \wPO_{m-k}(A^*)$.
\end{lemma}
\begin{proof}
    For a voter $i$, since $|A_i \cap W| + |A_i \cap W^*| = |A_i|$, we see that $|A_i^* \cap W^*| = |W^*| - |A_i \cap W^*| = (m-k) - |A_i| + |A_i \cap W|$. Hence, for a committee $W' \in \Committees_k$, $|A_i \cap W| \ge |A_i \cap W'|$ holds if and only if $|A_i^* \cap W^*| \ge |A_i^* \cap (W')^*|$ holds, and likewise for strict inequalities. Hence, the claims follow.
\end{proof}

\committeeMonotoneM*
\begin{proof}
    Since the above SAT is unsatisfiable for any $k$ if $m \le 8$, the first condition of committee monotonicity holds.

    To show that the second condition also holds, assume to the contrary that for some approval profile $A$ with $m \le 8$ and $k$, there exists a committee $W \in \PO_{k+1}(A)$ that contains no committee in $\PO_k(A)$. That is, for any $c \in W$, $W \setminus \{c\} \notin \PO_k(A)$. Note that since $C$ contains all committees, $k \le m-2$.

    Then, consider the dual profile $A^*$. By \Cref{lem:approval-dual}, $W^* \in \PO_{m-k-1}(A^*)$. However, for any $c \in C\setminus W^* = W$, since $W^* \cup \{c\} = (W \setminus \{c\})^*$ and $W \setminus \{c\} \notin \PO_k(A)$, we have $W^* \cup \{c\} \notin \PO_{m-k}(A^*)$. This contradicts the first condition of committee monotonicity, since $A^*$ also has $m \le 8$ candidates.

    An analogous proof holds for $\wPO$.
\end{proof}

\subsubsection{Connectedness}\label{appendix:connectedness}

To study the size and diameter of the connected component, we relate our problem to the \textit{k-level problem} from discrete geometry \citep{Dey98}.

\textbf{The k-level problem.} 
Let $\mathcal{L}$ be an arrangement of $m$ lines in $\R^2$. For a vertical line at position $x$, let $\ell_1(x), \ldots, \ell_m(x)$ denote the $m$ lines ordered by their $y$-coordinate at $x$ (highest to lowest). The \textit{$k$-level} is the piecewise-linear curve formed by the $k$-th highest line at each x-coordinate, including all vertices where the identity of the $k$-th line changes. Its complexity is defined as the number of vertices (breakpoints) on this curve. The k-level problem asks for an upper bound on this complexity as a function of $m$ and $k$.

\diameter*
\begin{proof}
    For $k=m$, the argument is trivial, so assume $k < m$.
    Let $W, W' \in \fPO(I)$ [$\wfPO(I)$] be supported by $\weight, \weight' \in \Weights^+$ [$\Weights$], which exist by \Cref{theorem:fractional_po_equal_weights}, and let $\gamma(\lambda) \coloneqq (1-\lambda)\weight + \lambda\weight'$ for $\lambda\in[0,1]$. Then $\gamma(\lambda)\in\Weights^+$ [$\Weights$], so every committee supported by some $\gamma(\lambda)$ lies in $\fPO(I)$ [$\wfPO(I)$]. By \Cref{lemma:weights_to_scores}, these are exactly the top-$k$ sets with respect to the scores $s_c(\lambda)\coloneqq\sum_{i\in N}\gamma(\lambda)_i v_{ic}$, each of which is linear in $\lambda$. We follow the top-$k$ sets as $\lambda$ increases from $0$ to $1$.
    For simplicity, first assume that the score lines are in general position, that is, no two are parallel or coincide, and no three pass through a common point.
    Then, along the connecting path, each change in the top-$k$ set corresponds to a single-candidate swap. 
    Each such change occurs when the lines of the $k$-th and $(k+1)$-st highest-scoring candidates cross, and hence, at a vertex of the $k$-level. Thus, the number of swaps is at most the complexity of the $k$-level.
    \citet{Dey98} proved that the complexity of the $k$-level in an arrangement of $m$ lines in the plane lies in $\mathcal{O}(m \cdot k^{1/3})$.  
    Therefore, the number of single-candidate swaps needed to transform any two committees $W, W' \in \fPO(I)$ into each other within $\fPO$ lies in $\mathcal{O}(m \cdot k^{1/3})$.
    The same steps hold if $W, W' \in \wfPO(I)$.

    It remains to remove the general-position assumption. We use a symbolic perturbation of the score-line arrangement.
    Let us define the following perturbation: For $\varepsilon>0$, replace each score line $s_c(\lambda)$ by
    \[
        s_c^\varepsilon(\lambda) \coloneqq
        s_c(\lambda)
        + \varepsilon\bigl((1-\lambda)\chi^W_c + \lambda \chi^{W'}_c\bigr)
        + \varepsilon^2(\alpha_c+\beta_c\lambda),
    \]
    where the numbers $\alpha_c,\beta_c$ are chosen so that, for every sufficiently small $\varepsilon>0$, the perturbed score lines are in general position.
    For all sufficiently small $\varepsilon>0$, this arrangement is in general position, and its unique top-$k$ set at $\lambda=0$ is $W$, while its unique top-$k$ set at $\lambda=1$ is $W'$.
    
    We now show that, for sufficiently small $\varepsilon>0$, every top-$k$ committee in the perturbed arrangement is supported by some original weight vector $\gamma(\lambda)$.
    For any committee $T\in\Committees$, define
    \[
        h_T(\lambda) \coloneqq \min_{c\in T} s_c(\lambda) - \max_{d\in C\setminus T} s_d(\lambda).
    \]
    By \Cref{lemma:weights_to_scores}, $T$ is supported by $\gamma(\lambda)$ if and only if $h_T(\lambda)\geq 0$.
    
    Now consider a committee $T$ that is not supported
    by any $\gamma(\lambda)$. Then $h_T(\lambda)<0$ for every $\lambda\in[0,1]$. Since $h_T$ is continuous and $[0,1]$ is compact, there exists $\delta_T>0$ such that $h_T(\lambda)\leq-\delta_T$ for all $\lambda \in [0,1]$.
    Since there are only finitely many committees, there exists $\delta>0$ such that this inequality holds with $\delta$ in place of $\delta_T$ for every committee not supported by any $\gamma(\lambda)$.
    
    Moreover, $s_c^\varepsilon\to s_c$ uniformly on $[0,1]$ as $\varepsilon\downarrow0$. Thus, for sufficiently small $\varepsilon>0$,
    \[
        |s_c^\varepsilon(\lambda)-s_c(\lambda)| <\frac{\delta}{3}
    \]
    for every candidate $c$ and every $\lambda\in[0,1]$.
    
    Consequently, for every committee $T$ not supported by any original weight vector and every $\lambda\in[0,1]$, we have
    \[
    \begin{aligned}
        \min_{c\in T}s_c^\varepsilon(\lambda) -\max_{d\in C\setminus T} s_d^\varepsilon(\lambda) \leq h_T(\lambda)+\frac{2\delta}{3} < 0.
    \end{aligned}
    \]
    Hence, such a committee cannot be a top-$k$ committee in the perturbed arrangement. Therefore, every top-$k$ committee encountered along the perturbed path is supported by some original weight vector $\gamma(\lambda)$, and consequently, belongs to $\fPO(I)$ [$\wfPO(I)$].

    Finally, since the perturbed arrangement is in general position, every change of the top-$k$ set consists of a single-candidate swap and occurs when the $k$-th and $(k+1)$-st highest score lines cross, i.e., at a vertex of the $k$-level. By the preceding argument, every committee on this perturbed top-$k$ path is supported by some original weight vector $\gamma(\lambda)$, and hence, belongs to $\fPO(I)$ ($\wfPO(I)$). The path starts at $W$ and ends at $W'$, and its length is at most the number of vertices of the perturbed $k$-level. By \citet{Dey98}, this number is $\mathcal{O}(m\cdot k^{1/3})$, which proves the claim.
\end{proof}

\diameterLower*

\begin{figure}[ht!]
\centering
\scalebox{0.8}{
\begin{tikzpicture}[
    yscale=0.6,
    xscale=0.6,
    voter/.style={anchor=south, font=\small},
    rowlabel/.style={anchor=east, font=\small},
    cell/.style={draw=black!40, minimum width=1cm, minimum height=1cm},
    approved/.style={fill=blue!40},
    unapproved/.style={fill=gray!10}
]

\node[voter] at (0.5,-1) {$u_1$};
\node[voter] at (1.5,-1) {$u_2$};
\node[voter] at (2.5,-1) {$u_3$};
\node[voter] at (3.5,-1) {$x_1$};
\node[voter] at (4.5,-1) {$x_2$};
\node[voter] at (5.5,-1) {$x_3$};

\node[font=\small\bfseries] at (1.5,-1.8) {$U$};
\node[font=\small\bfseries] at (4.5,-1.8) {$X$};

\draw[dashed, black!50] (3,-0.25) -- (3,11.25);

\foreach \x in {0,1,2} {
    \draw[approved] (\x,0) rectangle (\x+1,1);
}
\foreach \x in {3,4,5} {
    \draw[unapproved] (\x,0) rectangle (\x+1,1);
}
\node[rowlabel] at (-0.25,0.5) {$k\times C_1$};

\foreach \x in {0,1,2} {
    \draw[unapproved] (\x,1) rectangle (\x+1,2);
}
\foreach \x in {3,4,5} {
    \draw[approved] (\x,1) rectangle (\x+1,2);
}
\node[rowlabel] at (-0.25,1.5) {$k\times C_2$};

\foreach \i in {1,2,3} {
    \foreach \j in {1,2,3} {
        \pgfmathtruncatemacro{\row}{2 + 3*(\i-1) + (\j-1)}
        \pgfmathtruncatemacro{\missingU}{\i-1}
        \pgfmathtruncatemacro{\missingV}{3+\j-1}

        \foreach \x in {0,1,2,3,4,5} {
            \ifnum\x=\missingU
                \draw[unapproved] (\x,\row) rectangle (\x+1,\row+1);
            \else
                \ifnum\x=\missingV
                    \draw[unapproved] (\x,\row) rectangle (\x+1,\row+1);
                \else
                    \draw[approved] (\x,\row) rectangle (\x+1,\row+1);
                \fi
            \fi
        }
        \node[rowlabel] at (-0.25,\row+0.5) {$k\times Z_{\i\j}$};
    }
}
\end{tikzpicture}}
\caption{Instance $I = (A, k)$ used for the proof in \Cref{proposition:lower_diameter}. The candidates in $C_1$ are approved by all voters in $U$, the candidates in $C_2$ by all voters in $X$, and each block $Z_{ij}$ contains $k$ candidates approved by all voters except $u_i$ and $x_j$.}
\label{figure:diameter_constructed_instance}
\end{figure}

\begin{proof}
    \begin{enumerate}[leftmargin=*,label=\textit{(\roman*)}]
        \item
        Let the committee size be $k$. 
        We construct an instance with two voter groups $U=\{u_1,u_2,u_3\}$ and $X=\{x_1,x_2,x_3\}$, and three sets of candidates $C_1,C_2,C_3$.
        Let $C_1=\{a_1,\dots,a_k\}$ and $C_2=\{b_1,\dots,b_k\}$.
        Every candidate in $C_1$ is approved by exactly the voters in $U$, and every candidate in $C_2$ is approved by exactly the voters in $X$.
        For every pair $(i,j) \in [3]\times[3]$, let $Z_{ij}$ be the set of $k$ candidates, each approved by all voters except $u_i$ and $x_j$. 
        Let 
        \begin{equation*}
            C_3 \coloneqq \bigcup_{i,j \in [3]} Z_{ij}.
        \end{equation*} 
        
        We first show that the committees $W \coloneqq C_1$ and $W' \coloneqq C_2$ belong to $\fPO(I)$. 
        By the characterization of $\fPO(I)$, it suffices to provide supporting weight vectors $\weight \in \Weights^+$. \\
        For $W$, assign weight $(1-3\varepsilon)/3$ to each voter in $U$ and weight $\varepsilon$ to each voter in $X$, where $0<\varepsilon<1/9$. 
        Then every candidate in $C_1$ has score of $1-3\varepsilon$, every candidate in $C_2$ has score $3\varepsilon$, and every candidate in \(C_3\) has score at most
        \begin{equation*}
            1-\frac{1-3\varepsilon}{3}-\varepsilon=\frac23.
        \end{equation*}
        Since $1-3\varepsilon>\frac23$, the top-$k$ candidates are exactly the candidates in $C_1$. 
        Hence $W \in \fPO(I)$. 
        By symmetry, $W' \in \fPO(I)$.
        
        Next, we show that no committee in $\fPO(I)$ can contain candidates from both $C_1$ and $C_2$. 
        Let $\weight \in \Weights^+$ be any supporting weight vector. 
        Write $a \coloneqq \sum_{u \in U} \weight_u$ and $b\coloneqq \sum_{x \in X} \weight_x$, such that $a+b=1$. 
        Then every candidate in $C_1$ has score $a$, and every candidate in $C_2$ has score $b$.
        Let $u_i$ be a minimum-weight voter in $U$, and $x_j$ be a minimum-weight voter in $X$. 
        Then, $\weight_{u_i} \le \frac{a}{3}$, and $\weight_{x_j} \le \frac{b}{3}$.
        It follows that every candidate in $Z_{ij}$ has score
        \begin{equation*}
            1-\weight_{u_i}-\weight_{x_j} \ge 1-\frac{a}{3}-\frac{b}{3} = \frac23.
        \end{equation*}
        Since $\min\{a,b\} \le \frac12 < \frac23$, there are at least $k$ candidates in $C_3$ whose score is strictly larger than the smaller of the two scores $a$ and $b$. 
        Therefore, a top-$k$ committee cannot contain candidates from both $C_1$ and $C_2$.
        Thus a top-$k$ committee may contain candidates from at most one of $C_1$ and $C_2$ (in addition to candidates from $C_3$)
        
        Now consider any sequence
        \begin{equation*}
            W=W_0,W_1,\dots, W_j, \dots,W_t=W'
        \end{equation*}
        of committees in $\fPO(I)$ such that consecutive committees differ by a single-candidate swap. 
        Let $j$ be the smallest index such that $W_j$ contains a candidate from $C_2$. 
        So, $W_{j-1}$ contains no candidate from $C_2$. 
        Further, we know from our above argument that $W_j$ cannot contain any candidate from $C_1$. 
        As $W_j$ is obtained from $W_{j-1}$ by a single swap, it follows that $W_{j-1}$ contains at most one candidate from $C_1$. 
        Therefore, it requires at least $k-1$ single swaps to reach $W_{j-1}$.
        Moreover, $W_{j-1}$ contains no candidate from $C_2$, while $W'$ consists of $k$ candidates from $C_2$.
        Hence, we need at least $k$ single swaps to get to $W'$.
        Adding the swaps needed to reach $W_{j-1}$ and then $W'$ yields $2k-1$, establishing the claim.

        \bigskip
        \item
        Let there be two voters $N = [2]$ and the candidates be $C=\{c_1,\dots,c_m\}$, and let $\ell=m-k$.
        We construct valuations so that the supported committees form a sliding window of size $k$:
        \[
            W_r \coloneqq \{c_r,c_{r+1},\dots,c_{r+k-1}\}, \qquad r=1,\dots,\ell+1.
        \]
        Since $N = [2]$ we can parametrize any weight vector $\weight$ by $(t, 1-t)$ for $t \in [0,1]$.
        With this, we can define 
        \[
            s(t) \coloneqq \frac{k+1}{2}+\ell \cdot t.
        \]
        Choose an integer $M$ large enough so that all valuations below are non-negative. Define
        \[
            v_1(c_q):=M-q^2+q(k+1+2\ell), \quad \text{and}\quad v_2(c_q):=M-q^2+q(k+1) \quad \text{for all } q \in [m].
        \]
        These valuations are rational, in fact integral.
       
        For a weight vector $\weight(t)=(t,1-t)$, the score of candidate $c_q$ is
        \[
            s_{c_q}(t) = t\cdot v_1(c_q)+(1-t) \cdot v_2(c_q) = M-q^2+2q s(t) = M+s(t)^2-(q-s(t))^2,
        \]
        where the second equality follows by expanding with $s(t)^2$.
        Thus, for a fixed value of $t$, the candidates with the highest scores are exactly those whose indices $q$ are closest to $s(t)$.
        
        By \Cref{theorem:fractional_po_equal_weights}, and since every positive weight vector for two voters is of the form $\weight(t)=(t,1-t)$ with $t\in(0,1)$, the fractionally Pareto-optimal committees are exactly $W_1, W_2,\dots, W_{\ell+1}$. For the two endpoint windows $W_1$ and $W_{\ell+1}$, choose $t$ sufficiently close to $0$ and $1$, respectively.
    
        Moreover, consecutive committees differ by exactly one candidate, whereas if $|r-s|>1$, then $W_r$ and $W_s$ differ in at least two candidates, since $k\ge 2$.
        Therefore, non-consecutive windows are not adjacent under a single swap. 
        Hence, the swap graph induced by $\fPO(I)$ is exactly the path $W_1 \to W_2 \to \cdots \to W_{\ell+1}$.
        Consequently, to get from $W_1$ to $W_{\ell+1}$ by single swaps, ensuring that every committee on the path satisfies fractional Pareto-optimality, requires at least $\ell$ steps.
    \end{enumerate}
\end{proof}

\noindent
\begin{minipage}{0.62\textwidth}
    While the set of $\fPO$ committees is connected, the same is not true for standard Pareto-optimality. We show that this even holds for the very restrictive case of approval preferences. 
    For additive utilities, disconnectedness of $\PO$ was already shown by \citet{gorski2011connectedness}:
    Consider an instance $I$ with two voters, candidate set $C=\{c_1,\dots,c_{15}\}$, and committee size $k=3$. 
    Let
    \begin{align*}
        v_1&=(55,51,48,44,37,36,27,16,14,10,8,5,3,1,0),\\
        v_2&=(0,3,6,18,19,26,27,28,29,39,41,47,49,50,52).
    \end{align*}
    and consider committee $W=\{c_{10},c_{11},c_{14}\}$, with utility vector $u(W)=(10+8+1,\;39+41+50)=(19,130)$, marked in red in the figure on the right.
    One can verify that $W$ is Pareto-optimal, but no single-candidate swap yields a Pareto-optimal committee.
    Since $I$ admits further Pareto-optimal committees, such as $W'=\{c_1,c_2,c_3\}$, the set $\PO(I)$ is not connected.

    In the following, we strengthen this result by providing approval instances $I$ and $I'$ for which the set $\PO(I)$ and $\wPO(I')$, respectively, are disconnected.
    
\end{minipage}
\hfill
\begin{minipage}{0.35\textwidth}
\begin{tikzpicture}[
    xscale=0.025,
    yscale=0.025,
    every node/.style={font=\scriptsize},
    po/.style={circle, fill=blue!70, inner sep=1.5pt},
    iso/.style={circle, fill=red!75, inner sep=2.2pt},
    neigh/.style={circle, fill=green!55!black, inner sep=1.35pt},
    axis/.style={->, thick},
    frontier/.style={dashed, blue!40},
]

    \draw[axis] (0,0) -- (165,0) node[below] {$u_1$};
    \draw[axis] (0,0) -- (0,160) node[left] {$u_2$};
    
    \foreach \x in {50,100,150}
        \draw (\x,2) -- (\x,-2) node[below=2pt] {\tiny \x};
    \foreach \y in {50,100,150}
        \draw (2,\y) -- (-2,\y) node[left=2pt] {\tiny \y};
    
    \foreach \x/\y in {
    154/9,150/21,147/24,143/27,142/29,139/32,136/37,135/44,
    131/47,128/50,121/51,118/53,117/63,108/64,107/71,100/72,
    96/73,94/75,92/76,91/78,90/83,88/85,85/91,83/93,81/94,
    80/96,73/97,68/100,66/102,64/103,63/105,57/106,55/108,
    54/109,52/114,50/115,49/117,47/119,45/120,44/122,42/123,
    41/125,39/127,37/128,28/129,18/135,16/137,15/138,13/140,
    11/142,9/146,8/148,6/149,4/151
    }{
        \node[po] at (\x,\y) {};
    }
    
    \foreach \x/\y in {
    73/80,69/83,66/89,66/86,64/91,62/98,62/92,60/94,59/95,
    57/97,55/107,55/99,54/106,53/109,48/108,47/115,46/110,
    45/117,45/107,38/116,36/118,34/108,32/109,27/117,25/119,
    25/118,23/127,23/120,21/129,18/132,16/136,14/138,12/140,
    11/141,9/143
    }{
        \node[neigh] at (\x,\y) {};
    }
    
    \node[iso, label={[above right=-1pt]{$W$}}] at (19,130) {};
    
    
    \draw[frontier]
    (154,9)--(150,21)--(147,24)--(143,27)--(142,29)--(139,32)--(136,37)--(135,44)
    --(131,47)--(128,50)--(121,51)--(118,53)--(117,63)--(108,64)--(107,71)
    --(100,72)--(96,73)--(94,75)--(92,76)--(91,78)--(90,83)--(88,85)--(85,91)
    --(83,93)--(81,94)--(80,96)--(73,97)--(68,100)--(66,102)--(64,103)--(63,105)
    --(57,106)--(55,108)--(54,109)--(52,114)--(50,115)--(49,117)--(47,119)
    --(45,120)--(44,122)--(42,123)--(41,125)--(39,127)--(37,128)--(28,129)
    --(19,130)--(18,135)--(16,137)--(15,138)--(13,140)--(11,142)--(9,146)
    --(8,148)--(6,149)--(4,151);
    
    \node[po] at (68,138) {};
    \node[anchor=west] at (70,138) {Pareto-optimal};
    
    \node[iso] at (68,149) {};
    \node[anchor=west] at (70,149) {isolated};
    
    \node[neigh] at (68,127) {};
    \node[anchor=west] at (70,127) {one-swap neighbors};
\end{tikzpicture}
\captionof{figure}{Visualization of the counterexample for connectedness with additive utilities.}
\end{minipage}

\disconnectedness*

Similar to the proof for the monotonicity violation of $\PO$, we can encode the disconnectedness of a $\PO$ committee as a SAT problem. Unfortunately, this yields many clauses and does not terminate for large $m$ and $k$.
That said, we establish a sufficient structural condition on an instance to admit an isolated $\PO$ committee.

\begin{lemma}\label{lem:suff-cond-isolated}
    Let $A$ be an approval profile, let $k\geq 2$, and let $W\in\Committees_k$. Assume that every candidate $d\in C\setminus W$ has at least $k$ approval-equivalent copies, i.e., $\bigl\lvert\{d'\in C : N_{d'}=N_d\}\bigr\rvert \geq k$.    
    Then the following statements hold:
    \begin{enumerate}[label=\textit{(\roman*)}]
        \item If $W\in\PO_k(A)$ and $W\setminus\{c\}\notin\PO_{k-1}(A)$ for every $c \in W$, then $W$ is isolated in $\PO_k(A)$ under single-candidate swaps. Moreover, $\PO_k(A)$ is disconnected.
    
        \item If $W\in\wPO_k(A)$ and $W\setminus\{c\}\notin\wPO_{k-1}(A)$ for every $c \in W$, then $W$ is isolated in $\wPO_k(A)$ under single-candidate swaps. Moreover, $\wPO_k(A)$ is disconnected.
    \end{enumerate}
\end{lemma}

\begin{proof}
    First consider the $\PO$ case.
    Fix $c\in W$ and an arbitrary candidate $d\in C\setminus W$.
    The committee $X \coloneqq (W\setminus\{c\})\cup\{d\}$ is an arbitrary one-swap neighbor of $W$.
    We show that $X$ fails $\PO$.
    
    Since $W\setminus\{c\}\notin\PO_{k-1}(A)$, there exists a committee $B_c\in\Committees_{k-1}$ such that $u_i(B_c)\geq u_i(W\setminus\{c\})$ for every voter $i \in N$, with strict inequality for at least one voter.
    By assumption, there are at least $k$ candidates with the same approval set as $d$. Since $B_c$ contains only $k-1$ candidates, at least one such copy, say $d^*$, does not belong to $B_c$. Thus, $N_{d^*}=N_d$ and $d^* \notin B_c$. Consequently, $B_c\cup\{d^*\}$ is a feasible committee of size $k$.
    For every voter $i$, we have
    \begin{align*}
        u_i(B_c\cup\{d^*\})-u_i(X)
        &= u_i(B_c)+u_i(\{d^*\}) -u_i(W\setminus\{c\})-u_i(\{d\})\\
        &= u_i(B_c)-u_i(W\setminus\{c\}),
    \end{align*}
    where the second equality follows from $N_{d^*}=N_d$.
    Hence, $u_i(B_c\cup\{d^*\})\geq u_i(X)$ for every $i \in N$, with at least one inequality strict.
    Therefore, $B_c\cup\{d^*\}$ Pareto-dominates $X$, and so $X\notin\PO_k(A)$.
    Since $c\in W$ and $d\in C\setminus W$ were arbitrary, every
    one-swap neighbor of $W$ fails $\PO$.
    As $W$ itself satisfies $\PO$ by assumption, $W$ is isolated.

    Note that there exists another committee $W' \ne W$ in $\PO_k(A)$, implying disconnectedness of $\PO_k(A)$. This is because $W$ is not in $\fPO_k(A)$ as $\fPO$ satisfies committee monotonicity. However, since $\fPO_k(A)$ is non-empty, there exists another committee $W' \ne W$ in $\fPO_k(A)$, hence in $\PO_k(A)$.
    
    The same arguments hold for $\wPO$.
\end{proof}

We are now ready to prove the main theorem. We encode the conditions of \Cref{lem:suff-cond-isolated} using SAT.
For this SAT formulation, we do not encode each candidate separately; instead, we encode candidate types, with candidates of the same type sharing the same approval set.
This substantially reduces the size of the SAT instance.

We will encode $q>k$ candidate types $C=[q]$ with the desired committee being $W=[k]$.
To ensure that candidate types from $C\setminus W$ have enough copies, we define our committee space to include multisets. That is, let
\[
\mathcal{M}_{k}:=
\left\{
W': W' \text{ is a size-}k\text{ multiset over } C, \text{with }
\mu_{W'}(c)\leq 1 \text{ for every }c\in W
\right\},
\]
where $\mu_{W'}(c)$ is the multiplicity of $c$ in $W'$. Define $\mathcal{M}_{k-1}$ analogously.
These multisets represent exactly the set of feasible committees; each candidate is represented separately. Instead, we encode candidate types, with candidates of the same type sharing the same approval setting, which is analogous to monotonicity.
We use variables $x_A$ for non-empty $A\subseteq C$ to encode the presence of voter types and the non-emptiness clause $\bigvee_{\emptyset \ne A\subseteq C} x_A$.

A voter with ballot $A$ approves every copy of each type in $A$,
so her utility is $u_A(T)=\sum_{c\in A}\mu_T(c)$. $\succ$ is defined with respect to these utilities. Encode $\PO$ again with
\[
    \bigwedge_{W' \in \mathcal{M}_{k} \setminus \{W\}} \left[ \left( \bigvee_{A: W \succ_A W'}x_A \right) \vee \left( \bigwedge_{A: W \prec_A W' }\neg x_A \right)   \right].
\]

Non-reducibility is encoded analogously to non-extendability:
for each $c\in W$, we require a blocker of $W\setminus\{c\}$ from
$$
D(c)
:=
\{W''\in\mathcal{M}_{k-1}:c\in W''\}.
$$
This restriction loses no blockers: if $W''$ did not contain $c$, adding $c$ to $W''$ would yield a feasible size-$k$ multiset that Pareto-dominates $W$, contradicting its $\PO$ requirement.
Introduce variables $z_{c,W''}$ for $W''\in D(c)$ and require
$$
\bigwedge_{c\in W}\bigvee_{B\in D(c)}z_{c,B}.
$$
Again encode $z_{c,W''} = 1$ implies $W''$ Pareto dominating $W \setminus \{c\}$ :
$$\bigwedge_{c \in W}\bigwedge_{W'' \in D(c)}\left[\left( \neg z_{c, W''} \vee \bigvee_{A: W\setminus\{c\} \prec_A W''} x_A \right) \wedge \bigwedge_{A: W\setminus\{c\} \succ_A W''} (\neg z_{c,W''} \vee \neg x_A) \right].$$
This concludes the SAT encoding.

The formula is satisfiable for $q=12$ and $k=4$; expanding the
outsider types gives the $4+4\cdot 8=36$ candidates of \Cref{fig:counterexample_connectivity}.
We denote the resulting approval profile by $A$.

For completeness, we write down the resulting copied profile $A$.
Let $ W=\{c_1,c_2,c_3,c_4\}$.
The $32$ candidates outside $W$ form eight approval-equivalence classes
\[
    \begin{aligned}
        D_1&=\{c_5,c_6,c_7,c_8\},&
        D_2&=\{c_9,c_{10},c_{11},c_{12}\},\\
        D_3&=\{c_{13},c_{14},c_{15},c_{16}\},&
        D_4&=\{c_{17},c_{18},c_{19},c_{20}\},\\
        D_5&=\{c_{21},c_{22},c_{23},c_{24}\},&
        D_6&=\{c_{25},c_{26},c_{27},c_{28}\},\\
        D_7&=\{c_{29},c_{30},c_{31},c_{32}\},&
        D_8&=\{c_{33},c_{34},c_{35},c_{36}\}.
    \end{aligned}
\]
Thus, every candidate outside $W$ has four approval-equivalent copies.

A direct verification of the profile shows that $W\in\PO_4(A)$.
Moreover, each of its size-$3$ subcommittees is Pareto-dominated; the table gives the dominating committee separately:
\[
    \begin{array}{c|c}
        \text{subcommittee of }W & \text{Pareto-dominating committee} \\
        \hline
        \{c_2,c_3,c_4\} & \{c_1,c_{17},c_{25}\} \\
        \{c_1,c_3,c_4\} & \{c_2,c_9,c_{29}\} \\
        \{c_1,c_2,c_4\} & \{c_3,c_{21},c_{33}\} \\
        \{c_1,c_2,c_3\} & \{c_4,c_5,c_{13}\}.
    \end{array}
\]
Hence, by \Cref{lem:suff-cond-isolated}, $W$ is isolated in $\PO_4(A)$ and $\PO_4(A)$ is disconnected.

\medskip

We next show that $\wPO$ can also be disconnected. For this, contrary to the proof of $\PO$, we provide an analytical solution.
We construct an approval-based instance $I'$ satisfying the conditions of \Cref{lem:suff-cond-isolated}.
Let $k=5$, and the desired isolated committee is $W=\{w_1,\ldots,w_5\}$.
For every $i\in[5]$ and $r\in[3]$, let
\[
    D_{i,r} = \{d_{i,r}^1,\ldots,d_{i,r}^5\}
\]
consist of five approval-equivalent candidates.
Thus, for each $i$, there are three candidate types $D_{i,1},D_{i,2},D_{i,3}$, each with five copies.
Hence, the instance has $80$ candidates in total and satisfies the conditions of \Cref{lem:suff-cond-isolated}.

For every vector $\rho=(\rho_1,\ldots,\rho_5)\in[3]^5$, introduce a voter $a_\rho$ who approves all five copies of type $D_{i,\rho_i}$ for every $i\in[5]$, and no other candidates.
Thus, $\rho_i$ specifies which of the three candidate types in block $i$ is approved by voter $a_\rho$.
Hence, there are $3^5=243$ such voters.

Further, for every $i\in[5]$ and every $\rho\in[3]^{[5]\setminus\{i\}}$, introduce another voter $b_{i,\rho}$.
This voter approves exactly $w_i$ among the candidates in $W$, approves no candidate from the three types $D_{i,1},D_{i,2},D_{i,3}$, and, for every $j\neq i$, approves all candidates belonging to the two types $D_{j,r}$, for $r\in[3]\setminus\{\rho_j\}$.
In other words, in every block $j\neq i$, the coordinate $\rho_j$ specifies the unique type disapproved by $b_{i,\rho}$.
There are $5\cdot3^4=405$ such voters, and therefore, the instance contains $648$ voters in total.

We claim that $W\in\wPO_5(A')$.
Let $W'\neq W$ be an arbitrary committee of size $5$.
We show that $W'$ does not strictly improve every voter over $W$.
First suppose that, for every $i\in[5]$, at least one of the three types $D_{i,1}, D_{i,2}$, or $D_{i,3}$ is not represented in $W'$.
For each $i$, choose $\rho_i\in[3]$ such that $W'\cap D_{i,\rho_i}=\emptyset$.
The voter $a_\rho$ approves no candidate in $W'$.
Hence, $u_{a_\rho}(W')=0=u_{a_\rho}(W)$, so $W'$ does not strictly improve every voter.

It remains to consider the case in which, for some $i\in[5]$, all three types are represented in $W'$.
Then at least three of the five candidates in $W'$ belong to block $i$, so $W'$ contains at most two candidates outside this block.
We choose a voter $b_{i,\rho}$ who approves at most one of these remaining candidates. Recall that $b_{i,\rho}$ approves no candidate in block $i$, approves $w_i$, and disapproves every $w_j$ with $j\neq i$.
If $w_i\in W'$, then there is at most one further candidate of $W'$ outside block $i$.
If this candidate belongs to some type $D_{j,r}$, we choose $\rho_j=r$, so that $b_{i,\rho}$ disapproves it.
Thus, $u_{b_{i,\rho}}(W')\le1=u_{b_{i,\rho}}(W)$.
If $w_i\notin W'$, there are at most two candidates of $W'$ outside block $i$.
If these candidates belong to different blocks, we choose the corresponding coordinates of $\rho$ so that both are disapproved.
If they belong to the same block and have the same type, we choose that type as the disapproved type.
Finally, if they belong to two different types of the same block, $b_{i,\rho}$ can disapprove one of these two types and therefore approve at most the other candidate.
Candidates $w_j$ with $j\neq i$ are disapproved automatically.
Hence, $u_{b_{i,\rho}}(W')\le1=u_{b_{i,\rho}}(W)$.

Thus, no size-$5$ committee strictly improves every voter over $W$, and consequently $W\in\wPO_5(A')$.

It remains to show that no size-$4$ subcommittee of $W$ is weakly Pareto-optimal.
Fix $i\in[5]$ and define
\[
    B_i \coloneqq \{w_i,d_{i,1}^1,d_{i,2}^1,d_{i,3}^1\}.
\]
We claim that $B_i$ strictly improves every voter over $W\setminus\{w_i\}$.

For a voter $a_\rho$, $u_{a_\rho}(W\setminus\{w_i\})=0$, whereas $a_\rho$ approves exactly one of $d_{i,1}^1, d_{i,2}^1$ and $d_{i,3}^1$.
Hence, $u_{a_\rho}(B_i)=1$.
For a voter $b_{i,\rho}$, we likewise have $u_{b_{i,\rho}}(W\setminus\{w_i\})=0$, $u_{b_{i,\rho}}(B_i)=1$.
Finally, consider $b_{j,\rho}$ with $j\neq i$.
Since $w_j\in W\setminus\{w_i\}$, $u_{b_{j,\rho}}(W\setminus\{w_i\})=1$.
In block $i$, this voter approves exactly two of the three types, and therefore exactly two of $d_{i,1}^1$, $d_{i,2}^1$, and $d_{i,3}^1$
Hence, $u_{b_{j,\rho}}(B_i)=2$.

$B_i$ strictly improves every voter over $W\setminus\{w_i\}$, and therefore, $W\setminus\{w_i\}\notin\wPO_4(A')$ for every $i \in [5]$.

Since every candidate outside $W$ has five approval-equivalent copies, \Cref{lem:suff-cond-isolated} implies that $W$ is isolated in $\wPO_5(A')$ and the set $\wPO_5(A')$ is disconnected.

This concludes the proof of \Cref{proposition:disconnectedness}.

\newcommand{\approw}[2]{%
\foreach \x in {#2} {%
  \draw[fill=blue!40] (\x,#1) rectangle ++(1,1);
}%
}

\begin{figure}[t]
\centering
\scalebox{0.38}{
\begin{tikzpicture}[
    voter/.style={anchor=north}
]

  \fill[gray!10] (0,0) rectangle (41,36);
  \draw[step=1] (0,0) grid (41,36);

  \approw{35}{1,15,16,17,20,21,22,23,24,25,26,27,28,29,30,31,40}

  \approw{34}{0,2,3,4,5,6,7,8,9,10,11,12,13,18,19,20,21,22,23,40}

  \approw{33}{14,18,24,25,26,27,32,33,34,35,36,37,38,39,40}

  \approw{32}{19,28,29,30,31,32,33,34,35,40}

  \approw{31}{0,2,3,4,5,6,14,18,20,21,22,23,24,25,26,27,28,32,36,37,38,39,40}

  \approw{30}{0,2,3,4,5,6,14,18,20,21,22,23,24,25,26,27,28,32,36,37,38,39,40}

  \approw{29}{0,2,3,4,5,6,14,18,20,21,22,23,24,25,26,27,28,32,36,37,38,39,40}

  \approw{28}{0,2,3,4,5,6,14,18,20,21,22,23,24,25,26,27,28,32,36,37,38,39,40}

  \approw{27}{1,2,3,7,8,9,14,15,20,24,25,26,27,28,29,30,31,32,33,34,35,36,40}

  \approw{26}{1,2,3,7,8,9,14,15,20,24,25,26,27,28,29,30,31,32,33,34,35,36,40}

  \approw{25}{1,2,3,7,8,9,14,15,20,24,25,26,27,28,29,30,31,32,33,34,35,36,40}

  \approw{24}{1,2,3,7,8,9,14,15,20,24,25,26,27,28,29,30,31,32,33,34,35,36,40}

  \approw{23}{1,7,8,9,10,11,12,13,15,16,17,18,20,21,22,23,24,25,26,27,29,33,36,37,38,40}

  \approw{22}{1,7,8,9,10,11,12,13,15,16,17,18,20,21,22,23,24,25,26,27,29,33,36,37,38,40}

  \approw{21}{1,7,8,9,10,11,12,13,15,16,17,18,20,21,22,23,24,25,26,27,29,33,36,37,38,40}

  \approw{20}{1,7,8,9,10,11,12,13,15,16,17,18,20,21,22,23,24,25,26,27,29,33,36,37,38,40}

  \approw{19}{0,1,4,7,10,11,14,15,16,17,18,19,21,24,30,32,33,34,35,37,39,40}

  \approw{18}{0,1,4,7,10,11,14,15,16,17,18,19,21,24,30,32,33,34,35,37,39,40}

  \approw{17}{0,1,4,7,10,11,14,15,16,17,18,19,21,24,30,32,33,34,35,37,39,40}

  \approw{16}{0,1,4,7,10,11,14,15,16,17,18,19,21,24,30,32,33,34,35,37,39,40}

  \approw{15}{1,2,4,5,8,10,12,14,16,19,20,21,22,23,25,28,29,30,31,34,36,37,38,39,40}

  \approw{14}{1,2,4,5,8,10,12,14,16,19,20,21,22,23,25,28,29,30,31,34,36,37,38,39,40}

  \approw{13}{1,2,4,5,8,10,12,14,16,19,20,21,22,23,25,28,29,30,31,34,36,37,38,39,40}

  \approw{12}{1,2,4,5,8,10,12,14,16,19,20,21,22,23,25,28,29,30,31,34,36,37,38,39,40}

  \approw{11}{2,3,5,6,8,9,12,13,15,16,17,18,19,22,26,31,32,33,34,35,36,38,39,40}

  \approw{10}{2,3,5,6,8,9,12,13,15,16,17,18,19,22,26,31,32,33,34,35,36,38,39,40}

  \approw{9}{2,3,5,6,8,9,12,13,15,16,17,18,19,22,26,31,32,33,34,35,36,38,39,40}

  \approw{8}{2,3,5,6,8,9,12,13,15,16,17,18,19,22,26,31,32,33,34,35,36,38,39,40}

  \approw{7}{4,5,6,10,11,12,13,16,17,23,24,25,26,27,28,29,30,31,32,33,34,35,37,38,39,40}

  \approw{6}{4,5,6,10,11,12,13,16,17,23,24,25,26,27,28,29,30,31,32,33,34,35,37,38,39,40}

  \approw{5}{4,5,6,10,11,12,13,16,17,23,24,25,26,27,28,29,30,31,32,33,34,35,37,38,39,40}

  \approw{4}{4,5,6,10,11,12,13,16,17,23,24,25,26,27,28,29,30,31,32,33,34,35,37,38,39,40}

  \approw{3}{0,3,6,7,9,11,13,15,17,19,20,21,22,23,27,28,29,30,31,35,36,37,38,39,40}

  \approw{2}{0,3,6,7,9,11,13,15,17,19,20,21,22,23,27,28,29,30,31,35,36,37,38,39,40}

  \approw{1}{0,3,6,7,9,11,13,15,17,19,20,21,22,23,27,28,29,30,31,35,36,37,38,39,40}

  \approw{0}{0,3,6,7,9,11,13,15,17,19,20,21,22,23,27,28,29,30,31,35,36,37,38,39,40}

  \foreach \i in {1,...,36} {
    \pgfmathsetmacro{\yy}{36.5-\i}
    \node[anchor=east] at (-0.3,\yy) {$c_{\i}$};
  }

  \foreach \i in {1,...,41} {
    \pgfmathsetmacro{\xx}{\i-0.5}
    \node[voter] at (\xx,-0.5) {$\i$};
  }

\end{tikzpicture}
}
\caption{Counterexample for $\PO$ failing connectedness in the approval-based setting.}
\label{fig:counterexample_connectivity}
\end{figure}

\subsection{The Tension between Equality-Oriented Voting Rules and $\fPO$} \label{appendix_section:voting_rules}

\subsubsection{Welfarist Rules}\label{appendix:welfarist_rules}

\concaveWelfarist*

\begin{lemma}\label{lem:welfarist_bump}
    Suppose there exist a function $g_n: \R^n_{\geq 0} \rightarrow \R$, $x,y,z\in\Q^n_{>0}$ and $\lambda\in(0,1)$ such that $z < \lambda x+(1-\lambda)y$ coordinatewise and $g_n(z)>g_n(x),g_n(y)$. Then, the welfarist rule $f_g$ violates $\wfPO$. 
\end{lemma}

\begin{figure}[!ht]
    \centering
    \begin{tikzpicture}[
        scale=0.75,
        every node/.style={font=\small},
        cand/.style={circle, fill=blue!60, inner sep=1.8pt},
        chosen/.style={circle, fill=red!70, inner sep=1.8pt},
        frac/.style={circle, fill=green!50!black, inner sep=1.8pt},
        axis/.style={->, thick},
        aux/.style={dashed, gray!65},
        dom/.style={-{Stealth[length=2.5mm,width=2mm]}, very thick, gray!50,
                    shorten >=3pt, shorten <=3pt}
    ]
        \draw[axis] (0,0) -- (4.4,0) node[below] {$u_1$};
        \draw[axis] (0,0) -- (0,4.4) node[left] {$u_2$};
        \coordinate (x) at (1.1,3.7);
        \coordinate (y) at (3.7,1.1);
        \coordinate (m) at (2.4,2.4);
        \coordinate (z) at (1.7,1.7);
        \draw[aux] (x) -- (y);
        \node[cand, label={[above left=-2pt] {$x$}}] at (x) {};
        \node[cand, label={[below right=1pt] {$y$}}] at (y) {};
        \node[frac, label={[above right=-2pt] {$m=\frac{x+y}{2}$}}] at (m) {};
        \node[chosen, label={[below left=1pt] {$z$}}] at (z) {};
        \draw[dom] (z) -- (m);
    \end{tikzpicture}
    \caption{Geometric illustration of the proof for \Cref{lem:welfarist_bump} with $n=2$. 
    } 
    \label{fig:strictly-concave-counterexample}
\end{figure}

\begin{proof}
    We construct an instance with committee size $k=1$ and three candidates $a,b,d$ whose utility vectors are
    \[
        u(a)=z,\qquad u(b)=x,\qquad u(d)=y.
    \]
    Since $g_n(z)>g_n(x)$ and $g_n(z)>g_n(y)$, the rule $f_g$ uniquely selects $\{a\}$.
    However, the fractional committee assigning weight $\lambda$ to $b$, and $1-\lambda$ to $d$ induces the utility vector $\lambda x+(1-\lambda)y$ which is strictly larger than $z=u(\{a\})$ in every coordinate. 
    Thus, $\{a\}$ is strictly fractionally Pareto-dominated, so $\{a\}\notin\wfPO(I)$. 
    Hence, $f_g$ violates $\wfPO$.
\end{proof}

Using the previous lemma, we can now give the proof of \Cref{theorem:concave_welfarist}.

\begin{proof}
    Since $g_n$ is strictly concave and finite on $\R^n_{\ge 0}$, it is continuous on the interior $\R^n_{>0}$.

    Moreover, since $g_n$ is an aggregation function of a welfarist rule, it is monotone by definition. Together with strict concavity, this implies that $g_n$ is strictly monotone. To see this, let $x,y\in\R^n_{\ge0}$ with $x\le y$ and $x\neq y$. If $g_n(x)=g_n(y)$, then strict concavity gives
    \[
        g_n\!\left(\frac{x+y}{2}\right) > \frac{g_n(x)+g_n(y)}{2} = g_n(y).
    \]
    On the other hand, $\frac{x+y}{2}\le y$, so monotonicity implies $g_n(\frac{x+y}{2})\le g_n(y)$; a contradiction. Hence, $g_n(x)<g_n(y)$, and $g_n$ is strictly monotone.
     
    We now show that the condition of \Cref{lem:welfarist_bump} is satisfied.

    Since $n\ge 2$, let $\e_1,\e_2 \in \R^n$ be two distinct standard basis vectors. 
    Without loss of generality, assume that $g_n(\e_2 + \mathbf{1}) \geq   g_n(\e_1 + \mathbf{1})$.
    Now since $g_n$ is strictly monotone it holds that $g_n(\mathbf{1}) < g_n(\e_1 + \mathbf{1}) \leq g_n(\e_2 + \mathbf{1})$. 
    By the intermediate value theorem, there exists some $r \in [0,1]$ such that $g_n(r\e_2 + \mathbf{1}) = g_n(\e_1 + \mathbf{1})$.
    They are distinct because the first coordinate of $\e_1+\mathbf1$ equals $2$, while the first coordinate of $r\e_2+\mathbf1$ equals $1$.
    
    Now choose some $\lambda \in (0,1)$, we here fix it at $\lambda = \frac{1}{2}$. By strict concavity,
    \[
        g_n(\lambda (\e_1 + \mathbf{1}) + (1-\lambda)(r\e_2+ \mathbf{1}))>\lambda g_n(\e_1 +  \mathbf{1})+(1-\lambda)g_n(r\e_2 + \mathbf{1})=g_n(\e_1 +  \mathbf{1})=g_n(r\e_2+ \mathbf{1}).
    \]
    In order to apply \Cref{lem:welfarist_bump} with rational utility vectors, we need to choose rational points $x,y,z\in\Q^n_{>0}$ satisfying the required inequalities. 
    To this end, choose $s \in \N$ such that 
    $g_n(\frac{1}{2} (\e_1 + \mathbf{1}) + \frac{1}{2}(r\e_2+ \mathbf{1})) > g_n(\e_1 +  \mathbf{1}) + \frac{1}{s}$. 

    By continuity of $g_n$, there exists a $z \in \Q^n_{>0}$ with $z < \frac{1}{2} (\e_1 + \mathbf{1}) + \frac{1}{2}(r\e_2+ \mathbf{1})$ and $g_n(z) > g_n(\frac{1}{2} (\e_1 + \mathbf{1}) + \frac{1}{2}(r\e_2+ \mathbf{1}))-\frac{1}{3s}$.
    Similarly, there exists an $x \in \Q^n_{>0}$ with $x > \e_1 + \mathbf{1}$ and $g_n(x) < g_n(\e_1 + \mathbf{1})+ \frac{1}{3s}$, and a $y \in \Q^n_{>0}$ with $y > r\e_2 + \mathbf{1}$ and $g_n(y) < g_n(r\e_2 + \mathbf{1})+ \frac{1}{3s}$.
    By choosing these rational approximations sufficiently close to the respective points, we can ensure that all the above strict inequalities hold simultaneously.

    Thus, $z < \frac{1}{2} x + \frac{1}{2} y$, $g_n(z) > \max\{g_n(x),g_n(y)\}$ and due to \Cref{lem:welfarist_bump}, $f_g$ violates $\wfPO$ and therefore also $\fPO$.
\end{proof}

We provide an example here in the approval-based setting, showing that concavity alone is not incompatible with $\fPO$.
The construction provides a strictly concave function $g_n$ such that the maximizing committees are always approval-voting winners.

\begin{proposition}
    There exists a strictly concave family of functions $g$ such that the corresponding welfarist rule $f_g$ satisfies $\fPO$ in the approval setting.
\end{proposition}

\begin{proof}
    For a fixed number of voters $n$, consider the function
    \[
        g_n(x) \coloneqq \sum_{i=1}^n \left( x_i + \frac{1}{2n}\frac{x_i}{1+x_i} \right), \qquad x \in \R_{\geq 0}^n.
    \]
    Note that $g_n$ is anonymous, since it is invariant under permutations of the coordinates.

    We first show that every committee selected by $f_{g}$ is also selected by approval voting.
    For $x\in\R_{\geq0}^n$, let $S(x)\coloneqq\sum_{i=1}^n x_i$.
    Since $0\leq\frac{x_i}{1+x_i}<1$ for every $i \in [n]$, we have
    \[
        S(x) \leq g_n(x) = S(x) + \frac{1}{2n}\sum_{i=1}^n \frac{x_i}{1+x_i} < S(x)+\frac12.
    \]
    Now let $W$ be a committee that is not selected by approval voting.
    Then there exists a committee $W'$ with a strictly larger approval score.
    Since approval scores are integer-valued, $S(u(W'))\geq S(u(W))+1$.
    Consequently,
    \[
        g_n(u(W')) \geq S(u(W')) \geq S(u(W))+1 > S(u(W))+\frac12 > g_n(u(W)).
    \]
    Thus, $W$ cannot maximize $g_n$. Hence, every committee selected by $f_{g}$ is also selected by approval voting.
    Since approval voting satisfies $\fPO$ (using uniform voter weights), we conclude that $f_{g}$ satisfies $\fPO$.

    It remains to verify that $g_n$ is strictly concave. For every coordinate $i \in N$,
    \[
        \frac{\partial g_n}{\partial x_i}(x) = 1+\frac{1}{2n(1+x_i)^2},
    \]
    and the Hessian of $g_n$ is diagonal, with
    \[
        \frac{\partial^2 g_n}{\partial x_i^2}(x) = -\frac{1}{n(1+x_i)^3} <0.
    \]
    Note that all off-diagonal entries are zero. Therefore, for every non-zero vector $z\in\R^n$,
    \[
        z^\top\nabla^2 g_n(x)z = -\sum_{i=1}^n \frac{z_i^2}{n (1 + x_i)^3} <0.
    \]
    Hence, the Hessian is negative definite on $\R_{\geq0}^n$, and $g_n$ is strictly concave.
\end{proof}

\subsubsection{Thiele Rules}

\begin{lemma}
    If $h$ is not strictly monotone, then $h$-Thiele violates $\wPO$ for some approval profile $I = (A,k)$.     
    \label{lemma:notMonotoneNotPO} 
\end{lemma}
\begin{proof}
    Since $h$ is not strictly monotone, there exists $x \in \N$ such that  $h(x) \geq h(x+1)$. Consider a one voter instance $N = \{1\}$, $k = x+1$, $C = \{c_1, \dots, c_{x+1}\} \cup \{c'\}$, $A_1 = \{c_1, \dots, c_{x+1}\}$, so the candidate $c'$ is unapproved. Then, the committee $W' = \{c', c_1, \dots, c_x\}$ is selected by $h$-Thiele, since the score of any committee is either $h(x)$ or $h(x+1)$. But clearly $W'$ is strictly Pareto-dominated by $W = \{c_1, \dots, c_{x+1}\}$. Thus, $W' \notin \wPO(I)$.     
\end{proof}

\ThieleFPOviolation*
\begin{proof}
    We can assume that $h$ is strictly increasing since otherwise the rule does not even satisfy $\wPO$ due to \Cref{lemma:notMonotoneNotPO} and, therefore, also not $\wfPO$.

    Since scaling $h$ does not affect the outcome of the $h$-Thiele rule without loss of generality, we assume that $\delta_h(t) \leq 1$. Now, fix some large enough $x \in \N$ such that 
    \begin{equation*}
        \delta_h(t) > \delta_h(t+1)+\frac{3}{x}.
    \end{equation*}

    We now construct an instance based on the example in \Cref{figure:counter_example_weak_fPO}.
    
    Let $N = [6\cdot x \cdot {5 \choose 3}]$, and $C = \{c_1, c_2\} \cup C_{F} \cup C_{\text{opt}}$.
    We let $\lvert C_{\text{opt}}\rvert := t-1$ and 
    for every $c \in C_{\text{opt}}$, set $N_c := N$, where $N_c$ denotes the set of voters approving candidate $c$.  
    Next, we set $N_{c_1} = [\frac{n}{2}]$ and $N_{c_2} = N \setminus N_{c_1}$.
    Finally, let $C_F = C_F^1 \cup C_F^2 \cup C_F^3 \cup C_F^4 \cup C_F^5$, with 
    \begin{align*}
        C_F^i &= \{c_{A}^i \colon A \subseteq D_i, \lvert A \rvert =  30x+6\} \quad \text{for } i \in [5], \text{ where}\\
        D_1 &= \{i \cdot 6x+j \colon i \in \{0,1,2,7,8,9\}, j \in [6x]\}, \\
        D_2 &= \{i \cdot 6x+j \colon i \in \{0, 2, 4, 5, 6, 9\}, j \in [6x] \}, \\
        D_3 &= \{i \cdot 6x+j \colon i \in \{0,1, 3, 5, 6, 7\}, j \in [6x]\}, \\
        D_4 &= \{i \cdot 6x+j \colon i \in \{1, 3, 4, 6, 8, 9\}, j \in [6x] \}, \\
        D_5 &= \{i \cdot 6x+j \colon i \in \{2, 3, 4, 5, 7, 8\}, j \in [6x] \}, 
    \end{align*}
    and $N_{c_{A}^i} = A$. 
    
    We now claim that $h$-Thiele selects the committee $W$ with $W = C_{\text{opt}} \cup \{c_1,c_2\}$ whenever $k = t+1$. 
    It is easy to see that any optimal committee selects all candidates in $C_{\text{opt}}$.
    Further, observe that both $c_1$ and $c_2$ are approved by precisely half the voters, which corresponds to $\frac12 \cdot 6 \cdot {5 \choose 3} \cdot x = 30 \cdot x$ voters for each candidate. 
    Each candidate in $C_F$ is approved by precisely $30x + 6$ voters.

    Using the two observations above, we compare the score of $W$ with an arbitrary committee
    \[
        W'=C_{\text{opt}}\cup\{d,d'\}
    \]
    with $\{d,d'\}\neq\{c_1,c_2\}$.
    Let $a$ be the number of voters who approve neither $d$ nor $d'$, and let $b$ be the number of voters who approve both $d$ and $d'$.
    By construction, $a\ge 6x$ and $b\le a+12$, where the bound $b\le a+12$ follows from the approval-set sizes: if $d,d'\in C_F$, then $|N_d|+|N_{d'}|=2(30x+6)=n+12$, and if one of them is $c_1$ or $c_2$, then $|N_d|+|N_{d'}|\le n+6$.
    Moreover, $a\ge 6x$ because the block structure ensures that any pair $\{d,d'\}\neq\{c_1,c_2\}$ leaves at least one block of size $6x$ uncovered.
    
    Under $W$, every voter has utility $t$.
    Under $W'$, voters approving neither $d$ nor $d'$ have utility $t-1$, voters approving exactly one of them have utility $t$, and voters approving both have utility $t+1$.
    As a result, 
    \begin{align*}
    \operatorname{sc}_h(W)-\operatorname{sc}_h(W')
    =\;& n h(t)-\Bigl(a h(t-1)+b h(t+1)+(n-a-b)h(t)\Bigr)\\
    =\;& a\bigl(h(t)-h(t-1)\bigr)-b\bigl(h(t+1)-h(t)\bigr)\\
    =\;& a\delta_h(t)-b\delta_h(t+1)\\
    \ge\;& a\delta_h(t)-(a+12)\delta_h(t+1)\\
    =\;& a\bigl(\delta_h(t)-\delta_h(t+1)\bigr)-12\delta_h(t+1)\\
    \ge\;& 6x\bigl(\delta_h(t)-\delta_h(t+1)\bigr)-12\delta_h(t+1).
    \end{align*}
    Since $x$ was chosen such that $\delta_h(t)-\delta_h(t+1)>3/x$ and $\delta_h(t) \leq 1$, it follows that
    \[
    \operatorname{sc}_h(W)-\operatorname{sc}_h(W') > 6x\cdot\frac{3}{x}-12
    = 6 > 0.
    \]
    Thus, $W$ has a strictly larger $h$-Thiele score than every committee that differs from $W$ in the two non-universal candidates.
    Since every optimal committee contains all candidates in $C_{\mathrm{opt}}$, $W$ is the unique committee selected by $h$-Thiele.

    Nevertheless, uniformly selecting every candidate in $C_F$ with $\frac{2}{\lvert C_F \rvert}$ and all candidates in $C_{\text{opt}}$ with $1$ leads to a fractional committee in which every voter obtains utility 
    \begin{equation*}
        t-1 + \frac{2(30x+6)}{60x} = t + \frac{12}{60x}.
    \end{equation*} 
    This can be most easily seen by the fact that each voter approves a fraction of $\frac{3}{5}\frac{5x+1}{6x}  = \frac{1}{2} + \frac{1}{10x}$ of the candidates, while for $\{c_1, c_2\}$ this is only $1/2$. 

    Since the fractional committee assigns total weight $2$ uniformly across $C_F$, each voter has a utility strictly greater than $1$ for this fractional committee, whereas they have a utility of $1$ in the committee $\{c_1, c_2\}$. 
    Therefore, $W$ is not in $\wfPO$ and, thus, the $h$-Thiele rule violates $\wfPO$.
\end{proof}

\corrolaryCharacterization*
\begin{proof}
    Assume that $h$ satisfies $\delta_h(t+1) \geq \delta_h(t)$ for all $t\in \N_{>0}$ and is strictly monotone.
    Define the univariate function $f\colon \R_{\geq0}\to\R$ by
    \[
        f(x) = 
        \begin{cases}
            h(x) & \text{for } x\in\N,\\
            (1-\lambda)h(r) + \lambda h(r+1) & \text{for } x = r+\lambda,\ r\in\N,\ \lambda\in(0,1),
        \end{cases}
    \]
    i.e., the piecewise-linear interpolation of $h$ between consecutive integers, and let $g\colon\R^n_{\geq0}\to\R$ be given by $g_n(x)=\sum_{i\in N} f(x_i)$.

    We first establish that $f$ is convex and strictly increasing, and then lift both properties to $g$. The function $f$ is continuous and piecewise linear, with constant slope $h(r+1)-h(r)=\delta_h(r+1)$ on each interval $(r,r+1)$. Hence, its right derivative at any point $t\in\R_{\geq0}$ equals $\delta_h(\lfloor t\rfloor+1)$. Since $\delta_h(t)\le\delta_h(t+1)$ for all $t\in\N_{>0}$, this right derivative is non-decreasing in $t$, and a continuous piecewise-linear function with non-decreasing slope is convex. Thus, $f$ is convex.
    Moreover, strict monotonicity of $h$ yields $\delta_h(t)>0$ for all $t\in\N_{>0}$, so every slope is strictly positive and $f$ is strictly increasing.

    Observe that these properties transfer to $g$. As a sum of convex functions, each depending on a single coordinate, $g$ is convex. And if $x\le y$ coordinatewise with $x\neq y$, then $f(x_i)\le f(y_i)$ for every $i\in N$, with strict inequality in any coordinate where $x_i<y_i$. Summing gives $g_n(x)<g_n(y)$, and hence, $g$ is strictly monotone.

    Finally, since every voter's utility $u_i(W)$ is a natural number, we have $f(u_i(W))=h(u_i(W))$ and therefore $g_n(u(W))=\sum_{i\in N}h(u_i(W))$, which is exactly the $h$-Thiele score of $W$.
    Hence $f_g$ and $h$-Thiele select the same committees, and by \Cref{th:convexWelfarist} the strictly monotone convex function $g$ induces a rule satisfying $\fPO$ and consequently, so does $h$-Thiele.

    On the other hand, assume that $h$ is not strictly monotone. Then $h$-Thiele violates $\fPO$ which follows from \Cref{lemma:notMonotoneNotPO}.
    Assume that $h$ violates $\delta_h(t+1) \geq \delta_h(t)$ for some $t \in \N_{>0}$. Then, the violation follows from \Cref{theorem:Thiele}.
\end{proof}

\subsection{Restricted Domains} \label{appendix:restricted_domain}

\begin{definition}[\citet{pier2022corestable}]
    An approval instance is \emph{linearly consistent} (LC) if there exists a linear order $\sqsupset$ on $N \cup C$ such that for all voters $i \sqsupset j$ and candidates $c \sqsupset d$:
    \[
      c \in A_j \;\text{ and }\; d \in A_i \;\;\Longrightarrow\;\; c \in A_i.
    \]
\end{definition}

\thmLC*
Before giving the proof, we prove some intermediate steps.
\begin{lemma}\label{lem:lc-mono}
    Let $I=(A,k)$ be an LC instance with respect to the order $\sqsupset$, and let $c,d\in C$ with $c\sqsupset d$.
    Then, for every voter $i\in N_d\setminus N_c$ and every voter $j\in N_c$, it holds that $j\sqsupset i$.
\end{lemma}

\begin{proof}
    Assume for contradiction that there exist voters $i\in N_d\setminus N_c$ and $j\in N_c$ such that $i\sqsupseteq j$.
    Since $i\notin N_c$, while $j\in N_c$, we have $i\neq j$, and therefore $i\sqsupset j$.
    Now, using $c\sqsupset d$, $i\sqsupset j$, $j\in N_c$, and $i\in N_d$, the LC property implies $i\in N_c$.
    This contradicts $i\in N_d\setminus N_c$.
    Hence, every voter in $N_c$ must precede $i$ in the LC order.
\end{proof}
 
\begin{definition}
    For disjoint sets $S, U \subseteq C$ with $|S| = |U|$, we say $S$ \emph{covers} $U$ if $|S \cap A_i| \geq |U \cap A_i|$ for all $i \in N$ with one inequality being strict. A pair $(S, U)$ is a \emph{minimal covering pair} if (i) $S$ covers $U$, and (ii) no pair $(S', U')$ with $S' \subsetneq S$, $U' \subsetneq U$, $|S'| = |U'|$ exists such that $S'$ covers $U'$.
\end{definition}

\begin{restatable}{lemma}{lemLCCover}
\label{lem:lc-cover}
    Let $I = (A,k)$ be an LC instance with order $\sqsupset$, and let $(S, U)$ be a minimal covering pair. Let $c \in U$ be the $\sqsupset$-maximal candidate in $U$. Then, there exists a candidate $d \in S$ such that $N_c \subseteq N_d$.
\end{restatable}

\begin{proof}
    First, let us assume that $c$ is approved by at least one voter; otherwise, the statement holds trivially. 
    Thus, let $i^*$ be the $\sqsupset$-maximal voter in $N_c$. 
    Since $S$ covers $U$ and $c \in U \cap A_{i^*}$ holds, we have $S \cap A_{i^*} \neq \emptyset$. 
    Further, define $d$ to be the $\sqsupset$-maximal candidate in $S \cap A_{i^*}$.
    
    We claim that $N_c \subseteq N_d$ holds.
    Suppose, for the sake of contradiction, that there exists some voter $i \in N_c \setminus N_d$. 
    By $\sqsupset$-maximality of $i^*$ in $N_c$, we have $i \sqsubseteq i^*$ and as $i^* \in N_c \cap N_d$, we even get $i \sqsubset i^*$.
    Since $S$ covers $U$ and $i$ approves $c \in U$, we know that $i$ has to approve at least one candidate in $S$, say $e$. 
    As $i \notin N_d$, we have $e \neq d$. We now make a case distinction based on the order of $d$ and $e$.
     
    \medskip\noindent\textit{Case 1: $e \sqsupset d$.}\quad
    Since $d$ is the $\sqsupset$-maximal candidate in $S \cap A_{i^*}$, we have $i^* \notin N_e$. 
    Applying \Cref{lem:lc-mono} with the pair $e \sqsupset d$, by $i^* \in N_d$, $i^* \notin N_e$ and $i \in N_e$, we get that $i \sqsupset i^*$. 
    This contradicts $\sqsupset$-maximality of $i^*$ for voters in $N_c$.
     
    \medskip\noindent\textit{Case 2: $d \sqsupset e$.}\quad
    Since $i \in N_e$ and $i \notin N_d$, applying \Cref{lem:lc-mono}
    for any voter $j \in N_d$ yields $j \sqsupset i$.
    Set $S' \coloneqq S \setminus \{d\}$ and $U' \coloneqq U \setminus \{c\}$. In the following, we show that $S'$ covers $U'$, contradicting the minimality of $(S, U)$.
    For any voter $j$ with $d\notin A_j$ or $c\in A_j$, deleting $c$ from $U$ reduces $\lvert U\cap A_j\rvert$ by at least as much as deleting $d$ from $S$ reduces $\lvert S\cap A_j\rvert$. 
    Since $S$ covers $U$, it follows that
    \[
        \lvert U'\cap A_j\rvert \le \lvert S'\cap A_j\rvert.
    \]

    Hence, it remains to consider the voters who approve $d$, but not $c$. 
    Let $j$ be such a voter.
    We claim that $j \notin N_b$ for every $b \in U'$. 
    Since $c$ is the $\sqsupset$-maximal candidate in $U$, it suffices to consider $b \in U'$ with $b \sqsubset c$. 
    Assume for the sake of contradiction that $j \in N_b$ for some candidate $b$.
    We apply the LC condition with voters $j \sqsupset i$ and 
    candidates $c \sqsupset b$. 
    Since $c \in A_i$ and $b \in A_j$, it follows $c \in A_j$, contradicting our previous assumption that $c \notin A_j$.  
    Thus, no such voter approves any candidate in $U'$, yielding $\lvert U' \cap A_j \rvert \leq \lvert S' \cap A_j \rvert$. 
    
    Finally, we claim that voter $i$ approves strictly more candidates in $S'$ than in $U'$. 
    Since $S$ covers $U$, the inequality $\lvert U \cap A_i \lvert \leq \lvert S \cap A_i \rvert$ holds trivially. However, since $i \in N_c \setminus N_d$, we derive
    \begin{equation*}
        \lvert U' \cap A_i \rvert = \lvert U \cap A_i \rvert -1 < \lvert S \cap A_i \rvert = \lvert S' \cap A_i \rvert,
    \end{equation*} 
    showing that $i$ gets strictly higher approval score from $S'$ than from $U'$.
    This contradicts the minimality of the covering pair $(S, U)$. 
    
    As both cases lead to a contradiction, no voter $i\in N_c\setminus N_d$ exists. Hence $N_c\subseteq N_d$.
\end{proof}

Using \Cref{lem:lc-cover}, we can finally prove \Cref{thm:lc-sdo}.
\begin{proof}
    We first prove that every LC instance satisfies the SDO property.
    Assume that we have an LC instance $I$ together with two committees $W$ and $W'$ such that $W'$ Pareto-dominates $W$. 
    Let $S = W' \setminus W$ and $U = W \setminus W'$. 
    Consider a minimal covering pair $(S_0, U_0)$ which exists since $S$ covers $U$.  
    By \Cref{lem:lc-cover}, there exist candidates $c \in U_0$ and $d \in S_0$ with $N_c \subseteq N_d$. 
    If $N_c=N_d$, then $(S_0 \setminus \{d\}, U_0 \setminus \{c\})$ is still a cover, contradicting the minimality of $(S_0, U_0)$. 
    Hence, there exists a candidate $c \in W$ that is Pareto-dominated by a candidate $d \in C \setminus W$.
    The other direction of the SDO equivalence is trivial since if there exists a candidate $c \in W$ that is Pareto-dominated by $d \in C \setminus W$, then $W$ is Pareto-dominated by $(W\setminus \{c\}) \cup \{d\}$.

    It remains to show that the domain is closed under cloning. 
    Let $I$ be an instance satisfying LC with respect to the linear order $\sqsupset$, and fix $r\in\N_{>0}$. 
    Consider the $r$-cloned instance $rI$, and define the order $\sqsupset_r$ on the cloned candidates by placing the clones of each candidate consecutively as follows
    \[
        c^s \sqsupset_r d^t
        \quad\Longleftrightarrow\quad
        c\sqsupset d
        \text{ or }
        \bigl(c=d \text{ and } s<t\bigr).
    \]
    In particular, we have $i\sqsupset_r j$ if and only if $i \sqsupset j$ for all voters $i \neq j$, leaving the voter order unchanged.
    We show that $rI$ satisfies LC with respect to $\sqsupset_r$.
    Let $i,j\in N$ be voters and let $a,b$ be cloned candidates such that
    \[
        i\sqsupset_r j,
        \qquad
        a\sqsupset_r b,
        \qquad
        a\in A_j^{rI},
        \qquad
        b\in A_i^{rI}.
    \]
    We need to show that $a\in A_i^{rI}$.
    Write $a=c^s$ and $b=d^t$ for original candidates $c,d\in C$ and clone indices $s,t\in[r]$.
    
    If $c\neq d$, then $a\sqsupset_r b$ implies $c\sqsupset d$. 
    Moreover, $a=c^s\in A_j^{rI}$ implies $c\in A_j$ in the original instance, and $b=d^t\in A_i^{rI}$ implies $d\in A_i$. 
    Since $I$ satisfies LC, we obtain $c\in A_i$. 
    Thus, voter $i$ approves every clone of $c$ in $rI$, and in particular they approve $a\in A_i^{rI}$.
    
    If $c=d$, then $a$ and $b$ are clones of the same original candidate. 
    All clones have identical approval sets in $rI$, so $ N_a = N_b$. 
    Since $b\in A_i^{rI}$, it follows that $a\in A_i^{rI}$.
    
    Hence, $rI$ satisfies LC with respect to $\sqsupset_r$, and we may conclude that the LC domain is closed under cloning.
\end{proof}

Since the VCI domain can be interpreted as a one-dimensional Euclidean domain, it is natural to ask whether the result extends to higher dimensions. We close the section by showing that it already fails to extend in two dimensions.

\begin{definition}
    An approval instance $(A, k)$ is \emph{2D-Euclidean} if there exists a closed disk $D_i \subseteq \mathbb{R}^2$ for each voter $i$, and a closed disk $D_c \subseteq \mathbb{R}^2$ for each candidate $c$, respectively, such that
    \[
      c \in A_i \;\iff\; D_i \cap D_c \neq \emptyset.
    \]
\end{definition}

Thus, approvals are induced by geometric overlap: voter $i$ approves candidate $c$ precisely when their corresponding disks intersect.
Since $\fPO$ and $\PO$ already fail to coincide on the instance in \Cref{figure:counter_example_PO_fPO_small}, it suffices to embed that instance in the plane (see \Cref{fig:circles}).\footnote{The embedding was constructed with the help of ChatGPT 5.5 and Claude Opus 4.7.}
This gives the following result:

\begin{restatable}{proposition}{propTwoEuc}
\label{prop:2d-not-sdo}
    In the 2D-Euclidean domain, $\fPO \subsetneq \PO$. 
\end{restatable}

\begin{proof}
  Consider the instance $I$ depicted in~\Cref{figure:counter_example_PO_fPO_small}.
  We have argued that the committee $W = \{c_1, c_2\}$ satisfies $\PO$ but not $\fPO$ for this instance. It therefore suffices to show that $I$ is realizable in the 2D-Euclidean domain, that is, each voter and each candidate can be represented by a closed disk in $\R^2$ such that voter $i$ approves candidate $c$ if and only if their two disks intersect. One can verify that the assignment in \Cref{fig:circles} constitutes such a representation.
\end{proof}

\begin{figure}[t]
\begin{minipage}[t]{0.59\textwidth}
\vspace{0pt}
    \centering
    \begin{tikzpicture}[scale=0.5,
        cand/.style ={draw=blue!70!black, fill=blue!15, fill opacity=0.30, thick},
        voter/.style={draw=red!70!black, dashed, thick},
        isect/.style={fill=orange, fill opacity=0.55},   
        clab/.style ={blue!60!black, font=\small},
        vlab/.style ={red!60!black,  font=\small},
        axis/.style={->, thick},
        dot/.style  ={circle, fill=black, inner sep=1pt}]
    
      \coordinate (c1) at ( 1.517, 3.922);
      \coordinate (c2) at (-0.513,-3.703);
      \coordinate (c3) at (-0.008,-0.585);
      \coordinate (c4) at ( 3.124, 0.947);
      \coordinate (c5) at (-3.146, 1.334);
      \coordinate (v1) at ( 1.613, 1.375);
      \coordinate (v2) at (-0.782, 2.185);
      \coordinate (v3) at ( 0.021, 3.270);
      \coordinate (v4) at (-3.674,-2.071);
      \coordinate (v5) at ( 2.941,-2.698);
      \coordinate (v6) at ( 1.155,-2.638);
      \coordinate (v7) at (-2.248,-1.339);
    
    \draw[axis] (-6,0) -- (6,0) node[below] {$x$};
    \draw[axis] (0,-6) -- (0,6) node[left] {$y$};
    
    \foreach \x in {-5,-4,-3,-2,-1,1,2,3,4,5}
        \draw (\x,0.15) -- (\x,-0.15) node[below=2pt] {\tiny \x};
    \foreach \y in {-5,-4,-3,-2,-1,1,2,3,4,5}
        \draw (0.15,\y) -- (-0.15,\y) node[left=2pt] {\tiny \y};
    
      \draw[cand] (c1) circle (1.521);
      \draw[cand] (c2) circle (2.084);
      \draw[cand] (c3) circle (1.515);
      \draw[cand] (c4) circle (2.137);
      \draw[cand] (c5) circle (1.972);
    
      \draw[voter] (v1) circle (1.328);
      \draw[voter] (v2) circle (1.660);
      \draw[voter] (v3) circle (2.040);
      \draw[voter] (v4) circle (1.774);
      \draw[voter] (v5) circle (1.813);
      \draw[voter] (v6) circle (1.145);
      \draw[voter] (v7) circle (1.148);
    
      \foreach \a/\ra/\b/\rb in {%
        v1/1.328/c1/1.521, v1/1.328/c3/1.515, v1/1.328/c4/2.137,
        v2/1.660/c1/1.521, v2/1.660/c3/1.515, v2/1.660/c5/1.972,
        v3/2.040/c1/1.521, v3/2.040/c4/2.137, v3/2.040/c5/1.972,
        v4/1.774/c2/2.084, v4/1.774/c5/1.972,
        v5/1.813/c2/2.084, v5/1.813/c4/2.137,
        v6/1.145/c2/2.084, v6/1.145/c3/1.515,
        v7/1.148/c2/2.084, v7/1.148/c3/1.515, v7/1.148/c5/1.972}{%
        \begin{scope}
          \clip (\a) circle (\ra);
          \fill[isect] (\b) circle (\rb);
        \end{scope}}
    
      \foreach \n in {c1,c2,c3,c4,c5,v1,v2,v3,v4,v5,v6,v7} \node[dot] at (\n) {};
    
      \node[clab] at (c1) [above right] {$c_1$};
      \node[clab] at (c2) [left]       {$c_2$};
      \node[clab] at (c3) [above right] {$c_3$};
      \node[clab] at (c4) [right]       {$c_4$};
      \node[clab] at (c5) [left]        {$c_5$};
    
      \node[vlab] at (v1) [above right] {$1$};
      \node[vlab] at (v2) [above left]  {$2$};
      \node[vlab] at (v3) [above right] {$3$};
      \node[vlab] at (v4) [below left]  {$4$};
      \node[vlab] at (v5) [below right] {$5$};
      \node[vlab] at (v6) [below]       {$6$};
      \node[vlab] at (v7) [below left]  {$7$};
    
      \begin{scope}[shift={(4.6,4.4)}]
        \draw[cand]  (0,0) circle (0.28);
        \node[anchor=west, font=\small] at (0.5,0) {candidate};
        \draw[voter] (0,-0.8) circle (0.28);
        \node[anchor=west, font=\small] at (0.5,-0.8) {voter};
        \fill[isect] (-0.28,-1.6) rectangle (0.28,-1.32);
        \node[anchor=west, font=\small] at (0.5,-1.46) {approval (overlap)};
      \end{scope}
    \end{tikzpicture}
\end{minipage}%
\hfill
\raisebox{-0.5cm}{%
\begin{minipage}[t]{0.4\textwidth}
\vspace{0pt}
    \centering
    \begin{tabular}{c rrr}
        \toprule
        candidate/voter & $x$ & $y$ & $r$ \\
        \midrule
        $c_1$ & $1.517$  & $3.922$  & $1.521$ \\
        $c_2$ & $-0.513$ & $-3.703$ & $2.084$ \\
        $c_3$ & $-0.008$ & $-0.585$ & $1.515$ \\
        $c_4$ & $3.124$  & $0.947$  & $2.137$ \\
        $c_5$ & $-3.146$ & $1.334$  & $1.972$ \\
        \midrule
        $1$ & $1.613$  & $1.375$  & $1.328$ \\
        $2$ & $-0.782$ & $2.185$  & $1.660$ \\
        $3$ & $0.021$  & $3.270$  & $2.040$ \\
        $4$ & $-3.674$ & $-2.071$ & $1.774$ \\
        $5$ & $2.941$  & $-2.698$ & $1.813$ \\
        $6$ & $1.155$  & $-2.638$ & $1.145$ \\
        $7$ & $-2.248$ & $-1.339$ & $1.148$ \\
        \bottomrule
    \end{tabular}
\end{minipage}}

\caption{A visualization of a $2$D-Euclidean instance in which $\fPO$ and $\PO$ do not coincide. The coordinates of the candidates and voters for the embedding, along with their corresponding disk radii, are listed in the table on the right-hand side. }
\label{fig:circles}
\end{figure}
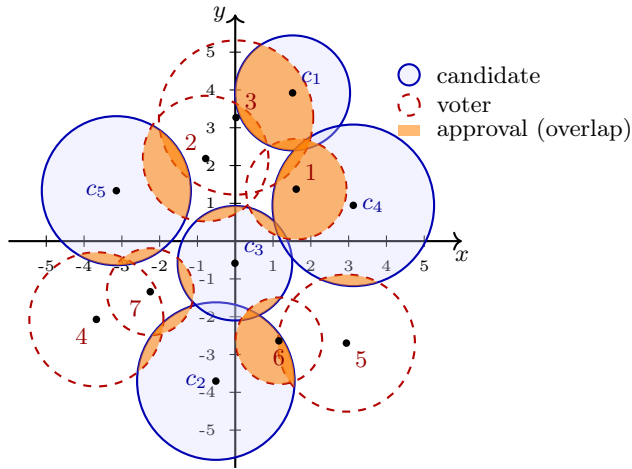
\end{document}